%% file: sample-acmsmall.tex
\documentclass[acmsmall]{acmart}
\input{misc/preamble.tex}

\AtBeginDocument{%
  }

\setcopyright{rightsretained} 
\acmJournal{POMACS}
\acmYear{2023} \acmVolume{7} \acmNumber{1} \acmArticle{26} \acmMonth{3} \acmPrice{}\acmDOI{10.1145/3579452}

\acmJournal{JACM}
\acmVolume{37}
\acmNumber{4}
\acmArticle{111}
\acmMonth{8}

\setcopyright{rightsretained} 
\acmJournal{POMACS}
\acmYear{2023} \acmVolume{7} \acmNumber{1} \acmArticle{26} \acmMonth{3} \acmPrice{}\acmDOI{10.1145/3579452}

\begin{document}
\title[Online Adversarial Stabilization of Unknown Networked Systems]{Online Adversarial Stabilization of \\ Unknown Networked Systems}



\author{Jing Yu}
\affiliation{%
  \institution{California Institute of Technology}
  \city{Pasadena}
  \country{United States}}
\email{jing@caltech.edu}

\author{Dimitar Ho}
\affiliation{%
  \institution{California Institute of Technology}
  \city{Pasadena}
  \country{United States}}
\email{dho@caltech.edu}

\author{Adam Wierman}
\affiliation{%
  \institution{California Institute of Technology}
  \city{Pasadena}
  \country{United States}}
\email{adamw@caltech.edu}







\begin{abstract}
  We investigate the problem of stabilizing an unknown networked linear system under communication constraints and adversarial disturbances. We propose the first provably stabilizing algorithm for the problem. The algorithm uses a distributed version of nested convex body chasing to maintain a consistent estimate of the network dynamics and applies system level synthesis to determine a distributed controller based on this estimated model. Our approach avoids the need for system identification and accommodates a broad class of communication delay while being fully distributed and scaling favorably with the number of subsystems.  
\end{abstract}



\received{August 2022}
\received[revised]{October 2022}
\received[accepted]{January 2023}

\maketitle

\section{Introduction}
Large-scale networked dynamical systems play a crucial role in many emerging engineering systems such as the power grid \cite{fang2011smart}, autonomous vehicles \cite{li2015overview}, and swarm robots \cite{morgan2014model}. Motivated by the success of learning-based control methods for single-agent (centralized) linear systems, there has been growing interest in learning distributed controllers for unknown networked systems composed of interconnected and spatially distributed linear time-invariant (LTI) subsystems \cite{bu2019lqr, fattahi2020efficient, furieri2020learning, ye2021sample, li2021distributed}.

However, since most existing literature ports centralized learning-based control techniques over to the distributed setting, almost all previous work assumes that the underlying dynamics are stable, or that a stabilizing and distributed controller is known. For a large-scale networked system, such assumptions are often unrealistic, because designing stabilizing distributed controllers itself is a significant task even if the dynamics model is available \cite{rotkowitz2005characterization,   han2003lmi, wang2014localized,fardad2014design, anderson2019system, zheng2020equivalence}. 

Recent work has begun to lift the assumption of the knowledge of a stabilizing controller in the centralized case, e.g. \cite{chen2021black, hu2022sample,  simchowitz2018learning}. This line of work follows the approach of system identification, either by letting the unstable system run open-loop or by exciting the system via control inputs. 
However, such approaches induce explosive transient behaviors due to the instability of the underlying system. 
Without proper generalization to the networked setting, such explosive behavior can cause catastrophic system degradation before a proper stabilizing controller can be learned. 

Further, until now, scalability and information constraints have only been considered separately in learning-based distributed controller design; no general approach exists. On the other hand, information constraints and scalability have been the central topics in distributed control for the past decade due to their theoretical challenge and practical importance \cite{rotkowitz2008information, zheng2017scalable,sturz2020distributed, matni2016regularization, wang2018separable, sturz2020distributed}. Therefore, it is crucial to simultaneously consider such constraints when designing learning-based distributed control algorithms for networked systems. 

\vspace{-0.05in}
\subsection{Contributions} In this work, we overcome the aforementioned challenges by leveraging recent advances in online learning and distributed control. In particular, we propose an approach that combines a distributed version of nested convex body chasing (NCBC), in order to maintain a consistent estimate of the network dynamics, with system level synthesis (SLS), in order to determine a distributed controller based on the selected consistent model. 
This combination yields the first online algorithm that provably stabilizes a networked LTI system with information constraints under adversarial disturbances (Theorem \ref{thrm:main}). The proposed algorithm (Algorithm \ref{alg:main}) is distributed and scales favorably to the number of subsystems in the network.

The proposed approach in this paper is fundamentally different than traditional system identification based methods, which incur prohibitively large state norm under adversarial disturbances, even in the simplest setting (see \Cref{fig:sysid_main}). The reason is that system identification-based approaches seek to learn the full system dynamics, which requires full excitation of the system against worst-case disturbances. On the other hand, our approach does not require precise knowledge of the system. Instead, we maintain model estimates that are consistent with the observations generated by the unknown system at all times.
A consequence of focusing on consistency is a natural endogenous exploration-exploitation scheme where our algorithm performs well (small state norm) while the selected model stays consistent, and gains information about the system whenever it observes a large state norm that renders the selected model inconsistent.

The main result of this paper is an input-to-state stability guarantee (\Cref{thrm:main}), where we draw novel connections between the path length property of NCBC techniques and system stability analysis. This follows from a set of novel technical results for SLS in the learning-based control context.
In particular, we generalize a previous result \cite{anderson2019system} on the characterization of the closed loop under SLS controllers that are synthesized from an arbitrary and potentially incorrect system model (\Cref{lem:closed-loop}). This result enables the analysis of our algorithm when each  subsystem uses local, asynchronous, and wrong model information for local controller synthesis.
Further, we derive a novel perturbation result with explicit constants for finite-horizon SLS synthesis (\Cref{thrm:sensitivity}) that globally bounds the sensitivity of the optimal solution to the SLS problem (a quadratic program with equality and sparsity constraints) with respect to the model. This result is also applicable in other contexts such as a class of MPC problems studied in \cite{borrelli2017predictive, alonso2021data, sieber2021system}. 

\begin{table}
  \caption{Maximum and top $90 \%$  infinity norm of the state ($\infnorm{x(t)}$) for different disturbance profiles averaged over 10 runs. Simulation details are provided in \Cref{sec:simulation}.}
    \vspace{-0.1in}
  \label{fig:sysid_main}
  {\small \begin{tabular}{cccl}
    \toprule
    Algorithm & Correlated Gaussian (Top $90\%$) & Uniform (Top $90\%$)  & State-dependent (Top $90\%$)\\
    \midrule
    \textbf{This work} & $\mathbf{1.21 \times 10^1}$ ($0.31 \times 10^1$) & $\mathbf{2.30 \times 10^1}$ ($0.36 \times 10^1$) & $\mathbf{7.14 \times 10^1}$ ($0.54 \times 10^1$) \\
    SysID & $5.12 \times 10^{11}$ ($1.71 \times 10^{11}$) & $5.12 \times 10^{11}$ ($1.71 \times 10^{11}$) & $5.12 \times 10^{11}$ ($1.71 \times 10^{11}$)\\
  \bottomrule
  \vspace{-0.3in}
\end{tabular}}
\end{table}

\subsection{Related work} This work contributes to a large and growing body of work on the topics related to learning-based control design, online control, and distributed control. We briefly review the literature most related to this work below. 

\textbf{Stabilization of unknown systems. } Stabilizing unknown linear systems has long been a fundamental problem studied in adaptive control theory \cite{ioannou2006adaptive}. It recently reemerged as a learning problem and received considerable attention from the machine learning community \cite{perdomo2021stabilizing, zhao2021learning, treven2021learning, hu2022sample}. Most works have been developed under single-agent setting, with a no-noise assumption \cite{lamperski2020computing, talebi2021regularizability} or Gausssian noise models \cite{faradonbeh2018finite, lale2022reinforcement}. Under the adversarial noise setting, which is the focus of this paper, the only work that guarantees stabilization for LTI systems is \cite{chen2021black}, with a system identification-based approach that achieves order-optimal regret. In contrast, we propose a novel framework for stabilization under adversarial noise that does not rely on accurate identification of the true dynamics. In particular, our method is the first algorithm to stabilize a networked LTI system under adversarial disturbances with information constraints while simultaneously achieving magnitudes of improvement in empirical performance over the state-of-the-art identification-based approach \cite{chen2021black} in the single-agent setting, despite the regret-optimal guarantee in \cite{chen2021black}. 

\textbf{Distributed control. }
Motivated by large-scale cyberphysical systems that are composed of physically distributed subsystems with local dynamical interactions, there is a large body of work on control design for networked systems \cite{zheng2020equivalence, anderson2019system, kashyap2019explicit}. Cyberphysical systems such as the power grid are commonly constrained by a communication layer that allows specific structure of information exchange among the subsystems. such information structure imposes significant challenges for optimal control design, often rendering the problem NP-hard \cite{tsitsiklis1985complexity}. In \cite{rotkowitz2005characterization}, it was shown that a large class of practically relevant distributed control problems is convex and tractable to solve. Since then, many works have focused on this class of problems \cite{lamperski2015optimal, fardad2014design}. However, \cite{wang2019system} observes that the complexity of computation and implementation of distributed controllers developed under this setting can be prohibitively expensive, thus not scalable to large-scale systems. The System Level Synthesis (SLS) framework is developed as a scalable alternative to distributed control design \cite{anderson2019system}. In particular, SLS allows order-constant complexity for synthesis and implementation, due to its special parameterization and implementation of the feedback controller. As a result, many works have adopted SLS as the basis for novel (learning-based) control algorithms in both distributed and centralized setting \cite{dean2020sample, didier2022system, alonso2021data, umenberger2020optimistic}. We contribute to the literature on SLS by developing a suit of technical results for SLS controllers that can find applications beyond the setting of this work.

\textbf{Learning distributed controllers. } 
Many learning-based control algorithms for networked systems adopt a centralized learning or computational approach with the objective of regret minimization, e.g., \cite{fattahi2020efficient, bu2019lqr,ye2021sample, faradonbeh2022joint, furieri2020learning}. All prior work use the stochastic noise or no-noise model and assume a known stabilizing distributed controller is given \cite{li2021distributed, alonso2021data, jing2021learning, alemzadeh2019distributed, talebi2021distributed, alemzadeh2021d3pi}. As far as we are aware, no previous work accommodates communication delay while doing both learning and control. 
The most related to our work are \cite{adaptivesls} and \cite{fattahi2020efficient}, where learning-based SLS controllers are designed to control unknown networked systems. Both of the methods require the knowledge of a stabilizing and distributed controller. \cite{adaptivesls} is only applicable to small-uncertainty scenarios, while \cite{fattahi2020efficient} requires a stabilizing distributed controller and performs centralized learning. In this work, we focus on stabilization and propose the first distributed learning-based control algorithm that guarantees stability for unknown networked systems under adversarial disturbances. 

\textbf{Online learning. }
The problem of online stabilization for unknown dynamical systems is an instance of online decision making problems, where an agent makes a sequence of decisions based on the feedback from an unknown environment with the goal of cost minimization. Online decision making is studied extensively in the online learning literature, with a line of work \cite{goel2019online, shi2020online, yeh2022robust, lin2022online} that makes interesting connections between convex function and body chasing \cite{antoniadis2016chasing, argue2022chasing} and linear control theory. In particular, \cite{ho2021online} proposes an online nonlinear robust control method based on convex body chasing that guarantees finite mistakes under adversarial disturbances without the need for system identification. While \cite{ho2021online} considers binary cost functions, we present novel technical results that establish the first connection between convex body chasing and stability analysis for both single-agent and networked linear dynamical systems.


\subsection{Notation}
Let $\|\cdot\|$ be the $\ell_2$ norm and $\|\cdot\|_F$ be the Frobenius norm.  We denote the $(i,j)$th position of a matrix $M$ as $M(i,j)$ and use $M(:,j), M(i,:)$ for the $j$th column and $i$th row of $M$ respectively. We use $[N]$ for the set of positive integers up to $N$. Positive integers are denoted as $\mathbb{N}_+$.
Bold face lower cases are reserved for vector signal of the form $\mathbf{x} := [x(0)^T,x(1)^T, \dots]^T$ with $x(t) \in \mathbb{R}^n$ is an infinite sequence of vectors indexed by time $t$. 
We reserve bold face capital letters for causal linear operators/transfer matrices with components $K[0], K[1],\ldots,$ such that
\begin{equation*}
    \mathbf{K} := \begin{bmatrix} K[0] & 0 & \dots & \\ K[1] & K[0] & 0 & \dots &  \\ \vdots & \ddots  & \ddots & \ddots \end{bmatrix}.
\end{equation*}
We write $\mathbf{y} = \mathbf{G}\mathbf{x}$ to mean that $y(t) = \sum_{k=0}^t G[k] x(t-k)$.
Given any binary matrix $\mathcal{C} \in \{1,0\}^{N\times N}$, we say $M \in \mathcal{C}$ for a matrix $M\in \mathbb{R}^{N \times N}$ if the sparsity of $M$ is $\mathcal{C}$. We use $\{e_j\}_{j=1}^n$ for the standard basis in $\mathbb{R}^n$.

\section{Preliminaries and problem setup}
\label{sec:dynamics}
We consider the task of stabilizing an unknown networked system made up of $N$ interconnected, heterogeneous linear time-invariant (LTI) subsystems, illustrated in Figure \ref{fig:example}. 
For each subsystem $i \in [N]$, let $x^i(t) \in \mathbb{R}^{n_i}$, $u^i(t) \in \mathbb{R}^{m_i}$, $w^i(t) \in \mathbb{R}^{n_i}$ be the local state, control, and disturbance vectors respectively. Each subsystem $i$ has dynamics,
\begin{equation}
        \label{eq:local_sys}
        x^i(t+1) = \sum_{j\in \N(i)} \left( A^{ij}x^j(t) +  B^{ij}u^j(t) \right) + w^i(t),
\end{equation}
where we write $j \in \N(i)$ if the states or control actions of subsystem $j$ affect those of subsystem $i$ through the open-loop network dynamics ($i \in \N(i)$).
Concatenating all the subsystem dynamics, we can represent the global dynamics as
\begin{equation}
    \label{eq:global_sys}
    x(t+1) = A x(t) + B u(t) + w(t),
\end{equation}
where $x(t) \in\mathbb{R}^{n_x}$, $u(t) \in \mathbb{R}^{n_u}$, $w(t) \in \mathbb{R}^{n_x}$, with $n_x = \sum_{i=1}^N n_i$ and $n_u = \sum_{i=1}^N m_i$, and we define $A^{ij},\, B^{ij} \equiv 0$ for all $j\not \in \N(i)$. 
The networked LTI model \eqref{eq:local_sys} has been extensively studied in the networked control literature for various applications such as robotic swarms \cite{mukherjee2022reinforcement}, voltage control for the distribution network of the power grid \cite{yeh2022robust}, and many other large-scale cyber-physical systems \cite{lemos2012distributed,zhang2016controllability}. An example is the linearized swing equation for power systems, where the global system is composed of a mesh of interacting buses \cite{gholami2020fast, wang2016localized}. In this setting, the states $x^i$ of each bus $i$ is two-dimensional and corresponds to the phase angle relative to some given setpoint and the associated frequency. The input $u^i$ at bus $i$ is the controllable load, while $w^i$ is the bounded load disturbances that are often correlated in space and time. 

\begin{figure}
\centering     
\subfigure[System with communication graph $\mathcal{G}^C$. ]{\label{fig:example}\includegraphics[scale = 0.35]{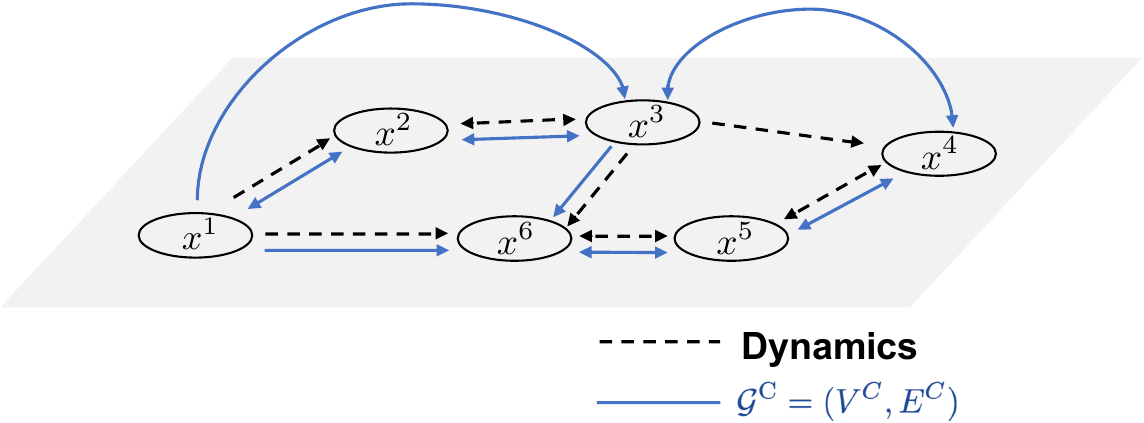}}\hfill
\subfigure[$A$, $B$ matrix with parameter $\Theta$.]{\label{fig:AB}\includegraphics[scale = 0.25]{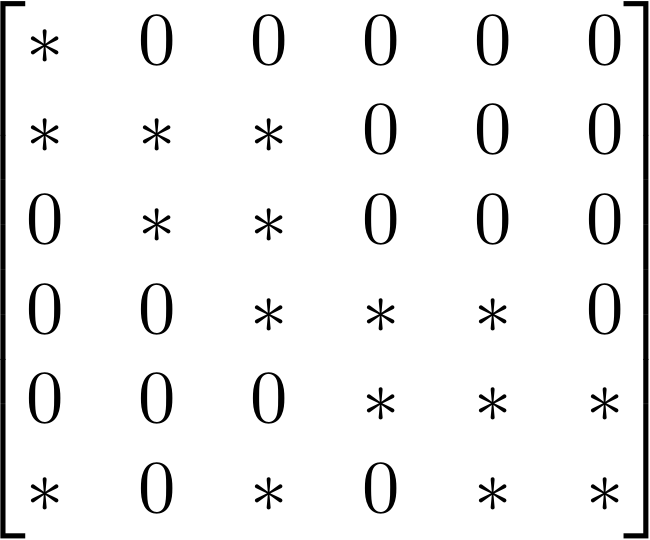}}
\hfill
\subfigure[Adjacency matrix $\mathcal{C}$.]{\label{fig:C}\includegraphics[scale = 0.25]{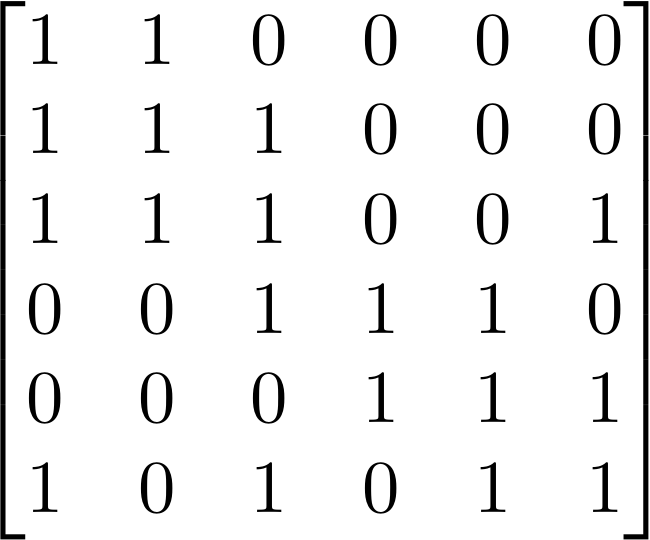}}
\caption{Example networked LTI system with information constraints.}
\end{figure}

We assume that the topology among the subsystems is known, \textit{i.e.}, the sets $\N(i)$ for $i \in [N]$ are known.
However, the parameters of the dynamics (entries of matrices $A^{ij}$, $B^{ij}$) are unknown. 
Let $\theta^i$ denote the unknown \emph{local parameter} for subsystem $i$, i.e., $\theta^i:= \left( A^{ij},\, B^{ij}\right)_{j\in \N(i)}$. Further, let $\Theta:= \left(\theta_1,\ldots,\theta_N\right)$ be the \emph{global parameter}. 
We write $A(\Theta)$ and $B(\Theta)$  (equivalently $A^{ij}(\theta^i)$, $B^{ij}(\theta^i)$) to emphasize that $A$  and $B$ are matrices constructed with appropriate zeros according to the network topology (known), and the nonzero entries specified by $\Theta$ (unknown). 
\begin{example}
    \label{ex:model}
    Consider the networked system in Figure \ref{fig:example} where each subsystem $i\in [6]$ has $x^i(t) \in \mathbb{R}$ and $u^i(t) \in \mathbb{R}$. For each $i$, the set $\mathcal{N}(i)$ contains the subsystems that has a dashed arrow pointing towards $x^i$ in the figure. For example, $\N(6) = \{1,\, 3,\, 5, \, 6\}$. 
    Each $A^{ij}$ and $B^{ij}$ for $j \in \N(i)$ is a scalar.
    The stacked global dynamics has matrix $A$ and $B$ with structure shown in Figure \ref{fig:AB}.
    The unknown local parameter $\theta^i$ corresponds to the $*$ entries of the $i^{\text{th}}$ row of $A$ and $B$, while the global parameter $\Theta$ is a vector containing $*$ entries in matrix $A$ and $B$.
\end{example}

We now introduce three core assumptions needed for our algorithm and analysis.  As we highlight below, these are standard assumptions in the learning-based control literature. 
\begin{assumption}[Adversarial disturbances]
    \label{assump:noise}
   $\left\|w(t)\right\|_{\infty} \leq W $ for \eqref{eq:global_sys}. 
\end{assumption}

\begin{assumption}[Compact Parameter Set]
    \label{assump:compact}
    The network structure $\N(i)$ for $i \in [N]$ is known. The true system parameter $\Theta^\star := \left(\theta^{1,\star},\ldots,\theta^{N,\star}\right)$ is an element of a known compact convex set $\mathcal{P}_0 = \mathcal{P}^1_0 \times \dots \times \mathcal{P}^{N}_0 $, which is a product space of local parameter sets where $\theta^{i,*} \in \mathcal{P}_0^i$. 
    The known parameter set is bounded such that there exists a known constant $\kappa >0$ where $\norm{[A\left(\Theta\right) \, B\left(\Theta\right)]}_F \leq \kappa$ for all $\Theta \in \mathcal{P}_0$.
\end{assumption}

\begin{assumption}[Controllability]
    \label{assump:controllable}
    For all $\Theta \in \mathcal{P}_0$, $(A(\Theta),B(\Theta))$ is controllable. 
\end{assumption}

Bounded adversarial disturbances is a common model in the adversarial online learning and control problems \cite{agarwal2019online, hazan2020nonstochastic, didier2022system}. Since we make no assumptions on how large the bound on the disturbance $W$ is, Assumption \ref{assump:noise} models a variety of disturbance models, such as bounded and correlated stochastic noise or state-dependent disturbances such as the linearization and discretization error for nonlinear continuous dynamics \cite{tu2019sample}. {Moreover, the known bound $W$ can be relaxed to an unknown parameter $\eta$ with $\eta \leq W$ for a known constant $W$ to reduce conservatism for large $W$. } 
Assumptions \ref{assump:compact} and \ref{assump:controllable} are standard in the learning-based control literature, e.g., see  \cite{cohen2019learning, agarwal2019online}. {We impose controllability in Assumption \ref{assump:controllable} for ease of exposition but it can be relaxed to stabilizability by adjusting the choice of model-based controller to an infinite-horizon controller such as the one proposed in \cite{yu2021localized} for the algorithm. }

\subsection{Stability} 
One of the fundamental goals for control design is to ensure stability. In this paper,
we aim to learn a stabilizing controller for the networked linear system \eqref{eq:global_sys} in the sense of input to state stability (ISS) \cite{sontag2008input}.
ISS is one of the main notions of stability for both linear and nonlinear systems \cite{jiang2001input, aswani2013provably}. Here we adapt the ISS definition to the $\ell^\infty$-norm.
\begin{definition}[ISS]
    A dynamical system of the form \eqref{eq:global_sys} is said to be input to state stable (ISS) if
there exist functions $\beta:\mathbb{R}_+ \times \mathbb{N} \rightarrow \mathbb{R}_+$ that is continuous, strictly increasing, and bijective with respect to the second argument with $\lim_{t\rightarrow \infty}\beta(a, t) = 0$ for all $a \geq 0$, $t\in \mathbb{N}$, and $\gamma : \mathbb{R}_+ \rightarrow \mathbb{R}_+$ that is continuous, strictly increasing, and bijective such that for all initial state $x(0)$, disturbance sequence $\mathbf{w}$, and time $t\geq t_0$ for $t_0 \in \mathbb{N}_+$, we have 
$\infnorm{x(t)} \leq \beta(\infnorm{x(t_0)}, t-t_0) + \gamma(\sup_{t \geq t_0} \infnorm{w(t)})$.
\end{definition}




\subsection{Distributed design and information constraints}
\label{sec:info_constraints}

For large-scale networks such as the power grid with state dimension in the orders of thousands to millions, it is unrealistic and prohibitively costly for a central agent to learn a global policy online. A promising remedy is to decompose the global policy learning into a \emph{local} one, where each subsystem in the network learns a local policy in a distributed fashion.
In this work, we propose a distributed learning-based control algorithm for the networked linear system \eqref{eq:global_sys} that guarantees stability of the global system. 

In addition to distributed design, networks of the form \eqref{eq:local_sys} are often modelled with additional information constraints that require careful consideration. In this work we consider two common information constraints. The first is \emph{communication delay}, where the dynamical system is endowed with a communication network that specify delayed information transmission among subsystems. The second is \emph{local information}, where each subsystem only computes with (delayed) local information within a specified neighborhood, and discard information outside of the neighborhood. 
We come back to these information constraints and present definitions in \Cref{sec:comm_constraints,sec:localized_control}.

\subsection{Algorithm preliminaries}

Our proposed algorithm makes use of two emerging techniques, one from the learning community, i.e., nested convex body chasing (NCBC), and one from the control community, i.e., system level synthesis (SLS). We provide important background on each below before introducing our algorithm in the next section.

\subsubsection{Preliminaries on NCBC}
\label{sec:NCBC}
The Nested Convex Body Chasing (NCBC) problem is a well-studied online learning problem \cite{bubeck2020chasing,argue2021chasing}. 
At every round $t$, the player is presented a convex body $\mathcal{K}_t \subset \mathbb{R}^n$ which is nested in the previous body, e.g., $\mathcal{K}_t \subseteq \mathcal{K}_{t-1}$.
The player selects a point $q_t \in \mathcal{K}_t$ with the objective of minimizing the total path length of the selection for $T$ rounds, e.g., $\sum_{t=0}^T \|q_{t+1} - q_{t}\|$. 
There are many algorithms for the NCBC problem such as greedy projection of the previously selected point onto the current body \cite{argue2019nearly}. Among these, the Steiner point selector has been shown to achieve optimal competitive ratio against the offline optimal selector \cite{bubeck2020chasing}. The Steiner point of a convex body $\mathcal{K}$ can be interpreted as the average of the extreme points and is defined as 
$$\text{St}(\mathcal{K}):= \mathbb{E}_{v: \|v\|\leq 1} \left[ g_{\mathcal{K}}(v) \right],$$ 
where $g_{\mathcal{K}}(v) := \text{argmax}_{x\in\mathcal{K}} v^\top x $ {and the expectation is taken with respect to the uniform distribution over the unit ball.} 
The Steiner point selector achieves the following total path length, 
\begin{equation}
\label{eq:steiner}
\sum_{t=0}^T \|\text{St}(\mathcal{K}_{t})- \text{St}(\mathcal{K}_{t+1})\| \leq n \cdot \text{diam}(\mathcal{K}_0), \, \text{for all } T \in \mathbb{N_+}.
\end{equation}
We note that the Steiner point can be approximated with any accuracy by solving sampling based linear programs, \cite[Algorithm 3]{argue2021chasing}.

\subsubsection{Preliminaries on SLS}
\label{sec:sls}
Even when the dynamics \eqref{eq:local_sys} is known, it remains challenging to design distributed and localized control policies that accommodates communication delay and information constraints due to nonconvexity and computational scalability issues. Motivated by this problem, \cite{wang2019system}
introduces the SLS framework that synthesizes distributed controllers by parameterizing controllers with the closed-loop system responses induced under them. In \cite{dean2019safely, dean2020robust, fattahi2020efficient}, SLS plays a central role for model-based learning algorithm design and analysis.

We illustrate SLS via a simple example. Consider a fixed static controller $K \in \mathbb{R}^{n_u \times n_x}$ such that $u(t) = Kx(t)$. Then the system \eqref{eq:global_sys} has the following closed-loop responses to the exogenous disturbances $\mathbf{w}$,
\begin{align}    
\label{eq:static_ex}
    x(t) = \sum_{k=0}^t (A + BK)^k w(t-k-1) , \quad 
    u(t) = \sum_{k=0}^t K(A + BK)^k w(t-k-1),
\end{align}
where we absorb the initial state $x(0)$ into $w(-1)$.  Instead of directly synthesizing $K$, SLS optimizes the linear operators that map $\mathbf{w}$ to $\mathbf{x}$ and $\mathbf{u}$. Let $\Phi^x[k] := (A + BK)^k$ and $\Phi^u[k] := K(A + BK)^k$. Then \eqref{eq:static_ex} can be written as $x(t) = \sum_{k=0}^t \PX{}[k]w(t-k-1)$ and $u(t) = \sum_{k=0}^t \PU{}[k]w(t-k-1)$. In signal and transfer matrix form, $\mathbf{x} = \bPX{} \mathbf{w}$ and $\mathbf{u} = \bPU{} \mathbf{w}$.
We call $\bPX{}$ and $\bPU{}$ with components $\PX{}[k]$ and $\PU{}[k]$ the \emph{closed-loop responses}. More generally, for a fixed linear causal controller $\mathbf{K}$, the closed loop dynamics of \eqref{eq:global_sys} can be written in signal/transfer matrix notation as 
\begin{align}
\mathbf{x} = \Z\mathbf{A} \mathbf{x} + \Z \mathbf{B} \mathbf{u} +\mathbf{w}, \quad \mathbf{u} = \mathbf{K} \mathbf{x}, 
\label{eq:signal}
\end{align}
where $\Z$ is the block-downshift operator with identity matrices of size $n_x$ by $n_x$ in all the first block sub-diagonal positions and zeros everywhere else. Operator $\mathbf{A}, \mathbf{B}$ are block diagonal matrices with matrix $A$ and $B$ on the diagonal respectively. 
We can similarly derive the closed-loop responses $\mathbf{\PX{}}:\mathbf{w} \rightarrow \mathbf{x} $ and $\mathbf{\PU{}}: \mathbf{w} \rightarrow \mathbf{u}$ from \eqref{eq:signal} as
$$\left[\begin{array}{c}\mathbf{x}\\\mathbf{u}\end{array}\right] 
= \left[\begin{array}{c} \left( I - \Z(\mathbf{A} + \mathbf{B} \mathbf{K}) \right)^{-1}\\ \mathbf{K}\left( I - \Z(\mathbf{A} + \mathbf{B} \mathbf{K})\right)^{-1}  \end{array}\right] \mathbf{w} = \left[\begin{array}{c} \mathbf{\PX{}}\\ \bPU{}\end{array}\right] \mathbf{w}. $$
Therefore, SLS uses the closed-loop responses $\bPX{}$ and $\bPU{}$ as the alternative parameterization for controller $\mathbf{K}$.
The following theorem characterizes an affine subspace of the achievable system responses $\mathbf{\PX{}}$ and $\mathbf{\PU{}}$ under some feedback linear controller $\mathbf{K}$.
\begin{theorem}[Adapted from \cite{anderson2019system}]
        \label{thrm:SLS}
        For system \eqref{eq:global_sys}, any linear causal operators $\mathbf{\PX{}},\mathbf{\PU{}}$ with finite impulse response of horizon $H$ and satisfying the following
                \begin{subequations}
                \label{eq:feasibility}
                \begin{align}
                        \PX{}[0] &= I, \quad
                        \PX{}[k+1] = A\PX{}[k] + B\PU{}[k]\,,\quad \text{ for } k = 0,\, \ldots,H-1 \label{eq:feasibility-1}\\
                        \PX{}[\tau] &=0 \text{ for } \tau \geq H \label{eq:feasibility-2}
                \end{align}
                \end{subequations}
         are closed-loop responses for \eqref{eq:global_sys} under a stabilizing linear controller $\mathbf{K}$. 
         Moreover, given any linear causal operators $\bPX{}$, $\bPU{}$ that satisfy \eqref{eq:feasibility}, the following SLS controller constructed using $\bPX{}$, $\bPU{}$, 
    \begin{subequations}
        \label{eq:central-controller}
        \begin{align}
            \hat{w}(t) &= x(t)  - \sum_{k=1}^{H-1} \Phi^x[k] \hat{w}(t-k)\label{eq:central1}\\
            u(t) &= \sum_{k=0}^{H-1} \Phi^u[k] \hat{w}(t-k)\label{eq:central2}
        \end{align}
        \end{subequations}
    with $\hat{w}(0) = x(0)$ achieves the desired closed-loop response prescribed by $\bPX{}$, $\bPU{}$.
\end{theorem}
We remark that here we are restricting to the space of linear causal operators with \emph{finite impulse responses} (FIR) up to horizon $H$, instead of the entire space of linear causal operators. The horizon $H$ is a system-dependent design parameter relating to controllability of \eqref{eq:global_sys}. Under Assumption \ref{assump:controllable}, $H\leq n_x$. Moreover, \eqref{eq:feasibility} provides affine constraints on finite number of nonzero parameters of the closed-loop responses. Therefore, one can tractably optimize the closed-loop responses with respect to a convex cost. A common choice is the Linear Quadratic Regulator (LQR) cost on the state and input expressed in terms of the closed-loop responses, e.g.,
\begin{equation}
    \label{eq:central-sls}
    \begin{aligned}
        &\min_{\bPX{},\, \bPU{}}  && \sum_{k=0}^\infty
        \left\|
         \begin{bmatrix} Q^{1/2} & 0\\ 0& R^{1/2}\end{bmatrix} \begin{bmatrix} \PX{}[k] \\ \PU{}[k] \end{bmatrix}
         \right\|_F \quad \text{s.t.} \,\,\, \eqref{eq:feasibility}\, .
    \end{aligned}
\end{equation}

In this work, we leverage the \textit{SLS controllers} \eqref{eq:central-controller} that is parameterized by and constructed from the operators $\bPX{}$, $\bPU{}$. The interpretation of \eqref{eq:central-controller} is intuitive. When $\bPX{}$, $\bPU{}$ satisfy \eqref{eq:feasibility}, they are valid closed-loop responses, mapping $\mathbf{w}$ to $\mathbf{x}$ and $\mathbf{u}$ under \eqref{eq:global_sys}. Then equation \eqref{eq:central1} estimates the disturbance entering the state in the last time step by computing the difference between the currently observed state $x(t)$ and the counterfactual state $\sum_{k=1}^{H-1} \Phi^x[k] \hat{w}(t-k)$ that should have been observed according to the closed-loop repsonse $\bPX{}$ if there were no disturbance. Indeed, a simple calculation using substitution will reveal that $\hat{w}(t) = w(t-1)$, i.e., that the estimated disturbance from an SLS controller constructed with operators that satisfy \eqref{eq:feasibility} is the perfect one-step delayed estimation of the true disturbances. Then \eqref{eq:central2} computes the control action to attenuate the estimated disturbance according to the prescription of the closed-loop responses $\bPU{}$.

\textbf{Distributed controller synthesis and implementation. }
A feature of SLS is that both the closed-loop response synthesis \eqref{eq:central-sls} and the controller implementation \eqref{eq:central-controller} can be performed in a \emph{distributed} manner, unlike the commonly adopted optimal LQR control method via the Riccati equation \cite{simchowitz2020naive}. This is crucial for scalability of the control algorithm for large-scale systems. 

Observe that \eqref{eq:central-sls} is a column separable problem. This means that we can partition matrix variables $\PX{}[k],\,\PU{}[k]$ into columns such as $\PX{}[k](:,i),\,\PU{}[k](:,i)$ corresponding to each subsystem $i$. We refer to \cite{wang2018separable} for the definition of column separability and the verification of \eqref{eq:central-sls} as a column separable problem.
Thus, subsystem $i$ only needs to solve the column subproblems corresponding to its dynamics \eqref{eq:local_sys} in the global dynamics \eqref{eq:global_sys} as follows. Let $\bm{\phi}^{i,x}$ and $\bm{\phi}^{i,u}$ denote the $i$th column of $\bPX{}$ and $\bPU{}$ respectively and let $\bm{\phi}^i$ collectively stand for $\bm{\phi}^{i,x}$, $\bm{\phi}^{i,u}$. The $i$th column subproblem is  
\begin{equation}
\label{eq:central-column}
\begin{aligned}
 &\min_{\bm{\phi^i}}  && \sum_{k=0}^\infty\left\|
\begin{bmatrix} Q^{1/2} & 0\\ 0& R^{1/2}\end{bmatrix} \begin{bmatrix} \phi^{i,x}[k] \\ \phi^{i,u}[k] \end{bmatrix}\right\|_F \\
        &\quad \text{s.t.} && \phi^{i,x}[k+1] = A \phi^{i,x}[k] + B \phi^{i,u}[k]\,\quad  \text{ for }k=0,\ldots,H-1 \\
        &\quad && \phi^{i,x}[0] = e_i,\quad \phi^{i,x}[H] = 0 \,,
\end{aligned}
\end{equation}
where the constraints in \eqref{eq:central-column} is the column-wise decomposition of the constraints \eqref{eq:feasibility} for the closed-loop repsonse synthesis \eqref{eq:central-sls}. 
It is straightforward to see that stacking the solutions to the column subproblems recovers the optimal solution to \eqref{eq:central-sls}.

When the dynamics interaction among subsystems \eqref{eq:local_sys} is sparse, additional sparsity can be imposed on the closed-loop responses during synthesis \eqref{eq:central-sls}. With sparse $\bPX{}$ and $\bPU{}$, the implementation of the controller \eqref{eq:central-controller} can be distributed in a similar decomposition as the synthesis procedure. In particular, each subsystem computes a disjoint subset of coordinates of $\hat{w}(t)$. Due to sparsity, such local computation for subsystem $i$ only requires the solutions to the column subproblems from the local neighbors of $i$ via communication instead of from the entire network.

\section{Online stabilization under adversarial disturbances}
\label{sec:centralized}
In this section, we propose a novel online algorithm presented in \Cref{alg:centralized} that stabilizes an unknown networked linear system \eqref{eq:global_sys} under bounded and potentially adversarial disturbances. The algorithm selects hypothesis models using methods for NCBC and constructs an SLS distributed controller based on the hypothesis model. Our approach is distinguished from prior learning-based control methods in that it does not perform system identification as part of the algorithm. 

We first introduce our algorithm without any communication or localization constraints. Then, in \Cref{sec:distributed}, we extend the algorithm to a distributed one that accommodates communication delay and local information (\Cref{alg:main}). {Though inspired by the approach in \cite{ho2021online}, \Cref{alg:centralized,alg:main} are the first to consider the control goal of stabilization, which can not be subsumed under the framework proposed in \cite{ho2021online} where only binary cost functions are considered. To cast stabilization in terms of a binary cost function, one needs to specify the largest norm of the state and control input of the closed-loop system, which is unavailable a priori\footnote{A crude approximation of the largest norm can be achieved by computing the worst-case state norm over all systems in the initial parameter set $\mathcal{P}_0$, but such approximation results in significant conservatism and requires the knowledge of control theoretical constants of the controller, e.g., SLS controllers, that may not always be available.}.} Moreover, our algorithms perform both the parameter selection and the model-based control design \textit{distributedly} for each local subsystem based on \textit{delayed} information from other subsystems, whereas \cite{ho2021online} is a single-agent algorithm.  

\Cref{alg:centralized} starts with the construction of a set of candidate models that are consistent with the online data (line~\ref{algoline:central-consist-1}) after observing the latest state transition (line~\ref{algoline:obs}). A hypothesis model is selected from the set of candidate models with NCBC techniques (line~\ref{algoline:central-consist-3}) if the previously selected hypothesis model is invalidate by the new observation (line~\ref{algoline:central-consist-2}). Based on the selected hypothesis model, model-based control design is performed using the SLS procedure introduced in \Cref{sec:sls} (line \ref{algoline:central-sls-1} - \ref{algoline:central-sls-2}).
We discuss the details of \Cref{alg:centralized} in the following subsections.
\begin{algorithm2e}[t]
\LinesNumbered
    \DontPrintSemicolon
    \SetNoFillComment
    \KwIn{Parameter set $\mathcal{P}_0$}
    \KwInit{$t=0$, $u(0)=0$}
    \For{$t = 1,2,\dots$ }{
             Observe $x(t)$ \label{algoline:obs}\;
             \tcc{CONSIST: Select consistent models}
             Construct $\mathcal{P}_t$ with \eqref{eq:central-set} \label{algoline:central-consist-1}\;
            \lIf{$\Theta_{t-1} \in \mathcal{P}_t$} 
{
    $\Theta_t \leftarrow \Theta_{t-1} $ 
}\label{algoline:central-consist-2}
\lElse
{$\Theta_t \leftarrow$ $\text{St}(\mathcal{P}_t)$ } \label{algoline:central-consist-3}
            \tcc{CONTROL: Perform model-based control with SLS}
            \label{algoline:model-based}
            Synthesize $\bPX{t}$, $\bPU{t}$ using \eqref{eq:central-sls} based on $\Theta_t$ \label{algoline:central-sls-1} \;
            Compute $u(t)$ using the SLS controller \eqref{eq:central-controller} with $\bPX{t}$, $\bPU{t}$   \label{algoline:central-sls-2}
    }
    \caption{Online stabilization under adversarial disturbances}
    \label[algorithm]{alg:centralized}
\end{algorithm2e}

\subsection{CONSIST: Consistent hypothesis model selection}
The first component of \Cref{alg:centralized} is to select a hypothesis model $\Theta_t$ in order to perform model-based control. We name this component CONSIST. Due to the potentially adversarial disturbances such as state-dependent noise, standard identification methods such as linear regression do not guarantee accurate estimation of the model. Instead, we leverage NCBC for hypothesis selection. 

After observing the latest state transition from $x(t-1)$, $u(t-1)$ to $x(t)$, the algorithm constructs the set of all $\Theta$'s such that $A(\Theta), B(\Theta)$ satisfy \eqref{eq:global_sys} with some admissible disturbances defined in Assumption \ref{assump:noise}. In particular, each observed transition defines a set of linear constraints on $\Theta$ and we construct the \textit{consistent parameter set}, $\mathcal{P}_t$ at each time $t$, as 
\begin{align}
    \label{eq:central-set}
    \mathcal{P}_t:= \left\{\Theta \in \mathcal{P}_{t-1} \,:\,  \left\|x(t) - \left(  A(\Theta) x(t-1) +  B(\Theta)u(t-1) \right) \right\|_\infty \leq W \right\}
\end{align}
with $\mathcal{P}_0$ as the local initial parameter set defined in Assumption \ref{assump:compact}. Note that the consistent parameter set $\mathcal{P}_t$ is always convex, and nested within the parameter set $\mathcal{P}_{t-1}$ recursively. 
Moreover, $\mathcal{P}_t$ is nonempty for all $t\in\mathbb{N}_+$ because the true parameter $\Theta^\star$ belongs to every $\mathcal{P}_t$. The key property of $\mathcal{P}_t$ is that for all $t$, any $\Theta_t \in \mathcal{P}_t$ could have generated the observed trajectory up to $x(t)$ and is equally likely to be the true system model. By construction, the observed state trajectory can be written as 
\begin{subequations}
\label{eq:consistency}
\begin{align}
    x(t) 
    &= A(\Theta^\star) x(t-1) + B(\Theta^\star)u(t-1) + w^\star(t-1)  \label{eq:true}\\
        &=  A({\Theta}_t) x(t-1) + B({\Theta}_t)u(t-1) + \tilde{w}(t-1), \label{eq:consistent}
\end{align}
\end{subequations}
where ${w^\star}(t)$ is the true disturbance and $\{\tilde{w}(k)\}_{t=0}^\infty$ is some admissible disturbance sequence such that $\|\tilde{w}(t)\|_\infty \leq W$. 
We say a model is \emph{consistent} with observations up to time $t$ if it belongs to $ \mathcal{P}_t$. Among all consistent models, we need to select a hypothesis model $\Theta_t$ in order to perform model based control. An ideal candidate is one that can remain inside of \emph{future} consistent parameter sets. To see why, consider an extreme case where the first selected parameter $ \Theta_1$ stays consistent for the entire online operation as we apply control actions generated based on $\Theta_1$.  
Since the consistent model \eqref{eq:consistent} generates the same trajectory as the true model \eqref{eq:true}, any guarantees that the model-based control policy has for $\Theta_1$ will manifest in the observation. Note $\Theta_1$ does not necessarily have to be close to $\Theta^\star$. {We remark that $\mathcal{P}_t$ is called the membership set in control literature \cite{bai1998convergence, akccay2004size}, where most work study the convergence properties of $\mathcal{P}_t$ to $\Theta^\star$ in the context of system identification given input-output data. In contrast, we construct $\mathcal{P}_t$ not to identify the system but to use it for downstrem control tasks in the online interactive setting.}

This intuition motivates us to select a $\Theta_t$ that could remain an element of the (yet unknown) future consistent parameter set. In particular, if the hypothesis model selected at a previous time is consistent for the current observation, we continue to use it. If the previous hypothesis model is invalidated by the new observation, then we want to select a new $\Theta_t$'s from the nested and convex body $\mathcal{P}_t$ with the objective of moving as little as possible for future bodies. This is an instance of NCBC introduced in \Cref{sec:NCBC}. 
The total path length cost function in NCBC formalizes a measure of \textit{model consistency} in our case: the less the a selector moves, the longer the selected points stay consistent overall. In \Cref{alg:centralized}, we select the Steiner point of $\mathcal{P}_t$ as the hypothesis model. The finite path length guarantee of Steiner point in \eqref{eq:steiner} 
can be interpreted as a finite budget for the adversarial disturbances: If the disturbances try to make the state norm large, then the selected (wrong) hypothesis model will be quickly invalidated thanks to the excitation from the disturbances. This will make CONSIST frequently re-select new hypothesis models. However, such inconsistent model selection has bounded occurrences due to the finite path length guarantee \eqref{eq:steiner} of the Steiner point, i.e., CONSIST gains information and stops moving eventually. 

\subsection{CONTROL: Model-based control with SLS}
After the selection of a hypothesis model $\Theta_t$ from the consistent parameter set, \Cref{alg:centralized} performs the SLS closed-loop response synthesis \eqref{eq:central-sls} and implementation \eqref{eq:central-controller} based on $\Theta_t$. We name this component of the algorithm CONTROL.

\subsection{Distributed implementation of \Cref{alg:centralized}}
\label{sec:distributed_implementation}
Per discussion in \Cref{sec:sls}, it is straight forward to see that \Cref{alg:centralized} can be implemented by each subsystem in a distributed fashion. In particular, in the CONSIST component, subsystem $i$ constructs a local consistent parameter set $\mathcal{P}^i_t$ based on the local observations generated from the local dynamics \eqref{eq:local_sys}. Subsystem $i$ then selects the Steiner point of $\mathcal{P}^i_t$ as its local hypothesis model $\theta^i_t$. In the CONTROL component, all subsystems collects the local hypothesis models from other subsystems and construct a global estimate $\Theta_t = (\theta^1_t,\ldots,\theta^N_t)$ since we assume no communication delay here. Based on $\Theta_t$, each subsystem synthesizes columns of $\bPX{t}$ and $\bPU{t}$ by solving the subproblems decomposed from \eqref{eq:central-sls}. After collecting and assembles the column solutions via instantaneous communication, each subsystem computes a disjoint subset of coordinates of $\hat{w}(t)$ and $u(t)$, corresponding to the positions of the local states $x^i(t)$ and input $u^i(t)$ in the global dynamics \eqref{eq:global_sys} respectively.

\subsection{Stability guarantee}
Our main results in this section is the following ISS guarantee for \Cref{alg:centralized}.
\begin{theorem}
    \label{thrm:global}
    Under Assumption \ref{assump:noise}-\ref{assump:controllable}, Algorithm \ref{alg:centralized} guarantees the stability of the closed loop of \eqref{eq:global_sys} in the sense of ISS such that for all $t \geq t_0$ 
    $$
         \max\{\|x(t)\|_\infty , \|u(t)\|_\infty\} \leq  \bigO \left( e^{{n_x}^{5/2}}\right) \cdot \left( e^{- (t-t_0)/H} x(t_0) + \sup_{t_0\leq k< t} \infnorm{w(k)} \right),
    $$
  where $x(t_0)$ is the initial condition, ${n_x}$ is the total state dimension of the global network \eqref{eq:global_sys}, and $H$ is the finite impulse response horizon for the SLS model-based control synthesis.
\end{theorem}
We remark that the decay factor $e^{-t/H}$ corroborates the fact that $H$ quantifies the controllability of the parameter set  $\mathcal{P}_0$. Intuitively, the smaller $H$ can be for the SLS synthesis \eqref{eq:synth} to be feasible, the easier the systems in the set can be learned and controlled. 
\begin{proof} 
The main idea of the proof is as follows. First, we characterize the closed loop dynamics of \eqref{eq:global_sys} under \emph{any} SLS controllers constructed with arbitrary linear causal operators (\Cref{lem:closed-loop}). We then relax the original SLS condition \eqref{eq:feasibility} in \Cref{thrm:SLS} to a sufficient condition for ISS of the closed-loop dynamics under bounded adversarial disturbances (\Cref{lem:sufficient}). Crucially, we show that the bounded path length property \eqref{eq:steiner} of the selected hypothesis models in \Cref{alg:centralized} implies the satisfaction of the sufficient condition for closed-loop stability. This implication is established through a novel perturbation analysis (\Cref{thrm:sensitivity}) of the SLS closed-loop response synthesis problem \eqref{eq:central-sls}. We defer the proofs of the helper lemmas used here to \Cref{sec:proof_global}. 

Specifically, we show that given arbitrary $\bPX{},\,\bPU{}$ with FIR horizon $H$, the closed-loop dynamics of \eqref{eq:global_sys} under an SLS controller constructed from $\bPX{},\,\bPU{}$ is characterized as follows.
\begin{lemma}[Closed-loop characterization]
    \label[lemma]{lem:closed-loop}
        The closed loop of \eqref{eq:global_sys} under Algorithm \ref{alg:centralized} is characterized as follows for all time $t \in \mathbb{N}$:
    \begin{subequations}
    \label{eq:cl}
        \begin{align}
            x(t)  &= \sum_{k=0}^{H-1} \PX{t}[k] \hat{w}(t-k), \quad
            u(t) = \sum_{k=0}^{H-1} \PU{t}[k] \hat{w}(t-k) \label{eq:cl-1}\\
            \hat{w}(t) &= \sum_{k=1}^H \left( A \PX{t-1}[k-1] + B \PU{t-1}[k-1] - \PX{t}[k] \right)\hat{w}(t-k) + {w}(t-1). \label{eq:cl-2}
        \end{align}
    \end{subequations}
    where $A,\,B$ are the true model parameters from \eqref{eq:global_sys} while $w(t)$ is the true unknown bounded disturbances with $\infnorm{w(t)} \leq W$. The linear causal operators $\bPX{t}$, $\bPU{t}$ are synthesized via \eqref{eq:central-sls} based on the selected hypothesis model at $t$ and $\hat{w}(t)$ is the estimated disturbance from the SLS controller \eqref{eq:central-controller}.  
\end{lemma}
This result generalizes \Cref{thrm:SLS} where we characterize the closed loop behaviour of SLS controllers constructed from \textit{any} linear casual operators, not necessarily those satisfying \eqref{eq:feasibility-1}. Under \Cref{alg:centralized}, we can further replace the true model in \eqref{eq:cl-2} with the selected hypothesis model (Steiner point of the consistent set) $\Theta_t$, i.e., 
$$\eqref{eq:cl-2} = \sum_{k=1}^H \left( A(\Theta_t) \PX{t-1}[k-1] + B(\Theta_t) \PU{t-1}[k-1] - \PX{t}[k] \right)\hat{w}(t-k) + \tilde{w}(t-1),$$ 
with admissible disturbances such that $\infnorm{\tilde{w}} \leq W$ due to the consistency property \eqref{eq:consistency} of $\Theta_t$. 

Moreover, \Cref{lem:closed-loop} leads to a simple sufficient condition for stability of the closed loops under any SLS controllers. To see this, we first argue that there exist constants that bound the decay rate of the closed loop responses synthesized from \eqref{eq:central-sls}.
In particular, due to the finite impulse response property imposed by \eqref{eq:feasibility-2} of the synthesized closed-loop responses, there always exists a large enough $C>0$ and $\rho \in (0,1)$ such that 
$$ \norm{\begin{bmatrix}\phi^{i,x}[k]\\ \phi^{i,u}[k]\end{bmatrix}}_F\leq C \rho^k \quad \text{for all closed-loop responses satisfying \eqref{eq:feasibility-2}}. $$
This property is commonly employed in SLS-based analysis \cite{dean2019safely,dean2020sample, fattahi2020efficient}. We use $C$ and $\rho$ for the sake of proof here and does not require the knowledge of them for Algorithm \ref{alg:centralized} to run.

With the decay property, according to \Cref{lem:closed-loop}, if $\infnorm{\hat{w}(t)} \leq \hat{W}_\infty$ for some $\hat{W}_\infty >0$, then we can bound the global state via \eqref{eq:cl-1} as follows, 
\begin{align*}
    \infnorm{x(t)}  &\leq \hat{W}_\infty \sum_{k=0}^{H-1} \infnorm{ \PX{t}[k] } \leq \hat{W}_\infty C^{1/2}{n_x}^{1/2} \frac{1}{1-\rho^{1/2}}.  \nonumber
\end{align*}
The bound on control input $\infnorm{u(t)}$ follows analogously. Therefore, the stability of the closed loop reduces to the boundedness of $\hat{w}(t)$ in \eqref{eq:cl-2}.
To show this, we prove the following.

\begin{lemma}[Sufficient condition for $H$-convolution ISS]
    \label[lemma]{lem:sufficient}
    Let $H \in \mathbb{N}_+$. 
    For $k \in [H]$, let $\left\{a_{t}[k]\right\}_{t=1}^\infty$ and $\left\{w_{t}\right\}_{t=0}^\infty$ be positive sequences.
    Let $\{s_t\}_{t=0}^\infty$ be a positive sequence such that
    \begin{align}
    s_t \leq \sum_{k=1}^H a_{t-1}[k]\cdot s_{t-k} + w_{t-1} \,.
    \end{align}
    Then $\{s_t\}_{t=0}^\infty$ is ISS if $\sum_{t=0}^\infty \sum_{k=1}^H a_{t}[k] \leq L$ for some $L \in \mathbb{R}_+$. In particular, for all $t\geq t_0$,
    \begin{align}
    \label{eq:iss-seq}
    s_t \leq e^{-(t-t_0)/H} \cdot e^L s_{t_0} + \frac{ \left( e^L + e -1 \right)}{e-1} \sup_{t_0\leq k< t} w_k \, .\end{align}
\end{lemma}

The above sufficient condition is suitable for analyzing dynamical evolution under adversarial inputs. Consider taking the norm on both sides of \eqref{eq:cl-2}. Then \Cref{lem:sufficient} is immediately applicable with $s_t = \infnorm{\hat{w}(t)} $, and 
\begin{equation}
    \label{eq:error-term}
    a_t[k] =\infnorm{ A(\Theta_t) \PX{t-1}[k-1] + B(\Theta_t) \PU{t-1}[k-1] - \PX{t}[k]} .
\end{equation}
Therefore, a sufficient condition for ISS of \eqref{eq:global_sys} under \Cref{alg:centralized} is the boundedness of \eqref{eq:error-term} summing over time $t\in \mathbb{N}_+$ and horizon $k\leq H$. This quantity represents the total error of the implemented closed-loop responses $\bPX{t}, \, \bPU{t}$ synthesized from the selected hypothesis dynamics model $\Theta_t$, with respect to the \textit{correct} closed-loop responses generated from the true model $\Theta^\star$. 

To bound \eqref{eq:error-term}, we make a crucial connection between the total path length of the Steiner point model selection in \Cref{alg:centralized} and \eqref{eq:error-term}. This is established via the following perturbation result for the SLS closed-loop response synthesis problem \eqref{eq:central-sls}, where the formal statement (\Cref{thm:mainsensitivity_global}) and proof is presented in Appendix \ref{sec:sensitivity}. 
\begin{theorem}[Informal, Perturbation bound]
    \label{thrm:sensitivity}
Let $\phi^\star(A,B) := [\mathbf{x}^{\star,\top}, \mathbf{u}^{\star,\top}]^\top$ denote the concatenated optimal solution to the following optimization problem
\begin{equation}
\label{eq:mpc}
\begin{aligned}
    &\min_{x,u}  &&
    \,\,\, \sum_{t = 0}^H x(t)^TQx(t) + u(t)^TRu(t)   \\
    &\text{s.t.} &&  x(t+1) = Ax(t) + Bu(t), \quad
     x(0) = x_0, \quad x(H) = 0 \, ,
\end{aligned}
\end{equation}
with $Q,R \succ 0$.
Let $(A_1,B_1)$ and $(A_2,B_2)$ be two system matrices such that \eqref{eq:mpc} is feasible. Then the corresponding optimal solutions $\phi^\star(A_1, B_1)$ and $\phi^\star(A_2, B_2)$ satisfy
\begin{align*}
\|\phi^\star(A_1, B_1) - \phi^\star(A_2, B_2)\|_F \leq
\Gamma \norm{\begin{bmatrix} A_1-A_2\\B_1-B_2\end{bmatrix}}_F \,,
\end{align*}
where $\norm{\phi^\star(A,B)}_F := \sum_{k=0}^H \norm{[x(k)^\top , u(k)^\top]}_F$. Constant $\Gamma>0$ involves the system theoretical quantities for $A_1, A_2, B_1, B_2, Q,R$.
\end{theorem}
The quadratic program \eqref{eq:mpc} corresponds to the column-wise decomposed subproblems of the SLS closed-loop response synthesis \eqref{eq:central-sls}. Therefore, \eqref{eq:error-term} can be bounded as follows.
\begin{align*}
     \eqref{eq:error-term} &= \infnorm{ A(\Theta_t) (\PX{t-1}[k-1] - \PX{t}[k-1]) + B(\Theta_t) (\PU{t-1}[k-1] -\PU{t}[k-1])}\\
     &\leq {2n_x} \kappa \norm{\begin{bmatrix}  \PX{t-1}[k-1] - \PX{t}[k-1] \\ \PU{t-1}[k-1] -\PU{t}[k-1]\end{bmatrix}}_F,
\end{align*}
where the equality is due to the constraint \eqref{eq:feasibility} during the model-based control step in Line \ref{algoline:model-based} of \Cref{alg:centralized}. The inequality invokes Assumption \ref{assump:compact}. Finally we show the total error summing \eqref{eq:error-term} over all time step $t$ and horizon $k \leq H$ is bounded by the total path length of the selected hypothesis models via the Steiner point.
\begin{align}
    \sum_{t=0}^\infty \sum_{k = 1}^H \eqref{eq:error-term} &\leq {2n_x \kappa}  \sum_{t=0}^\infty \sum_{k = 1}^H \norm{\begin{bmatrix}  \PX{t-1}[k-1] - \PX{t}[k-1] \\ \PU{t-1}[k-1] -\PU{t}[k-1]\end{bmatrix}}_F \nonumber\\
    &\leq {2n_x^{3/2} \kappa \Gamma} \sum_{t=0}^\infty \norm{ \Theta_{t-1} - \Theta_t}_F \quad \leq {2n_x^{5/2} \kappa \Gamma} \text{diam}(\mathcal{P}_0), \label{eq:total-error}
    \end{align}
where we use \Cref{thrm:sensitivity} for the second inequality and the total path length bound \eqref{eq:steiner} of the Steiner point selector for the last inequality. Finally, we plug the total bound \eqref{eq:total-error} in \eqref{eq:iss-seq} for an ISS bound on $\hat{w}(t)$, which gives the desired state and control input bound in \Cref{thrm:global}.
\end{proof}

\begin{remark}
    NCBC algorithms other than the Steiner point selector can be substituted in \Cref{alg:centralized} as long as the finite path length guarantee \eqref{eq:steiner} holds. Therefore, we can use a more computationally efficient algorithm with respect to the number of constraints in \eqref{eq:central-set}, such as greedy projection, at the expense of a larger worst-case path length bound. 
Such trade-off is potentially important since the number of constraints in \eqref{eq:central-set} grows linearly with time. A topic of continuing work is to find an efficient representation of \eqref{eq:central-set} that does not involve linear growth in the number of constraints.
\end{remark}

\paragraph{Comparison of \Cref{thrm:global} with previous results} 
Compared to the state-of-art system identification-based algorithm for online control under adversarial disturbances given in \cite{chen2021black}, which induces $\Omega(2^n)$ state and control input norm, our algorithm also incur state norms that are exponential-polynomial in the global dimension. However, our bound is a worst-case guarantee which is on average not achieved during deployment. On the other hand, the exponential bound in \cite{chen2021black} is qualitatively obtained, since system identification-based methods require full excitation of the system despite adversarial disturbances \cite[Lemma 14]{chen2021black}. This is the reason behind the orders of magnitude of performance improvement of our algorithm over system identification-based methods observed in the numerical study shown in \Cref{fig:sysid_main}. 

\paragraph{Comparison of \Cref{thrm:sensitivity} with previous results} The Lipschitz continuity of optimal control problems, similar to \eqref{eq:mpc}, has been investigated in learning-based LQR literature, e.g., \cite{tu2019gap,faradonbeh2020optimism}. However, our perturbation result \Cref{thrm:sensitivity} (formal statement in \Cref{thm:mainsensitivity_global}) is with regard to a finite-horizon quadratic program with terminal state constraints, whereas previous Lipschitz continuity analysis is performed with respect to the infinite-horizon LQR optimal gain. As a result, we use a different set of tools from matrix theory, unlike the Riccati equation (value function) based analysis for infinite-horizon LQR problems in previous works. In \Cref{sec:distributed}, we further generalize the perturbation result to handle sparsity constraints.

\section{Adversarial stabilization with information constraints}
\label{sec:distributed}
The implementation of \Cref{alg:centralized} assumes that each subsystem has instantaneous access to the information from other subsystems, such as the local consistent hypothesis models, and the column solutions to the subproblems decomposed from \eqref{eq:central-sls}. Such instantaneous information sharing is often unrealistic in large-scale networked control systems. Therefore, in this section we extend the presentation in \Cref{sec:centralized} to a fully distributed algorithm, shown in \Cref{alg:main}, that for the first time guarantees the stability of unknown interconnected LTI systems with information constraints under bounded adversarial disturbances. These results are the main contributions of the paper. 

Specifically, we consider two classes of information constraints, namely \emph{communication delay} and \emph{local information}, which we define formally below. After defining these information constraints, we describe the adjustments to \Cref{alg:centralized} and present our main result.

\subsection{Communication delay} 
\label{sec:comm_constraints}
A key feature of large-scale networked systems is that information observed locally at each subsystem cannot be immediately available to the global network. Instead, information sharing among subsystems is constrained by communication limitations.  
Such limitations often lead to delayed partial observation and pose further challenges for learning-based algorithm design \cite{ye2021sample, li2021distributed, alonso2021data}. 
To formalize the communication constraints, we define a communication graph $\mathcal{G}^C=(V^C, E^C)$ for \eqref{eq:global_sys}, where $V^C = [N]$ and $E^C$ is the set of directed communication link from one subsystem to the other. Self-loops at all vertices are included in  $E^C$ and they represent zero delay. The communication graph is demonstrated by the solid blue lines in Figure \ref{fig:example}. We use $\mathcal{C} \in \{1,0\}^{N\times N}$ to denote the adjacency matrix associated with the communication graph $\mathcal{G}^C$. 
Moreover, we define the information delay induced by $\mathcal{G}^C$ as follows. 
\begin{definition}[Information delay]
\label{def:delay}
    The information delay from subsystem $i$ to $j$ is defined to be the total distance of the shortest path from $i$ to $j$ according to $\mathcal{G}^C$ and is denoted as $d(i\rightarrow j)$.
\end{definition}

Globally, the $k$th power of the adjacency matrix $\mathcal{C}^k$ has nonzero $(i,j)$th entry if subsystem $i$ gets $k$-delayed information from subsystem $j$.
Locally, at time step $t$, subsystem $i$ has access to subsystem $j$'s full information up to time $t-d(j \rightarrow i)$. Moreover, $d(j\rightarrow i)$ is the smallest integer such that $\mathcal{C}^{d(j\rightarrow i)} (i,j) \not =0$. With slight abuse of notation, we write $\mathcal{C}^k$ to mean the support of the matrix so $\mathcal{C}^k\in \{1,0\}^{N\times N}$. 
\begin{example}
    Consider the system in Figure \ref{fig:example} where the solid blue line denotes the communication among subsystems. The adjacency matrix $\mathcal{C}$ is depicted in Figure \ref{fig:C}.
    Observe that $\mathcal{C}(1,3) =0 $ but $\mathcal{C}^2(1,3) \not = 0$. Therefore, the delay from subsystem 3 to subsystem 1 is $d(3 \rightarrow 1) = 2$. 
\end{example}

Given $\mathcal{G}^C$, we make a mild assumption on the communication delay. 
This assumption ensures that the graph describing the global dynamics is a subgraph of the communication graph. Such an assumption ensures nested information structure \cite{ho1972team} and is commonly adopted \cite{lamperski2015optimal,ye2021sample}. It holds true for systems where communication operates at least as fast as the dynamical propagation.
\begin{assumption}[Communication Topology]
    \label{assump:comm}
     $\mathcal{C}(i,j) = 1$ for all $j \in \mathcal{N}(i)$.
\end{assumption}

The communication delay model considered here is well-established in the distributed control literature \cite{lamperski2015optimal, shah2013cal, ye2022regret} and is applicable to many engineering systems \cite{shi2012robust, ma2015stabilization}.  We refer interested readers to \cite{rotkowitz2008information} for a detailed discussion on information structures and their consequences for distributed control design. While we specify the communication delay to be synchronous with the discrete time dynamics propagation for ease of exposition, our results can be readily applied to systems with faster communication than the dynamics propagation.



\subsection{Local information}
\label{sec:localized_control}
Even though communication delay causes asynchronous partial information for each subsystem, eventually each subsystem can obtain the delayed global information. 
However, due to the scale of the global network, it can be prohibitively costly for subsystems to compute their local control actions using such delayed global information. Moreover, a larger delay between subsystems means, intuitively, that they are more dynamically decoupled due to Assumption \ref{assump:comm}. Therefore, by discarding information from far-away subsystems, each subsystem has a smaller and more up-to-date information set. 
A common approach is to require each subsystem  $i$ to only use delayed information from a local neighborhood.
In this work, we define three neighborhoods, $\din{i}$, $\dout{i}$, and $\M{i}$ that subsystem $i$ is allowed to access information from.
This is sometimes referred to as localized control in multi-agent reinforcement learning \cite{qu2020scalable,lin2020distributed,qu2020bscalable} and distributed control \cite{alonso2021data, wang2018separable} as a method for ensuring a scalable implementation of the control policy in large-scale networked systems. 
Below we define each of the neighborhoods.

\begin{definition}[$\bar{d}$-incoming/outgoing neighbors] The $\bar{d}$-incoming and outgoing neighbors of subsystem $i$ according to $\mathcal{G}^C$ are respectively
\label{def:neighbors}
    $$\din{i}=\{j\in[N]: d(j\rightarrow i )\leq \bar{d}\}\,, \quad \dout{i}=\{j\in[N]: d(i\rightarrow j )\leq \bar{d}\}\,.$$
\end{definition}
The localization parameter $\bar{d}$ is a design choice that is network structure dependent. Here we focus on the cases where the dynamics topology and communication graph have sparse enough edges that the network structure can be leveraged to design a localization parameter $\bar{d}$ (given) that is much smaller than the size of the global network and scales well with the number of subsystems. 

\begin{definition}[$\bar{d}$-interaction neighbors] 
\label{defn:M}
The $\bar{d}$-interaction neighbors of subsystem $i$ according to local interaction \eqref{eq:local_sys} and $\mathcal{G}^C$ is defined as
    $$\M{i} = \left\{ \ell \in [N]: j \in \N(\ell) \text{ for some }j \in \dout{i} \right\}.$$
\end{definition}
The intuition behind $\M{i}$ is that any subsystem $\ell \in \M{i}$ is dynamically influence by subsystem $j$ because $j \in \N(\ell)$. Furthermore, $j$ makes local decisions such as $u^j(t)$ based on the information from subsystem $i$ because $j \in \dout{i}$. Therefore, it is sensible for subsystem $i$ to take the information from $\ell$ into consideration during decision making, since $\ell$ will be indirectly affected by decisions made at $i$ through information sharing and dynamical interaction via $j$.

Finally, we make the following feasibility assumption. 
\begin{assumption}[Feasibility]
    \label{assump:feasibility}
    For all $\Theta \in \mathcal{P}_0$, there exists a stabilizing controller for $A(\Theta), B(\Theta)$ such that each agent with local dynamics \eqref{eq:local_sys} uses delayed and locally available information from its $\bar{d}$-interaction, incoming, and outgoing neighbors according to $\mathcal{G}^C$.
\end{assumption}
\Cref{assump:feasibility} ensures the well-posedness of the distributed controller learning problem and is commonly employed  \cite{ibrahimi2012efficient, abbasi2011regret, li2021safe}. If a parameter set $\mathcal{P}_0$ has a few singular points where $(A,B)$ loses feasibility such as when $B=0$, a simple heuristic is to ignore these points in the algorithm since we assume the underlying system is controllable. We discuss the case of nonconvex parameter sets in \Cref{sec:convex}.

\subsection{A fully distributed and localized algorithm}
We now describe how to extend \Cref{alg:centralized} to handle communication delay and localized control constraints.  To do this we add additional information exchange steps to \Cref{alg:centralized} in each of the two components. The full algorithm is shown in \Cref{alg:main}. For ease of exposition, we let the subsystems have scalar state and fully actuated control actions ($n_x=n_u=N$) in order to minimize notation. It is straight-forawrd to generalize the presented algorithm and analysis to vector subsystems.
\begin{algorithm2e}[t]
    \LinesNumbered
    \DontPrintSemicolon
    \SetNoFillComment
    \KwIn{Parameter set $\mathcal{P}_0$}
    \KwInit{$t=0$, $u(0)=0$, $\mathcal{I}(i,0)=\emptyset$ for $i \in [N]$}
    \For{$t = 1,2,\dots$ }{
        \For{ Subsystem $i = 1,2,\dots,N$}{
            Observe $x^i(t)$ \label{algoline:observe} \\
            \tcc{CONSIST: Select consistent models}
             Construct $\pit$ with \eqref{eq:local-chasing}\;
 \lIf{$\theta^i_{t-1} \in \pit$}{$\thit \leftarrow \theta^i_{t-1} $ \label{line:consist-3}}
            \lElse
            {$\thit \leftarrow$ $\text{St}(\pit)$ } 
            \tcc{CONTROL: Perform model-based control with SLS}
            Assemble local estimate of the global model $A\left(\hat{\Theta}^i_t\right),B\left(\hat{\Theta}^i_t\right)$ with \eqref{eq:delayed-local-model} \label{algoline:assemble}\\
            Synthesize closed-loop response columns $\colit$ using  \eqref{eq:synth} based on $A\left(\hat{\Theta}^i_t\right),B\left(\hat{\Theta}^i_t\right)$ \label{algoline:local-synth}\;
            Assemble delayed local column solutions $\bigcup_{j\in \din{i}} \coljt$ \label{algoline:assemble2}\;
            Compute local control action $u^i(t)$ using \eqref{eq:local-controller} with the assembled column solutions \label[algorithm]{algoline:control}     
        }
    }
    \caption{Distributed online stabilization under information constraints}
    \label[algorithm]{alg:main}
\end{algorithm2e}

\subsubsection{CONSIST}
\label{sec:local_select}
This component of \Cref{alg:main} is identical to that of the distributed implementation of \Cref{alg:centralized} discussed in \Cref{sec:distributed_implementation}. Formally, subsystem $i$ constructs the \textit{local consistent parameter set}, $\pit$ according to local dynamics \eqref{eq:local_sys} as 
\begin{align}
    \label{eq:local-chasing}
    \pit:= \left\{\theta^i \in \mathcal{P}^i_{t-1} \,:\,  \left\|x^i(t) - \left(\sum_{j\in \N(i)}  A^{ij}(\theta^i) x^j(t-1) +  B^{ij}(\theta^i)u^j(t-1) \right) \right\|_\infty \leq W \right\}
\end{align}
with $\mathcal{P}_0^i$ as the local initial parameter set defined in Assumption \ref{assump:compact}. The communication delay pattern allows the construction of $\pit$ because each subsystem $i$ precisely has access to $x^j(t-1)$ and $u^j(t-1)$ from its immediate dynamical interaction neighbors $\mathcal{N}(i)$ by Assumption \ref{assump:comm}. 

Analogous to \Cref{alg:centralized}, each subsystem $i$ selects the Steiner point of $\pit$ as the \emph{local hypothesis model} if the previous selection is invalidated by the latest observation.

\subsubsection{CONTROL}
\label{sec:subcontroller}
Since the local hypothesis models are no longer shared instantly among subsystems due to the communication delay and local information constraints, we modify the model-based control component of \Cref{alg:centralized} and carefully keep track of the available information. To give an overview, at every step $t$, subsystem $i$ first assembles a local estimate of the ``global'' model using delayed information from other subsystems (line \ref{algoline:assemble}). Based on the estimated global model, subsystem $i$ synthesizes the $i$th column of the SLS closed-loop responses by solving the column subproblem of \eqref{eq:central-sls} as discussed in \Cref{sec:distributed_implementation} (line \ref{algoline:local-synth}). Then, subsystem $i$ assembles a local SLS controller with the local column solutions $\colit$ computed from the previous step and the delayed column solutions from other subsystems (line \ref{algoline:assemble2}). Finally, the local control action is computed using the locally assembled SLS controller \eqref{eq:local-controller} (line \ref{algoline:control}).

\textbf{Local estimate of the global model (line \ref{algoline:assemble}). }
After selecting a local hypothesis model, Subsystem $i$  assembles a local estimate of the ``global'' parameter by collecting the available (delayed) local hypothesis models from its neighbors in $\M{i}$,
\begin{equation}
    \label{eq:delayed-local-model}
\hat{\Theta}^i_t := \left( {\theta}^j_{t-d(j\rightarrow i)}\right)_{j\in \M{i}},
\end{equation}
where the local neighborhood $\M{i}$ (\Cref{defn:M}) represents the the set of neighbors whose model information $i$ needs for synthesizing its local column solution later in \eqref{eq:synth}. 

\textbf{Local column synthesis (line \ref{algoline:local-synth}). }
Analogous to line~\ref{algoline:central-sls-1} in \Cref{alg:centralized}, subsystem $i$ now performs model-based control via SLS by solving the {column} subproblem \eqref{eq:central-column} with additional communication delay and local information constraints based on the locally estimated ``global'' parameter $\Thit$. It is well-established that information constraints described in \Cref{sec:comm_constraints,sec:localized_control} becomes convex sparsity constraints on $\bPX{}$ and $\bPU{}$ \cite{wang2019system}. In particular, these information constraints can be represented as binary matrices $\mathcal{C}^k$ (for delay) and $\mathcal{C}^{\bar{d}}$ (for local information) with $k \in [H]$. Now, the column subproblem for subsystem $i$ changes from \eqref{eq:central-column} to 
\begin{subequations}
    \label{eq:synth}
    \begin{alignat}{3}
        &\min_{\colitx, \, \colitu}  && \quad
        \sum_{k=0}^\infty\left\|
\begin{bmatrix} Q^{1/2} & 0\\ 0& R^{1/2}\end{bmatrix} \begin{bmatrix} \phi^{i,x}_t[k] \\ \phi^{i,u}_t[k] \end{bmatrix} \right\|_F \label{eq:sls-cost} \\
        &\text{s.t.} && \lcolitx[k+1] = \A{\Thit}\lcolitx[k] +  \B{\Thit}\lcolitu[k]\,, \quad \text{ for } k = 0,1,\dots,H-1\label{eq:characterization2}\\
        & &&   \lcolitx[0] = e_i, \quad \lcolitx[H]=0 \label{eq:characterization1}\\
        & && \lcolitx[k],\,\, \lcolitu[k] \,\in \,\mathcal{C}^k(:,i) \cap  \mathcal{C}^{\bar{d}}(:,i)\,,\quad \text{ for } k = 0,1,\dots,H-1. \label{eq:comm-constraint}
    \end{alignat}
\end{subequations}
where \eqref{eq:sls-cost}-\eqref{eq:characterization2} are the same LQR cost and closed-loop response characterization in \eqref{eq:central-column}. The communication and local information constraints are introduced via \eqref{eq:comm-constraint}. We refer interested readers to \cite{wang2016localized, yu2021localized} for a standard derivation on how \eqref{eq:comm-constraint} is equivalent to the information constraints specified in \Cref{sec:comm_constraints,sec:localized_control}. The problem \eqref{eq:synth} is always feasible due to Assumption \ref{assump:controllable} and \ref{assump:feasibility}. 

Delay in the local parameter information results in differently synthesized columns of different $\bPX{}, \bPU{}$ for different subsystems. This contrasts \Cref{alg:centralized} where all subsystems use the same global model as input to the local synthesis problems and output a column of the same $\bPX{}, \bPU{}$. 

\textbf{Asynchronous closed-loop response assembly (line~\ref{algoline:assemble2}). }
Once local closed-loop columns are synthesized, subsystem $i$ has to assemble other relevant columns from subsystem $j$ from $\din{i}$ in order to perform the downstream task of local control action computation via the local version of the SLS controller \eqref{eq:central-controller}, shown in \eqref{eq:local-controller}. In particular, \eqref{eq:local-controller} requires the $i^{\text{th}}$ element of every column $j$ such that $\mathcal{C}^{\bar{d}}(i,j) \not = 0$. By definition, $\din{i}$ (\Cref{def:neighbors}) is the set of $j$'s such that $\mathcal{C}^{\bar{d}}$ has nonzero $(i,j)^{\text{th}}$ element. Thus, only closed-loop columns from $j \in \din{i}$ are required. The assembled closed-loop responses for each subsystem has asynchronous columns with varying delays. 

\textbf{Local Control Action Computation (line~\ref{algoline:control}). }
The final step in CONTROL is to compute a local control action, where each subsystem $i$  plugs the assembled closed-loop responses into the SLS controller \eqref{eq:central-controller}. Due to the sparsity constraints (from information constraints) enforced on the column solutions during the synthesis \eqref{eq:synth}, the matrix-vector computation in \eqref{eq:central-controller} does not require the entire network's delayed column solution. Instead, subsystem $i$ computes a local version of \eqref{eq:central-controller},
\begin{subequations}
    \label{eq:local-controller}
        \begin{align}
        \hat{w}^i(t) &= x^i(t) - \sum_{j \in \din{i}}\sum_{k = 1}^{H-1} \lcoljtx[k](i)\cdot \hat{w}^j(t-k)  \label{eq:sls-1}\\
        u^i(t) &=  \sum_{j \in \din{i}}\sum_{k = 0}^{H-1}  \lcoljtu[k](i) \cdot \hat{w}^j(t-k), \label{eq:sls-2}
        \end{align}
\end{subequations}
where $x^i(t), u^i(t), \hat{w}^i(t) \in \mathbb{R}$ are the local state, control action, and estimated disturbance respectively. The local controllers are initiated with $\hat{w}^i(0) = x^i(0)$.
Similar to the global controller \eqref{eq:central-controller}, the intuition behind \eqref{eq:local-controller} is that each subsystem $i$ \textit{counterfactually} 
assumes that the global closed loop of \eqref{eq:global_sys} behaves exactly as the columns $\coljt$ prescribe. In particular, the $i^{\text{th}}$ position of the $j$th column solution $\bm{\lcoljt}$ maps the $j^{\text{th}}$ position of $\mathbf{w}$ ($\mathbf{w}^j$) to the $i^{\text{th}}$ position of $\mathbf{x}$ and $\mathbf{u}$ ($\mathbf{x}^i$ and $\mathbf{u}^i$). 
Therefore, \eqref{eq:sls-1} estimates the local disturbances by comparing  observed local state $x^i(t)$ and the counterfactual state computed with $\coljt$'s. Then \eqref{eq:sls-2} acts upon the computed disturbance.  

In this step, the errors caused by the delayed information propagate further during \eqref{eq:local-controller} when each subsystem computes control action using the assembled closed-loop column solutions from different sets of sub-controllers in \eqref{eq:delayed-local-model}. 
This contrasts the setting in \Cref{alg:main}, where without communication delay, all subsystems use the globally agreed closed-loop operators $\bPX{}$,$\bPU{}$ to compute the local control action using \eqref{eq:central-controller}. 

Thanks to \eqref{eq:comm-constraint}, regardless of the delay, all closed-loop columns has the correct sparsity required by the communication and locality constraints. 
Consequently, any assembled closed loop columns used for \eqref{eq:local-controller} at each subsystem preserve the required sparsity. Therefore, the SLS controller implemented with these column solutions conforms to the information constraints.

\subsection{Stability guarantee}
\label{sec:distributed_stability}
We now present the main result of this paper.  This is the first stabilization result for a distributed policy (\Cref{alg:main}) in a networked setting with unknown dynamics, communication delay, local information constraint and adversarial disturbances.
\begin{theorem}[Stability]
    \label[theorem]{thrm:main}
    Under Assumptions \ref{assump:noise}-\ref{assump:feasibility}, Algorithm \ref{alg:main} guarantees the ISS of the closed loop of \eqref{eq:global_sys} such that for all $t \geq t_0$, 
    $$ \max\{\|x(t)\|_\infty , \|u(t)\|_\infty\} \leq  \bigO \left( e^{(\bar{n})^{9/2}\bar{d}}\right) \left( e^{- (t-t_0)/H} x(t_0) + \sup_{t_0\leq k\leq t} \infnorm{w(k)} \right) \,,$$
    where $x(t_0)$ is the initial condition, local dimension $\bar{n} = \max\{\|\mathcal{C}^{\bar{d}} \|_1,\, \|\mathcal{C}^{\bar{d}} \|_\infty,\, \max_j \abs{\M{j}} \} $ represents the total state dimension in the $\bar{d}$-neighborhood specified by the dynamics interaction \eqref{eq:local_sys} and the communication graph $\mathcal{G}^C$. Parameter $\bar{d}$ is the largest local delay each subsystem allows for delayed information, and $H$ is the SLS closed-loop response finite impulse horizon.
\end{theorem}
\Cref{thrm:main} highlights that only the local constants $\bar{d}$ and $\bar{n}$ impact the stability guarantee, in contrast to the dependence on the global network dimension in \Cref{alg:centralized} and in system-identification based approaches \cite{chen2021black}.  Further, the result makes explicit that communication delay adds an exponential factor of error on the state deviation from the desired steady state compared to \Cref{thrm:global}. 
When the network connectivity is sparse, local constants $\bar{n}$ and $\bar{d}$ can remain small even if the number of subsystems in the network is large and growing \citep{yazdanian2014distributed, wang2016localized}.  

\textbf{Proof Outline.} The proof of \Cref{thrm:main} follows a similar structure as that of \Cref{thrm:global}. We defer formal proofs to \Cref{sec:proof}. The main challenge here is to characterize the error caused by asynchronous information at different subsystems throughout the algorithm due to delay. 

To begin, we use \Cref{lem:closed-loop} and show that despite the fact that each subsystem in Algorithm \ref{alg:main} uses differently delayed information to compute the local parameter, sub-controller, and control actions, the closed loop for the global system under such distributed policy can be characterized with a simple global representation. In particular, denote the actual closed-loop response implemented by \Cref{alg:main} as $\bPX{t}$, $\bPU{t}$. By observation, each element of $\PX{t}[k]$, $\PU{t}[k]$ is
$$\PX{t}[k](i,j) := \lcoljtx[k](i), \quad \PU{t}[k](i,j) := \lcoljtu[k](i) \,.$$
Therefore, the closed loop of \eqref{eq:global_sys} under \Cref{alg:main} can be characterized by \eqref{eq:cl} with $\bPX{t}$, $\bPU{t}$. It follows from \Cref{lem:sufficient} that as long as the error term 
    \begin{align}
    \label{eq:error-2}
        \sum_{t=1}^\infty  \sum_{k=1}^H \infnorm{ A(\Theta_t) \PX{t-1}[k-1] + B(\Theta_t) \PU{t-1}[k-1] - \PX{t}[k]}  \, 
    \end{align}
is bounded, then the closed loop is ISS. Here $\Theta_t$ is the consistent global model constructed from the \emph{local} consistent hypothesis models selected by all subsystems at time $t$. In \Cref{sec:proof}, we quantify the effect of delay that manifests in $\bPX{t}$ and $\bPU{t}$.

To bound \eqref{eq:error-2}, we extend the perturbation bound in \Cref{thrm:sensitivity} to accommodate the additional sparsity constraints in \eqref{eq:synth} (\Cref{cor:sls-sensitivity}). This result allows us to make a connection between \eqref{eq:error-2} and the total path length of each subsystem;s local parameter selection. Furthermore, \Cref{cor:sls-sensitivity} has potential application for a class of SLS-based distributed and localized MPC problems \cite{alonso2021data, sieber2021system}.

\section{Numerical examples}
\label{sec:simulation}
The main contribution of this work focuses on deriving a stability guarantee for the proposed method under adversarial disturbances and information constraints. In this section, we provide a preliminary numerical exploration of the performance improvement of our approach compared the state-of-the-art adversarial control method in the single-agent case in \Cref{sec:single-agent}. We further test our method on a mesh network of discretized swing dynamics for power systems, where we demonstrate near-optimal performance of \Cref{alg:centralized} and \Cref{alg:main} compared to the offline optimal controller synthesized according to the true dynamics in \Cref{sec:multi-agent}. Further, we study the effect of the localization parameter and the network size under correlated Gaussian noise. 
\subsection{Single-agent: Double integrator dynamics}
\label{sec:single-agent}

We consider the classic double integrator dynamics \cite{recht2019tour}, 
\begin{equation*}
    \begin{bmatrix} x^1 \\  x^2 \end{bmatrix}(t+1) = \begin{bmatrix} 1 & 1 \\ 0 & 1\end{bmatrix} \begin{bmatrix} x^1 \\  x^2 \end{bmatrix}(t) + \begin{bmatrix} 0 \\ 1\end{bmatrix} u(t) +  \begin{bmatrix} w^1 \\  w^2 \end{bmatrix}(t),
\end{equation*}
where $x(t) = [x^1, x^2]^T(t) \in \mathbb{R}^2$, $u(t) \in \mathbb{R}$.  Disturbance $w(t)\in \mathbb{R}^2$ is the bounded ($\infnorm{w(t)} \leq 1$). The system models a unit mass vehicle with position ($x^1$) and velocity ($x^2$) as its state under force $u$. 

To the best of our knowledge, the only online algorithm that guarantees stability under bounded adversarial disturbances is \cite{chen2021black}, where system identification is performed before a certainty-equivalent controller is synthesized based on the estimated dynamics. Therefore, we study the performance of our algorithm and that of \cite{chen2021black}. The results are summarized in \Cref{fig:sysid_main}, where we report the averaged maximum and top $90\%$ state deviation from origin, i.e. $ \max_t \|x(t)\|_\infty$ across 10 runs under three different disturbance profiles. In particular, we generate correlated (across coordinates) Gaussian noise projected to $-1$ and $1$, the uniform disturbance, and the projected state-dependent adversarial disturbance, where the adversary chooses $w(t) = \text{sign}\left(A(\Theta^\star) x(t) + B(\Theta^\star) u(t)\right)$.
 
To instantiate \cite{chen2021black}, we use exact system theoretical constants required for the algorithm and perform the black-box system identification algorithm in \cite[Algorithm 2]{chen2021black} with identification accuracy set to be $10^{-2}$ (largest error tolerable by the algorithm). Then, we generate a stabilizing controller with \cite[Algorithm 3]{chen2021black}. For the proposed approach, we use the optimal LQR feedback gain in place of the centralized SLS controller \eqref{eq:central-sls} and \eqref{eq:central-controller}, since under Assumption \ref{assump:controllable}, the SLS controller synthesized under with LQR cost is equivalent to the optimal LQR feedback \cite{anderson2019system}.  
We remark that for all disturbance profiles and regardless of the choice of stabilizing controller, the system identification algorithm of \cite{chen2021black} always requires control inputs in the order of $10^{11}$. Therefore, across all disturbances, the trajectories generated by \cite{chen2021black} are nearly identical.

\subsection{Multi-agent:Discretized swing dynamics in a power system}
\label{sec:multi-agent}
We now consider a power network with randomly generated sparse edges representing dynamical interactions over a 5 by 5 mesh, where each vertex represents a bus, illustrated in \Cref{fig:mesh}. The local dynamics at bus $i$ is given by the two-state discretized swing equations \cite{anderson2019system},
\begin{align*}
    &x^i(t+1) = \begin{bmatrix}
        1 & \Delta t \\ -\frac{\sum_{j\in \mathcal{N}(i)} k_{ij}}{m_i} \Delta t & 1 \end{bmatrix} x^i(t) + \sum_{j\in\mathcal{N}(i)}\begin{bmatrix}
        0 & 0 \\ -\frac{k_{ij}}{m_i} \Delta t& 0 \end{bmatrix} x^j(t) + \begin{bmatrix} 0\\1\end{bmatrix} \left(u^i(t) + w^i(t)\right)
\end{align*}
where the states are the phase angle (first state) and frequency (second state) deviation from the set point (origin), $\Delta t = 0.1$s is the discretization time step, and $m_i, \, k_{ij},\, u^i,\, w^i, $ are the inertia, line susceptance between bus $i$ and $j$, control action, and external disturbance respectively.   We assume each bus has a phase measurement unit and a frequency sensor to measure $x^i$.

We randomly generate each $k_{ij}$ between $[0.1,\,1]$ and $m_i$ between $[0.1,\,10]$, and assume these parameters are unknown to the algorithm except their bounds. The global network is generated to be open-loop unstable. We use correlated (across buses) Gaussian disturbances with a known bound. In \Cref{fig:swing} we compare the performance of \Cref{alg:centralized} (information shared globally and without delay), \Cref{alg:main}, and the offline optimal distributed SLS controller synthesized from \eqref{eq:central-sls} with the knowledge of $k_{ij}$'s and $m_i$'s, all subject to the same distributed control design requirements. Specifically, the communication network is assume to be the same as the dynamical interaction mesh graph, and we choose the localization parameter to be $\bar{d}=3$, which is much smaller compared to the network size of 25. The centralized algorithm where no communication delay is present matches closely with the trajectory generated by the offline optimal controller, whereas the presence of the information constraints for \Cref{alg:main} degrades the performance. However, we highlight that despite the exponential dependency on the local dimensions in \Cref{thrm:main}, the actual performance of \Cref{alg:main} in this case is significantly better than the theoretical guarantee. 

Furthermore, we compare the effects of different localization parameter choices. On the one hand, larger $\bar{d}$ results in larger worst-case guarantee in \Cref{thrm:main} due to delayed information for local computation. On the other, larger $\bar{d}$ means that each agent in the network can access more (delayed) information. This trade-off manifests on the left of \Cref{tab:varying-d}, where $\bar{d}=5$ appears to achieve lower average state norm over 4 random runs with correlated Gaussian noises, slightly outperforming controllers with $\bar{d}=3$ (too little information) and $\bar{d}=10$ (too much delay from far-away neighbors). On the right of \Cref{tab:varying-d}, we corroborate \Cref{thrm:main} where the stability guarantee only depends on local constants $\bar{d}$ and $\bar{n}$. We randomly generate 3x3, 5x5, and 6x6 mesh networks of similar network structure, and the resulting state norm does not scale with the network size.

\begin{figure}
\centering     
\subfigure[5 by 5 mesh network]{\label{fig:mesh}\includegraphics[scale = 0.35]{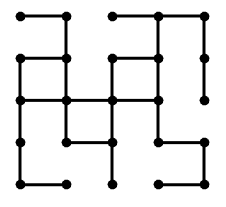}} \quad \quad 
\subfigure[State trajectory of the optimal distributed controller, Algorithm \ref{alg:centralized}, \ref{alg:main}.]
{\label{fig:swing}\includegraphics[scale = 0.32]{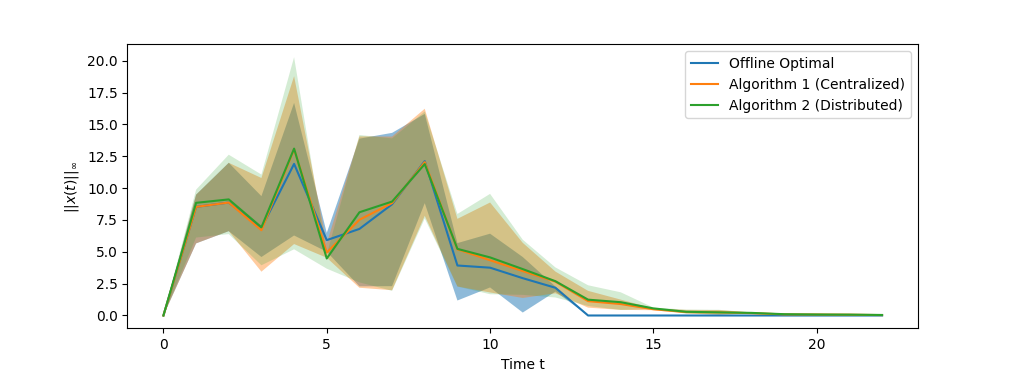}}
\caption{Example networked LTI system with information constraints.}
\end{figure}

\begin{table}
  \caption{Comparison of the state norm ($\infnorm{x(t)}$) for different localization parameters $d$ on the 5 by 5 network (left) and comparison for different network sizes with $N$ agents with fixed localization parameter $d=3$ (right).}
  \label{tab:varying-d}
  {\small \begin{tabular}{ccc|ccc}
    \toprule
    Localization parameter & Mean & Top 95\% &Network Size & Mean & Top 95\%\\
    \midrule
    $\bar{d}=3$ & $3.98$ & $14.02$ & $N=9$ & 2.96 &  10.28 \\
    $\bar{d}=5$ & $3.85$  & $  14.18$ &  $N=25$ & 3.98  &  14.02\\
    $\bar{d}=10$ & 4.19 &    14.08 & $N=36$ &4.27 &   14.05\\
  \bottomrule
\end{tabular}}
\end{table}

\section{Concluding Remarks}
In this work, we propose the first learning-based algorithm that provably achieves online stabilization for networked LTI systems subject to communication delays under adversarial disturbances. 
We leverage nested convex body chasing and distributed control. The novel approach achieves orders of magnitude of performance improvement over state-of-the-art methods for single-agent systems and handles information delays for networked multi-agent systems. {Since most systems are time-varying in nature, an immediate extension of this work is to combine general convex body chasing and model-based control methods to handle time-varying dynamical systems. 
Future directions include extending the communication model to incorporate stochastic and time-varying delays among agents as well as exploring connections to emerging results in function and body chasing, such as when predictions are available \cite{christianson2022chasing}.}

\begin{acks}
The authors thank Varun Gupta and Yingying Li for helpful discussions as well as the anonymous reviewers for their careful reading of this paper and insightful suggestions. 
This work was supported by the National Science Foundation under grants CNS-2146814, CPS-2136197, CNS-2106403, NGSDI-2105648. 
\end{acks}

\bibliographystyle{ACM-Reference-Format}
\bibliography{misc/reference}

\newpage
\appendix

\section{Notation Summary}

\begin{table}[H]
    \caption{Notations and definitions for the model setup, algorithms, and proofs}
    \label{tab:notations}
    \centering
    \begin{tabular}{p{0.2\linewidth}  p{0.8\linewidth}}
        \toprule
        Notation & Meaning  \\
        \midrule
        $x^i(t)$, $u^i(t)$, $w^i(t)$ & Local state ($\mathbb{R}^{n_i}$), control action ($\mathbb{R}^{m_i}$), and disturbances ($\mathbb{R}^{n_i}$) at subsystem $i$;\\
        $\N(i)$ & Dynamical neighbors of subsystem $i$ where $x^j(t-1)$ affects $x^i(t)$ for $j \in \N(i)$;\\
        $x(t)$, $u(t)$, $w(t)$ & Global state, control action, and disturbance vector concatenated from the local ones in \eqref{eq:local_sys}; \\
        $A^{ij}$, $B^{ij}$ & Local dynamics matrices describing how states and control action of subsystem $j$ affects subsystem $i$ for $j \in \N(i)$ in \eqref{eq:local_sys}; \\
        $A$, $B$ & Concatenated global dynamics matrices from $A^{ij}$'s and $B^{ij}$'s ;\\
        $\theta^i$ & The parameters for the nonzero locations in local dynamics matrices and we write $A^{ij}(\theta^i)$, $B^{ij}(\theta^i)$. In particular, $\theta^i \cap \theta^j = \emptyset$ for all $i \not =j$;\\
        $\Theta$ & The concatenated local parameters for the global dynamics with $\Theta:= \bigcup_{i \in [N]} \theta^i$;\\
        $\mathcal{P}_0$ & The known initial compact convex parameter set where the true dynamics parameter lies; \\
        $\mathcal{G}^C$ & Communication graph defined over system \eqref{eq:global_sys} with vertices $V^C$  corresponding to subsystems and directed edges $ E^C$;\\
        $\mathcal{C}$ & The adjecency matrix of $\mathcal{G}^C$;\\
        $\delay{i}{j}$ & Communication delay from subsystem $i$ to subsystem $j$ defined as the graph distance from $i$ to $j$ according to $\mathcal{G}^C$; \\
        $\din{i}$ & $\bar{d}$-incoming neighbors of subsystem $i$ where $\din{i}:=\{j\in[N]: d(j\rightarrow i )\leq \bar{d}\}$. In particular, $j\in\din{i}$ if $\mathcal{C}^{\bar{d}} (i,j) \not = 0$; \\
        $\dout{i}$ & $\bar{d}$-outgoing neighbors of subsystem $i$ where $\dout{i}:=\{j\in[N]: d(i\rightarrow j )\leq \bar{d}\}$. In particular, $j\in\din{i}$ if $\mathcal{C}^{\bar{d}} (j,i) \not = 0$;\\
        $\M{i}$ & Subsystems whose model information is needed for sub-controller synthesis at subsystem $i$ with Algorithm \ref{alg:main} where $\M{i} = \left\{ \ell \in [N]: j \in \N(\ell) \text{ for some }j \in \dout{i} \right\}$;\\
        $\bar{d}$-neighbor of $i$  & The union of all subsystems in $\din{i}$, $\dout{i}$, $\M{i}$; \\ 
        $\pit$ & Local consistent parameter set constructed by subsystem $i$ at time $t$ with \eqref{eq:local-chasing}; \\
        {$\thit$} & {Local consistent parameter for subsystem $i$ for $A^{ij}$ and $B^{ij}$ constructed with Algorithm \ref{alg:main}};  \\
        $\hat{\Theta}^i_t$ & The assembled local estimate of the "global" parameter where $\hat{\Theta}^i_t := \bigcup_{j\in \M{i}} {\theta}^j_{t-d(j\rightarrow i)}$; \\
        {$\colit$} & {{Local column solutions generated by subsystem $i$ at time $t$ from \eqref{eq:synth}}};  \\
        $\colitx$, $\colitu$ & The $x$ and $u$ components of $\colit$, respectively. They are synthesized from \eqref{eq:synth};\\
        $\Theta_t$ & The collection of all local consistent parameters at time $t$ where ${\Theta}_t = \bigcup_{i=1}^N \thit$;\\
        {$A_t, B_t$, $\tilde{w}(t)$} & {The global consistent matrices $A(\Theta_t), B(\Theta_t)$}, and corresponding admissible disturbance;  \\
        $a^i_t$, $b^i_t$ & The $i^{\text{th}}$ row of $A_t$, $B_t$ respectively;\\
        $\hat{w}(t)$ & {Concatenated global estimated disturbance from $\hat{w}^i(t)$ in \eqref{eq:local-controller}}; \\
        $\bPX{t}$ & {Concatenated global closed loop operators where $\PX{t}[k](i,j) := \lcoljtx[k](i)$ from \eqref{eq:local-controller} };  \\
        $\bPU{t}$ & {Concatenated global closed loop operators where $\PU{t}[k](i,j) := \lcoljtu[k](i)$ from \eqref{eq:local-controller} };  \\
        \bottomrule
    \end{tabular}
\end{table}

\begin{table}[H]
    \caption{Constants used throughout the paper}
    \label{tab:notations}
    \centering
    \begin{tabular}{p{0.2\linewidth}  p{0.8\linewidth}}
        \toprule
        Constants & Meaning  \\
        \midrule
        $N$ & {Number of subsystems in the global dynamics \eqref{eq:global_sys}};\\
        $n_i, m_i$ & {Local state and control action dimension for subsystem $i$ in \eqref{eq:local_sys}};\\
        $n_x, n_u$ & {Global state and control dimension with $n_x = \sum_{i=1}^N n_i$ and $n_u = \sum_{i=1}^N m_i$};\\
        $W$ & The known bound on the true disturbances such that $\infnorm{w(t)} \leq W$;\\
        $\kappa$ & The bound on all possible system matrices where $\norm{A(\Theta)}_2$, $\norm{B(\Theta)}_2 \leq \kappa$ for all $\Theta \in \mathcal{P}_0$;\\
        $\bar{d}$ & The localization parameter such that each subsystem is constrained to only use information from its $\bar{d}$-neighbors in Algorithm \ref{alg:main}; \\
        $\bar{n}$ & The largest total local state dimension for the $\bar{d}$-neighbors of the subsystems where $\bar{n} = \max\{\|\mathcal{C}^{\bar{d}} \|_1,\, \|\mathcal{C}^{\bar{d}} \|_\infty,\, \max_j \abs{\M{j}} \}$ ;\\
        $C$, $\rho $  &  The decay rate for the closed-loop columns $\colit$ synthesized in \eqref{eq:synth} such that $\norm{\lcolit[k]}_2 \leq C \rho^k$; \\
        \bottomrule
    \end{tabular}
\end{table}

\section{Proofs for Section \ref{sec:centralized}}
\label{sec:proof_global}

Below we restate and prove the auxiliary results needed for the proof of \Cref{thrm:global} in \Cref{sec:centralized}.
\begin{lemma}[Closed loop Dynamics]
    \label[lemma]{lem:closed-loop-a}
        The closed loop of \eqref{eq:global_sys} under Algorithm \ref{alg:centralized} is characterized as follows for all time $t \in \mathbb{N}_+$:
    \begin{subequations}
        \begin{align}
            x(t)  &= \sum_{k=0}^{H-1} \PX{t}[k] \hat{w}(t-k), \quad
            u(t) = \sum_{k=0}^{H-1} \PU{t}[k] \hat{w}(t-k) \label{eq:cl-a-1}\\
            \hat{w}(t) &= \sum_{k=1}^H \left( A(\Theta_t) \PX{t-1}[k-1] + B(\Theta_t) \PU{t-1}[k-1] - \PX{t}[k] \right)\hat{w}(t-k) + {w}(t-1). \label{eq:cl-a-2}
        \end{align}
    \end{subequations}
    where $A,\,B$ are the true model parameters from \eqref{eq:global_sys} while $w(t)$ is the true unknown bounded disturbances with $\infnorm{w(t)} \leq W$. The linear causal operators $\bPX{t}$, $\bPU{t}$ are synthesized via \eqref{eq:central-sls} based on the selected hypothesis model at $t$ and $\hat{w}(t)$ is the estimated disturbance from the SLS controller \eqref{eq:central-controller}.  
\end{lemma}
\begin{proof}
    First, we write out the global closed-loop dynamics of \eqref{eq:global_sys} under the SLS controller \eqref{eq:central-controller} with the synthesized closed-loop responses,
    \begin{subequations}
        \begin{align}
            x(t)  &= A\left(\Theta^\star\right) x(t-1) + B\left(\Theta^\star\right) u(t-1) + w(t-1) \label{eq:cl1}\\
            \hat{w}(t) &= x(t) - \sum_{k = 1}^{H-1} \Phi^x_t[k] \hat{w}(t-k) \label{eq:cl2} \\
            u(t) &=  \sum_{k = 0}^{H-1} \Phi^u_t[k] \hat{w}(t-k), \label{eq:cl3}
        \end{align}
    \end{subequations}
    where \eqref{eq:cl1} is the global dynamics \eqref{eq:global_sys} while \eqref{eq:cl2} and \eqref{eq:cl3} are the implemented SLS controller. 
    Now, we use the \text{consistency} property of all the consistent hypothesis model $\Theta_t$ selected by \Cref{alg:centralized} and represent dynamics \eqref{eq:cl1} in terms of the global consistent parameter $A_t:=A(\Theta_t), B_t := B(\Theta_t)$,
    \begin{equation}
        \label{eq:cl1-1}
        x(t)  = A_t x(t-1) + B_t u(t-1) + \tilde{w}(t-1),
    \end{equation}
    with admissible consistent disturbances $\infnorm{\tilde{w}(t)} \leq W$ for all time $t$. The replacement of ($A\left( \Theta^\star\right)$, $B\left( \Theta^\star\right))$, $w(t)$)  with ($A_t$, $B_t$, $\tilde{w}(t)$) is by definition of the consistent set \eqref{eq:central-set}. 
    Next, observe that by moving $x(t)$ to the left side, \eqref{eq:cl2} becomes:
    \begin{align}
        x(t)  &= \sum_{k = 1}^{H-1} \Phi^x_t[k] \hat{w}(t-k) + \hat{w}(t)  \nonumber \\
        & = \sum_{k = 0}^{H-1} \Phi^x_t[k] \hat{w}(t-k)        \,, \label{eq:cl4}
    \end{align}    
    where in the last equality we used the fact that each $\PX{t}[0] = I$ by the constraint \eqref{eq:feasibility}. 
    Now we substitute \eqref{eq:cl1-1} into \eqref{eq:cl2} to get 
    \begin{subequations}
    \begin{align}
        \hat{w}(t) &= x(t) - \sum_{k = 1}^{H-1} \Phi^x_t[k] \hat{w}(t-k)\\
        &= A_t x(t-1) + B_t u(t-1) - \sum_{k = 1}^{H-1} \Phi^x_t[k] \hat{w}(t-k) + \tilde{w}(t-1)\\
        &= A_t \sum_{k = 0}^{H-1} \PX{t-1}[k] \hat{w}(t-1-k) + B_t \sum_{k = 0}^{H-1} \PU{t-1}[k] \hat{w}(t-1-k) - \sum_{k = 1}^{H-1} \Phi^x_t[k] \hat{w}(t-k)  \nonumber  \\
        &\qquad + \tilde{w}(t-1)  \label{eq:18c}\\
        &= \sum_{k = 1}^{H-1} \left( A_t\PX{t-1}[k-1] +  B_t\PU{t-1}[k-1] - \PX{t}[k]  \right)\hat{w}(t-k) \nonumber \\
        &\qquad +\left( A_t \PX{t-1}[H-1] + B_t \PU{t-1}[H-1] - \PX{t-1}[H]\right)\hat{w}(t-H) + \tilde{w}(t-1) \label{eq:18d} \\
        &= \sum_{k = 1}^{H} \left( A_t\PX{t-1}[k-1] +  B_t\PU{t-1}[k-1] - \PX{t}[k]  \right)\hat{w}(t-k)  + \tilde{w}(t-1), \label{eq:cl5}
    \end{align}   
    \end{subequations}
    where in \eqref{eq:18c} we substituted \eqref{eq:cl4} and \eqref{eq:cl3} into $x(t-1)$ and $u(t-1)$ respectively. In \eqref{eq:18d}, we grouped the terms according to $\hat{w}(t-k)$ and used the fact that the closed-loop responses are synthesized in \eqref{eq:central-sls} such that $\PX{t-1}[H] = 0$ for all $t$. 
    Together, \eqref{eq:cl3},\eqref{eq:cl4}, and \eqref{eq:cl5} are as requested.
\end{proof}

\begin{lemma}[Sufficient condition for $H$-convolution ISS]
    \label[lemma]{lem:sufficient-a}
    Let $H \in \mathbb{N}$. 
    For $k \in [H]$, let $\left\{a_{t}[k]\right\}_{t=1}^\infty$ and $\left\{w_{t}\right\}_{t=1}^\infty$ be positive sequences.
    Let $\{s_t\}_{t=0}^\infty$ be a positive sequence such that
    \begin{align}
    s_t \leq \sum_{k=1}^H a_{t-1}[k]\cdot s_{t-k} + w_{t-1} \,.\label{eq:seq-recursion}
    \end{align}
    Then $\{s_t\}_{t=0}^\infty$ is bounded if $\sum_{t=0}^\infty \sum_{k=1}^H a_{t}[k] \leq L$ for some $L \in \mathbb{R}_+$. In particular, for all $t\geq t_0$,
    $$s_t \leq e^{-(t-t_0)/H} \cdot e^L s_{t_0} + \frac{ \left( e^L + e -1 \right)}{e-1} \sup_{t_0\leq k< t} w_k \, .$$
\end{lemma}
\begin{proof}
Fix $t_0$ and $t \geq t_0$. Denote $\{z_{t_i}\}$ as a finite subsequence of $\{s_{\tau}\}_{\tau = t_0}^t$ such that 
\begin{align*}
    z_{t_N} &= s_t\\
    z_{t_{i-1}} &= \max_{t_i - H \leq \tau \leq t_i - 1} s_{\tau} \,, \quad \text{ for } i = N, N-1,\dots,1,
\end{align*}
with $t_N = t$ and $z_{t_i} = s_{t_i}$.  
This construction of the $\{z_{t_i}\}$ has to terminate at $z_{t_0} = s_{t_0}$. Therefore, $N$ is at least $\frac{(t-t_0)}{H}$ and at most $t-t_0$. By the recursive relationship of $s_t$ in \eqref{eq:seq-recursion}, we have for any $i$,
\begin{align}
    z_{t_i} = s_{t_i} &\leq \sum_{k=1}^H a_{t_i-1}[k] s_{t_i-k} + w_{t_i-1} \nonumber \\
    &\leq \left( \sum_{k=1}^H a_{t_i -1}[k] \right) z_{t_i -1} + w_{t_i-1} \nonumber \\
    &= \hat{a}_{t_i-1} \cdot z_{t_i -1 } + w_{t_i-1}, \label{eq:z-recursion}
\end{align}
where we use the fact that $a_t[k]\geq 0$ for all $t$ and $k$. We also denote $\hat{a}_{t_i -1} = \left( \sum_{k=1}^H a_{t_i -1}[k] \right) $ for the last equality. By the recursion \eqref{eq:z-recursion}, we have 
\begin{align}
    s_t = z_{t_N}
    \leq \prod_{i=1}^N \hat{a}_{t_i-1} \cdot z_{t_0} + \left(\sup_{t_0\leq k <t} w_k\right) \left(  1+ \sum_{j=1}^N \prod_{i=j}^N \hat{a}_{t_i -1}   \right) \label{eq:s-last-step}
\end{align}
Now, $\prod_{i=j}^N \hat{a}_{t_i - 1} = \prod_{i=j}^N \left(\left( \hat{a}_{t_i - 1}-1\right)  + 1\right)  \leq \prod_{i=j}^N e^{\hat{a}_{t_i -1} - 1} = e^{\sum_{i=j}^N \left(\hat{a}_{t_i -1} -1\right)} \leq e^{L-(N-j+1)}$, where the last inequality is due to the hypothesis that $\sum_{t=0}^\infty \hat{a}_t \leq L$. Plug this inequality for $\prod_{i=j}^N \hat{a}_{t_i - 1} $ back to \eqref{eq:s-last-step}, we continue with
\begin{align*}
     s_t &\leq e^{-(t-t_0)/H} \cdot s_{t_0}e^L + \left(\sup_{t_0\leq k <t} w_k\right)  \left(1+\sum_{j=1}^N e^{L-(N-j)}\right)\\
     &\leq e^{-(t-t_0)/H} \cdot s_{t_0}e^L + \left(\sup_{t_0\leq k <t} w_k\right)  \left(1 + e^L \sum_{j=0}^{N-1} e^{-j}\right)\\
     &\leq e^{-(t-t_0)/H} \cdot s_{t_0}e^L + \left(\sup_{t_0\leq k <t} w_k\right)  \left(1 + e^L \frac{1}{e-1}\right)\,,
\end{align*}
where we used $z_{t_0} = s_{t_0}$ and that $N$ is at least $(t-t_0)/H$. This is the required bound, which holds for any $t, t_0 \in \mathbb{N}$.
\end{proof}


\section{Proof of Theorem \ref{thrm:main}}
\label{sec:proof}

\begin{theorem}[Stability, Scalar Subsystems]
    Under Assumptions \ref{assump:noise}-\ref{assump:feasibility}, Algorithm \ref{alg:main} guarantees the ISS of the closed loop of \eqref{eq:global_sys} with 
    $$  \max\{\|x(t)\|_\infty , \|u(t)\|_\infty\} \leq  \bigO \left( e^{(\bar{n})^{9/2}\bar{d}}\right) \left( e^{- (t-t_0)/H} x(t_0) + \sup_{t_0\leq k\leq t} \infnorm{w(k)} \right) \,,$$
    where $x(t_0)$ is the initial condition, local dimension $\bar{n} = \max\{\|\mathcal{C}^{\bar{d}} \|_1,\, \|\mathcal{C}^{\bar{d}} \|_\infty,\, \max_j \abs{\M{j}} \} $ represents the total state dimension in the $\bar{d}$-neighborhood specified by the dynamics interaction \eqref{eq:local_sys} and the communication graph $\mathcal{G}^C$. Parameter $\bar{d}$ is the largest local delay each subsystem allows for delayed information, and $H$ is the SLS closed-loop response finite impulse horizon.
\end{theorem}
\begin{proof}
    \label[Proof]{prof:main}
We first characterize the closed loop dynamics of \eqref{eq:global_sys} under Algorithm \ref{alg:main}. In particular, despite the fact that each subsystem uses differently delayed information to compute the local parameter, column solutions to the closed-loop responses, and control actions, the closed loop for the global system under such distributed policy can be simply characterized as 
\begin{subequations}
    \begin{align}
        x(t)  &= \sum_{k=0}^{H-1} \PX{t}[k] \hat{w}(t-k), \quad
        u(t) = \sum_{k=0}^{H-1} \PU{t}[k] \hat{w}(t-k) \label{eq:cl-xu}\\
        \hat{w}(t) &= \sum_{k=1}^H \left( A_t \PX{t-1}[k-1] + B_t \PU{t-1}[k-1] - \PX{t}[k] \right)\hat{w}(t-k) + \tilde{w}(t-1),\label{eq:cl-w}
    \end{align}
\end{subequations}
by \Cref{lem:closed-loop-a}. Here $u(t),\,\hat{w}(t)$ are concatenated control action and estimated disturbance from \eqref{eq:local-controller}. $A_t, B_t$ are the global consistent parameter concatenated with the local consistent parameters $A^{ij}(\thit), B^{ij}(\thit)$. Vector $\tilde{w}(t)$ are the admissible consistent disturbances corresponding to $A_t$, $B_t$ with the property that $\infnorm{\tilde{w}(t)}\leq W$ for all time $t$. Operators $\bPX{t}, \bPU{t}$ are shorthand for global closed-loop operators when \eqref{eq:local-controller} is implemented, with $$\PX{t}[k](i,j) := \lcoljtx[k](i), \quad \PU{t}[k](i,j) := \lcoljtu[k](i) \,.$$


We follow similar procedure in the proof of \Cref{thrm:global} and bound $\infnorm{\hat{w}(t)}$ from \eqref{eq:cl-w} by examining the following dynamical evolution,
\begin{align}
    \infnorm{\hat{w}(t)} 
    &\leq  \sum_{k=1}^H \infnorm{ A_t \PX{t-1}[k-1] + B_t \PU{t-1}[k-1] - \PX{t}[k]} \infnorm{\hat{w}(t-k)} + \infnorm{\tilde{w}(t-1)}  \,. \label{eq:middle-w-hat}
 \end{align}
By \Cref{lem:sufficient-a}, as long as $\sum_{t=1}^\infty  \sum_{k=1}^H \infnorm{ A_t \PX{t-1}[k-1] + B_t \PU{t-1}[k-1] - \PX{t}[k]} \leq L$ for some positive constant $L$, then we can bound \eqref{eq:middle-w-hat} with 
\begin{align}
    \infnorm{\hat{w}(t)} &\leq  e^{-(t-t_0)/H} \cdot e^L x(t_0) + \sup_{t_0 \leq k <t}\infnorm{\tilde{w}(t)} \frac{ \left( e^L + e -1 \right)}{e-1}  .\nonumber
\end{align}
Therefore, what's left is to show $$\sum_{t=1}^\infty  \sum_{k=1}^H \infnorm{ A_t \PX{t-1}[k-1] + B_t \PU{t-1}[k-1] - \PX{t}[k]}\leq L \, ,$$ 
which is proved in \Cref{prop:bounded-error-a} where $L = \bigO \left( \text{poly}\left(\bar{n}\right)\bar{d}\right)$. This concludes the proof.
\end{proof}

\begin{lemma}[Bounded error for closed loop operators]
    \label[proposition]{prop:bounded-error-a}
    Let $\bPX{t}, \bPU{t}$ denote the global closed loop operators concatenated from sub-controllers generated with Algorithm \ref{alg:main} where $\PX{t}[k](i,j) := \lcoljtx[k](i)$ and $\PU{t}[k](i,j) := \lcoljtu[k](i)$. Denote matrices $A_t, B_t$ as the global consistent parameter concatenated with local consistent parameters $A^{ij}(\thit), B^{ij}(\thit)$. Then we have
    \begin{align}
        \sum_{t=1}^\infty  \sum_{k=1}^H &\infnorm{ A_t \PX{t-1}[k-1] + B_t \PU{t-1}[k-1] - \PX{t}[k]} \label{eq:error1}\\
        &\leq  (\bar{d} + 3 ) \bar{n}^3 \text{diam}(\mathcal{P}_0)  \left(\kappa \bar{n}^{\frac32} \Gamma H + \frac{ C}{1-  \rho} \right) \, , \nonumber
    \end{align}
    where $\bar{n} = \max\{\|\mathcal{C}^{\bar{d}} \|_1,\, \|\mathcal{C}^{\bar{d}} \|_\infty,\, \max_j \abs{\M{j}} \} $, and $\bar{d}$ is the largest local delay each subsystem considers for the algorithm, while $H$ is SLS controller horizon. Here, $\Gamma$ is a system-theoretical constant that does not depend on the global dynamics properties detailed in \Cref{thm:mainsensitivity_global}.
\end{lemma}
\begin{proof}
    To ease notation, we use $a^i_t$ and $b^i_t$ to denote the $i^{\text{th}}$ row of $A_t$ and $B_t$ respectively.
    
    Our strategy is to bound each term in \eqref{eq:error1} for a fixed $t$ and $k$. We will see that the summation of these terms over all $k$ and $t$ remain bounded. Each term in \eqref{eq:error1} can be bounded as follows. 
        \begin{align}
            &\infnorm{ A_t \PX{t-1}[k-1] + B_t \PU{t-1}[k-1] - \PX{t}[k]} \nonumber \\
            =& \max_{i \in [N]} \sum_{j\in\din{i}} \abs{ \trans{a^i_t} \PX{t-1}[k-1](:,j) + \trans{b^i_t} \PU{t-1}[k-1](:,j) - \underbrace{\PX{t}[k](i,j)}_{\text{Defined to be } \lcoljtx[k](i)  }  } \,.\label{eq:tail1}
        \end{align}
   Due to the sparsity constraints that correspond to the information constraints placed on the closed-loop responses during synthesis \eqref{eq:synth}, the only nonzero elements in a particular row $i$ of $\PX{t}[k]$ are the positions at $j \in \din{i}$. Hence, we can write  sum of each row $i$ as sum of the elements in position $(i,j)$ where $j \in \din{i}$ in \eqref{eq:tail1}. 
    Recall that $\lcoljtx$ are synthesized in \eqref{eq:synth} such that 
    \begin{align}
        \label{eq:coljti}
        \lcoljtx[k](i) = &\trans{a^i_{t - \delay{j}{i} - \delay{i}{j}}}\lcoljtx[k-1] \nonumber \\ 
        &+\trans{b^i_{t - \delay{j}{i} - \delay{i}{j}}}\lcoljtu[k-1]
    \end{align}
    because $\lcoljtx$  is synthesized by $j$ at time $t-\delay{j}{i}$.  The $i^{\text{th}}$ position of $\lcoljtx$ in particular uses model information from subsystem $i$, which is transmitted to $j$ from $i$ with delay $\delay{i}{j}$. Therefore, we substitute \eqref{eq:coljti} into \eqref{eq:tail1} to get
        \begin{align}
            \eqref{eq:tail1} =& \max_{i \in [N]} \sum_{j\in\din{i}} \Bigg| \trans{a^i_t} \PX{t-1}[k-1](:,j) - \trans{a^i_{t - \delay{j}{i} - \delay{i}{j}}}\lcoljtx[k-1] \nonumber \\
            &+ \trans{b^i_t} \PU{t-1}[k-1](:,j) -  \trans{b^i_{t - \delay{j}{i} - \delay{i}{j}}}\lcoljtu[k-1]  \Bigg| \label{eq:tail2}
        \end{align}
    Adding and subtracting $\trans{a^i_t}\lcoljtx[k-1] $ and $\trans{b^i_t}\lcoljtu[k-1] $ in \eqref{eq:tail2}, we can group terms and get 
    \begin{subequations}
        \label{eq:tail3}
        \begin{align}
            \eqref{eq:tail2} &\leq \max_{i \in [N]} \sum_{j\in\din{i}} \Bigg| \trans{a^i_t} \left(\PX{t-1}[k-1](:,j) - \lcoljtx[k-1]\right) \nonumber \\
            & \quad \quad \quad \quad \quad \,\,\,\,\,\,\,+\trans{b^i_t} \left(\PU{t-1}[k-1](:,j) - \lcoljtu[k-1]\right)\Bigg| \label{eq:group1}\\
            &+  \max_{i \in [N]}\sum_{j\in\din{i}} \Bigg|  \trans{a^i_t - {a^i_{t - \delay{j}{i} - \delay{i}{j}}}} \lcoljtx[k-1] \nonumber \\
            &\quad \quad \quad \quad \quad \,\,\,\,\,\,\,+  \trans{b^i_t - {b^i_{t - \delay{j}{i} - \delay{i}{j}}}} \lcoljtu[k-1] \Bigg|. \label{eq:group2}
        \end{align}            
    \end{subequations}
    We now consider \eqref{eq:group1} and \eqref{eq:group2} separately. For the remainder of the proof, we use $\coljtx$ and $\coljtu$ as shorthand for the $j$th column of $\bPX{t}$ and $\bPU{t}$ respectively. Apply Cauchy-Schwarz, 
    \begin{subequations}
        \begin{align}
            \eqref{eq:group1} &\leq \max_{i \in [N]} \sum_{j\in\din{i}} \norm{a^i_t}_2 \norm{\phi^{j,x}_{t-1}[k-1] - \lcoljtx[k-1]}_2 \nonumber \\
            & \quad \quad \quad \quad \quad \,\,+ \norm{b^i_t}_2 \norm{\phi^{j,u}_{t-1}[k-1] - \lcoljtu[k-1]}_2 \\
           \text{(by Assumption \ref{assump:compact})}\quad &\leq \kappa \cdot \max_{i \in [N]} \sum_{j\in\din{i}} \norm{\phi^{j,x}_{t-1}[k-1] - \lcoljtx[k-1]}_2 \nonumber \\
            & \quad \quad \quad \quad \quad \,\,+ \norm{\phi^{j,u}_{t-1}[k-1] - \lcoljtu[k-1]}_2 \\
            &= \kappa \cdot \max_{i \in [N]} \sum_{j\in\din{i}}\left(   \sum_{\ell \in \dout{j}} \abs{\phi^{j,x}_{t-1}[k-1](\ell) - \lcoljtx[k-1](\ell)}^2  \right)^{1/2} \nonumber \\
            &  \quad \quad \quad \quad +  \left(   \sum_{\ell \in \dout{j}} \abs{\phi^{j,u}_{t-1}[k-1](\ell) - \lcoljtu[k-1](\ell)}^2  \right)^{1/2} \label{eq:28c} \\
            &= \kappa \cdot \max_{i \in [N]} \sum_{j\in\din{i}}\left(   \sum_{\ell \in \dout{j}} \abs{\lcoljltx[k-1](\ell) - \lcoljtx[k-1](\ell)}^2  \right)^{1/2} \nonumber \\
            &  \quad \quad \quad \quad +  \left(   \sum_{\ell \in \dout{j}} \abs{\lcoljltu[k-1](\ell) - \lcoljtu[k-1](\ell)}^2  \right)^{1/2},  \label{eq:group1tail1}        
        \end{align}
    \end{subequations}
    where to arrive at \eqref{eq:28c} we used the fact that the nonzero elements in any column/sub-controller synthesized or assembled at subsystem $j$ corresponds to the elements in $\dout{j}$. The last equality comes from the definition of $\bPX{t-1}$,$\bPU{t-1}$. Continuing, we bound any sum using the largest summand multiplied by the number of summands:
    \begin{subequations}
        \begin{align}
            \eqref{eq:group1} \leq \eqref{eq:group1tail1}
            &\leq \kappa \cdot \max_{i \in [N]} \sum_{j\in\din{i}}\left(   \bar{n} \cdot \max_{\ell \in \dout{j}} \norm{\lcoljltx[k-1] - \lcoljtx[k-1]}_2^2  \right)^{1/2} \nonumber \\
            &  \quad \quad \quad \quad +  \left(   \bar{n} \cdot \max_{\ell' \in \dout{j}} \norm{\phi^{j,u}_{t-1-\delay{j}{\ell'}}[k-1] - \lcoljtu[k-1]}_2^2  \right)^{1/2},       \label{eq:grouptail1}    \\
            &=  \kappa {\bar{n}}^{3/2}\cdot \max_{i \in [N]} \max_{j\in\din{i}}\Bigg( \left( \max_{\ell \in  \dout{j}} \norm{\lcoljltx[k-1] - \lcoljtx[k-1]}_2^2  \right)^{1/2}\nonumber\\
            & \quad \quad \quad \quad +  \left( \max_{\ell' \in  \dout{j}} \norm{\phi^{j,u}_{t-1-\delay{j}{\ell'}}[k-1] - \lcoljtu[k-1]}_2^2  \right)^{1/2}   \Bigg). \label{eq:group1tail2}
         \end{align}
    \end{subequations}
    Recall that $\lcoljlt$  are generated by subsystem $j$ using model information $\hat{\Theta}^j_{t-1-\delay{j}{\ell}}$ during synthesis procedure (\eqref{eq:delayed-local-model}, Algorithm \ref{alg:main}). Similarly, $\lcoljt$ are generated using $\hat{\Theta}^j_{t-\delay{j}{i}}$. Therefore, we can invoke \Cref{cor:sls-sensitivity} and arrive at 
        \begin{align}
            \eqref{eq:group1} &\leq \eqref{eq:group1tail1} \leq \eqref{eq:group1tail2} \nonumber\\
            &\leq \kappa {\bar{n}}^{3/2}  \Gamma \max_{i \in [N]} \max_{j\in\din{i}} \Bigg( \left( \max_{\ell \in \dout{j}} \norm{\hat{\Theta}^j_{t-1-d(j\rightarrow \ell)} - \hat{\Theta}^j_{t-\delay{j}{i}}}_F^2  \right)^{1/2} \nonumber \\
            &\quad \quad \quad \quad  \quad \quad \quad \quad \quad \quad + \max_{\ell' \in \dout{j}}\left(\norm{ \hat{\Theta}^j_{t-1-d(j\rightarrow \ell')} - \hat{\Theta}^j_{t-\delay{j}{i}} }_F^2 \right)^{1/2} \Bigg) \label{eq:group1tail3}
        \end{align}
    For any fixed $i$, $j$, $\ell$, $\ell'$, the following holds true.
        \begin{align}
            \eqref{eq:group1tail3} &= \kappa {\bar{n}}^{3/2}  \Gamma \left( \norm{\hat{\Theta}^j_{t-1-d(j\rightarrow \ell)} - \hat{\Theta}^j_{t-\delay{j}{i}}}_F + \norm{ \hat{\Theta}^j_{t-1-d(j\rightarrow \ell')} - \hat{\Theta}^j_{t-\delay{j}{i}} }_F \right) \nonumber\\
            &= \kappa {\bar{n}}^{3/2}  \Gamma \sum_{m\in\M{j}} \norm{\theta^m_{t-1-d(j\rightarrow \ell)-d(m\rightarrow j)} - \theta^m_{t-\delay{j}{i} - \delay{m}{j}}}_F \nonumber\\
            &\quad \quad \quad \quad  + \sum_{m\in\M{j}} \norm{\theta^m_{t-1-d(j\rightarrow \ell')-d(m\rightarrow j)} - \theta^m_{t-\delay{j}{i} - \delay{m}{j}}}_F \nonumber\\
            &\leq \kappa {\bar{n}}^{3/2} \Gamma \sum_{m\in\M{j}} \left(\sum_{p=\min(t_1,t_2)}^{\min (t_1,t_2)+\delta_t +1} \norm{\theta^m_{t-p+1} - \theta^m_{t-p}}_F + \sum_{p=\min(t_1',t_2)}^{\min (t_1',t_2)+\delta_t' +1} \norm{\theta^m_{t-p+1} - \theta^m_{t-p}}_F\right), \label{eq:group1-final}
        \end{align}
    where we define $t_1 = 1+ \delay{j}{\ell} + \delay{m}{j}$, $t_1' = 1+ \delay{j}{\ell'} + \delay{m}{j}$, $t_2 = 1+ \delay{j}{i} + \delay{m}{j}$, and $\delta t = \abs{\delay{j}{i} - \delay{j}{\ell} -1}$, $\delta t' = \abs{\delay{j}{i} - \delay{j}{\ell'} -1}$. We stop at \eqref{eq:group1-final} for the moment for our bound for \eqref{eq:group1} and change course to bound the other term \eqref{eq:group2} in \eqref{eq:tail3}. We start with cauchy-schwarz for \eqref{eq:group2}.
        \begin{align}
            \eqref{eq:group2} &\leq  \max_{i \in [N]}\sum_{j\in\din{i}} \norm{  {a^i_t - {a^i_{t - \delay{j}{i} - \delay{i}{j}}}}}_2 \norm{\lcoljtx[k-1] }_2\nonumber \\
            &\quad \quad \quad \quad \quad \,\,\,\,\,\,\,+ \norm{ {b^i_t - {b^i_{t - \delay{j}{i} - \delay{i}{j}}}}}_2 \norm{\lcoljtu[k-1]}_2\nonumber\\
            &\leq C\rho^{k-1} \bar{n} \cdot  \max_{i \in [N]}\max_{j\in\din{i}} \norm{  {a^i_t - {a^i_{t - \delay{j}{i} - \delay{i}{j}}}}}_2 + \norm{ {b^i_t - {b^i_{t - \delay{j}{i} - \delay{i}{j}}}}}_2 \nonumber \\
            &=  C\rho^{k-1} \bar{n} \cdot  \max_{i \in [N]}\max_{j\in\din{i}} \norm{  {\theta^i_t - {\theta^i_{t - \delay{j}{i} - \delay{i}{j}}}}}_2 \label{eq:group2-final}
        \end{align}
    Here we have used the decay property of the finite-impulse-response closed-loop responses to bound the decay rate of the sub-controllers. The last equality holds by recalling that we have defined $a^i_t$ and $b^i_t$ to be the $i^{\text{th}}$ row of the $A_t$ and $B_t$ respectively, which is constructed from the global consistent parameter $\Theta_t = \cup_{i=1}^N \theta^i_t$. Therefore, by definition, $[a^i_t, b^i_t] = \theta^i_t$.
    
    We now return to bound $\sum_{t=0}^\infty \sum_{k=1}^H\infnorm{ A_t \PX{t-1}[k-1] + B_t \PU{t-1}[k-1] - \PX{t}[k]}$. In particular, we have so far showed that

        \begin{align}
         \label{eq:final}
            \sum_{t=0}^\infty \sum_{k=1}^H\infnorm{ A_t \PX{t-1}[k-1] + B_t \PU{t-1}[k-1] - \PX{t}[k]} \leq \sum_{t=0}^\infty \sum_{k=1}^H \eqref{eq:group1-final} + \eqref{eq:group2-final}. 
        \end{align}
    Therefore, our goal is to bound each component of the right hand side. Specifically, 
        \begin{align}
            &\sum_{t=0}^\infty \sum_{k=1}^H \eqref{eq:group1-final} \nonumber\\
            &\quad \leq  \sum_{t=0}^\infty \sum_{k=1}^H \kappa {\bar{n}}^{3/2} \Gamma \sum_{m\in\M{j}} \left(\sum_{p=\min(t_1,t_2)}^{\min (t_1,t_2)+\delta_t +1} \norm{\theta^m_{t-p+1} - \theta^m_{t-p}}_F + \sum_{p=\min(t_1',t_2)}^{\min (t_1',t_2)+\delta_t' +1} \norm{\theta^m_{t-p+1} - \theta^m_{t-p}}_F\right),
        \end{align}
    for a different tuple of $(i\in[N],j\in\din{i},\ell\in\dout{j},\ell'\in\dout{j})$ at each $t$. However, for any $(i,j,\ell,\ell')$, the following holds.
    \begin{subequations}
        \begin{align}
            &\sum_{t=0}^\infty \sum_{k=1}^H \eqref{eq:group1-final} \nonumber\\
            &\leq \kappa {\bar{n}}^{3/2} \Gamma 
            \sum_{k=1}^H
            \sum_{m\in\M{j}} \left(\sum_{p=\min(t_1,t_2)}^{\min (t_1,t_2)+\delta_t +1} \sum_{t=0}^\infty \norm{\theta^m_{t-p+1} - \theta^m_{t-p}}_F + \sum_{p=\min(t_1',t_2)}^{\min (t_1',t_2)+\delta_t' +1}  \sum_{t=0}^\infty \norm{\theta^m_{t-p+1} - \theta^m_{t-p}}_F\right), \nonumber\\
            &\leq 2\kappa {\bar{n}}^{9/2} \Gamma H  \text{diam}(\mathcal{P}_0) \left( \max_{i\in[N]\,, j\in\din{i},\,\ell\in \dout{j}}  (1+1+|\delay{j}{i})-\delay{j}{\ell} -1 | \right) \label{eq:compete} \\
            &\leq 2\kappa  {\bar{n}}^{9/2} \Gamma H  \text{diam}(\mathcal{P}_0) (\bar{d}+3) \label{eq:group1-end}.
        \end{align}
    \end{subequations}
    Here we have used in the competitiveness of each local Steiner point selector via \eqref{eq:steiner} in \eqref{eq:compete} with competitive ratio of $\bar{n}/2$. Furthermore, by definition of $\din{i}$ and $\dout{j}$, we know that the largest delay for $\delay{j}{i}$ and $\delay{j}{\ell}$ for any choice of $i,j,\ell$ is less than $\bar{d}$.

    Finally, we investigate the second component of the right hand side of \eqref{eq:final}.
    \begin{subequations}
        \begin{align}
            \sum_{t=0}^\infty \sum_{k=1}^H \eqref{eq:group2-final} &= \sum_{t=0}^\infty \sum_{k=1}^H  C\rho^{k-1} \bar{n} \cdot  \max_{i \in [N]}\max_{j\in\din{i}} \norm{  {\theta^i_t - {\theta^i_{t - \delay{j}{i} - \delay{i}{j}}}}}_2 \\
            &\leq \sum_{k=1}^H C \rho^{k-1} \bar{n} \max_{i \in [N]}\max_{j\in\din{i}} \sum_{p=0}^{\delay{j}{i} + \delay{i}{j} +1} \sum_{t=0}^{\infty}\norm{\theta^i_{t-p+1} - \theta^i_{t-p}}_2\\
            &\leq C\bar{n}^3 
            \text{diam}(\mathcal{P}_0)(\bar{d}+1)/(1-\rho) \label{eq:group2-end}
        \end{align} 
    \end{subequations}
    where we once again used the competitive ratio of the local Steiner point selector \eqref{eq:steiner}. Moreover, by definition of $\din{i}$, the largest delay $\delay{i}{j}$ for any $j \in \din{i}$ is less than $\bar{d}$.

    Finally, we have the bound on the target quantity with \eqref{eq:group1-end} and \eqref{eq:group2-end} and conclude
    \begin{align*}
        \sum_{t=1}^\infty  \sum_{k=1}^H \infnorm{ A_t \PX{t-1}[k-1] + B_t \PU{t-1}[k-1] - \PX{t}[k]} &\leq \eqref{eq:group1-end} + \eqref{eq:group2-end}\\
        &\leq  2(\bar{d} + 3 ) \bar{n}^3 \text{diam}(\mathcal{P}_0)  \left(\kappa \bar{n}^{\frac32} \Gamma H + \frac{ C}{1-  \rho} \right) \, . \nonumber
    \end{align*}
\end{proof}


\begin{corollary}[of Theorem \ref{thm:mainsensitivity_global}, Structured SLS sensitivity]
    \label[corollary]{cor:sls-sensitivity}
    Consider the optimal solutions $\phi$, ${\phi}'$ to \eqref{eq:synth} with two different parameters input $\Theta$, ${\Theta'}$ respectively. Then we have 
    \begin{align*}
        \norm{\phi - {\phi'}}_2 \leq \Gamma \norm{\Theta - {\Theta'}}_2,
    \end{align*}
    with $ \Gamma = \bigO\left(\Gamma_A + \Gamma_B \right) $ where $\Gamma_A$ and $\Gamma_B$ are constants in \Cref{thm:mainsensitivity_global}.
\end{corollary}
\begin{proof}
The SLS synthesis problem that we consider in  \eqref{eq:synth} has one additional sparsity constraints than general SLS synthesis presented in \eqref{eq:objcol} to which \Cref{thm:mainsensitivity_global} apples. Therefore, we need to de-constrain the synthesis problem \eqref{eq:synth} and turn it into a problem of the form \eqref{eq:objcol} in order to apply \Cref{thm:mainsensitivity_global}. To do so, we follow the procedure in section IV.A of \cite{yu2021localized}, where a re-parameterization of $\coljtu$ is used to characterize all sparse $\coljtu$ which will result in sparse $\coljtx$ according to the dynamical evolution \eqref{eq:characterization2}. 
First, we rewrite \eqref{eq:characterization2} with the nonzere variables grouped together as follows.
\begin{equation}
\label{eqn:bdy}
\begin{bmatrix}
\tilde{\phi}^{j,x}\\ \tilde{\phi}^{j,x}_b
\end{bmatrix}[k+1]
 = 
\begin{bmatrix}
A^{(j)}_{nn} & A^{(j)}_{nb}\\
A^{(j)}_{bn}  & A^{(j)}_{bb} 
\end{bmatrix}
\begin{bmatrix}
        \tilde{\phi}^{j,x}\\ \tilde{\phi}^{j,x}_b
\end{bmatrix}[k]  + 
\begin{bmatrix}
B^{(j)}_{n} \\
B^{(j)}_{b} 
\end{bmatrix}
\tilde{\phi}^{j,u}[k]
\end{equation}
where $\tilde{\phi}^{j,x}$ denotes the vector of nonzero entries in $\coljtx$ and  $\tilde{\phi}^{j,x}_b$ denotes the ``boundary'' positions of $\tilde{\phi}^{j,x}$. The ``boundary'' positions of $\tilde{\phi}^{j,x}$ corresponds to the positions in the vector that would become nonzero from zero due to the dynamical evolution  \eqref{eq:characterization2} in one time step. We refer interested reader to \cite{yu2021localized} for detailed setup/derivation for \eqref{eqn:bdy}. 
We also partition $A$,$B$ in \eqref{eq:characterization2} to correspond the entries that are associated with $\tilde{\phi}^{j,x}$ and  $\tilde{\phi}^{j,x}_b$. $\tilde{\phi}^{j,u}$ denote the reduced vector with only non-zero entries of $\lcolitu$. 
\begin{lemma}[Lemma 2, \cite{yu2021localized}] 
\label[lemma]{lem:jing}
If $B_b^{(j)}B_b^{(j)^ \dagger} = I$, then
the vectors $\{v^j[k]\}$ characterize all $\tilde{\phi}^{j,u}[k]$ via 
    \begin{equation}
    \label{eq:subs}
        \tilde{\phi}^{j,u}[k] = -B_b^{(j),\dagger} A_{bn}^{(j)} \tilde{\phi}^{j,x}[k] + \left(I - B_b^{(j),\dagger} B_b^{(j)}\right) v^j[k]\,.
    \end{equation}
\end{lemma}
We remark that the pseudo-inverse condition in \Cref{lem:jing} is equivalently to Assumption \ref{assump:feasibility}, as observed in \citet{yu2021localized} and \citet{anderson2017structured}. 

We can now substitute \eqref{eq:subs} into the synthesis problem \eqref{eq:synth} and obtain an SLS synthesis problem in the same form as \eqref{eq:objcol} with transformed dynamical evolution in terms of the new variables $\tilde{\phi}^{j,x}[k]$ and $v^j[k]$. Consider the optimal solutions $\tilde{\phi}$ and $\tilde{\phi'}$ (concatenated from $\tilde{\phi}^{j,x}$ and $v^j$ ) computed from the de-constrained problem with two different model input $\Theta$ and $\Theta'$. By \Cref{thm:mainsensitivity_global}, we have
\begin{equation}
\label{eq:last-one}
    \norm{\tilde{\phi} - \tilde{\phi}'}_2 \leq \left(\Gamma_A + \Gamma_B\right) \norm{\Theta - \Theta'}_F.
\end{equation}
Observe that
$$\tilde{\phi}^{j,u} = \begin{bmatrix}-B_b^{(j),\dagger} A_{bn}^{(j)} & \left(I - B_b^{(j),\dagger} B_b^{(j)} \right) \end{bmatrix} \tilde{\phi}. $$
Therefore, we could bound the sensitivity of the solution to \eqref{eq:synth} via
\begin{align*}
    \norm{\phi - \phi'}_2 
    &\leq \norm{\begin{bmatrix}
        I & 0 \\ -B_b^{(j),\dagger} A_{bn}^{(j)} & \left(I - B_b^{(j),\dagger} B_b^{(j)} \right) 
        \end{bmatrix}\tilde{\phi} - \begin{bmatrix}
        I & 0 \\ -B_b^{' (j),\dagger} A_{bn}^{'(j)} & \left(I - B_b^{'(j),\dagger} B_b^{'(j)} \right) 
        \end{bmatrix}\tilde{\phi}'}_2\\
    &\leq  \norm{\left(\begin{bmatrix}
        0 & 0 \\ -B_b^{(j),\dagger} A_{bn}^{(j)} +B_b^{'(j),\dagger} A_{bn}^{'(j)} & - B_b^{(j),\dagger} B_b^{(j)}  + B_b^{'(j),\dagger} B_b^{'(j)} 
        \end{bmatrix}\right)\tilde{\phi}}_2\\
        & \quad \quad \quad \quad \quad  +\norm{ \begin{bmatrix}
        I & 0 \\ -B_b^{'(j),\dagger} A_{bn}^{'(j)} & \left(I - B_b^{'(j),\dagger} B_b^{' (j)} \right) 
        \end{bmatrix}\left(\tilde{\phi}- \tilde{\phi}' \right)}_2\\
    &\leq  \frac{4C\kappa }{\sigma_{\min} (1-\rho) }  
        + \left(2+ \frac{2\kappa }{\sigma_{\min}}\right)(\Gamma_A + \Gamma_B) \norm{\Theta - \Theta'}_F\\
    &= \bigO\left(\Gamma_A + \Gamma_B\right ) \norm{\Theta - \Theta'}_F,
\end{align*}
where $\sigma_{\min}$ denotes the minimum singular value of the matrix $B_b$ for all $B({\theta^i})$ with $\theta^i\in\mathcal{P}^i_0$. Note that the left pseudo-inverse has the largest singular value of $1/ \sigma_{\min}$ with $\sigma_{\min}$ the smallest singular value of the original matrix. Due to \Cref{assump:controllable} and \Cref{assump:feasibility}, we know that $B_b$ has to be bounded from below so that \eqref{eq:synth} is feasible.
We have also used the fact that norm of an lower triangular block matrix is upperbounded by the sum of the norm of each component block. We invoke the exponential decay property of the closed-loop responses to bound the decay rate of $\tilde{\phi}$ by relating the nonzero component of the solution to \eqref{eq:synth} and $\tilde{\phi}$ via \eqref{eq:last-one}.
\end{proof}

\section{Perturbation Analysis of $\mathcal{H}_2$-optimal SLS Synthesis}
\label{sec:sensitivity}
\input{misc/sensitivity}

\section{Extensions to Non-Convex Parameter Set Setting}
\label{sec:convex}
Representing model uncertainty as convex compact parameter sets is not always practical; sometimes potentially even impossible. Our approach can be readily extended to compact non-convex parameter sets $\mathcal{S}$, if those can be written as a finite union of convex sets $\bigcup^{N}_{i=1} \mathcal{P}_i$. This class of non-convex sets covers a large range of practical scenarios and the presented approach can be extended without losing stability guarantees. We can ensure by wrapping the proposed algorithm in a high-level routine SETSELECT, which runs the algorithm on the smaller convex sets $\mathcal{P}_i$ until they become entirely inconsistent: 
\begin{enumerate}
    \item At $t=0$, we select an arbitrary convex set $\mathcal{P}_{k_0}$ and perform consistent model chasing with CONSIST as before.
    \item If at some point $\mathcal{P}_{k_0}$ becomes entirely inconsistent, we select an arbitrary set $\mathcal{P}_{k_1}$ from the remaining collection $\{\mathcal{P}_1,\dots, \mathcal{P}_{N}\} \setminus \mathcal{P}_{k_0}$ and restart CONSIST with that set $\mathcal{P}_{k_1}$. If $\mathcal{P}_{k_1}$ is also entirely inconsistent, repeat that selection process.
\end{enumerate}
Per definition, the above algorithm never violates consistency. Because there are finitely many convex sets $\mathcal{P}_{k_i}$, the cost accrued due to restarting CONSIST scales up the total movement cost of the convex counterpart by a fixed constant. Overall, the stability proof is not impacted.

\end{document}

%% file: misc/preamble.tex
\usepackage{setspace}
\usepackage{graphics,graphicx,color}
\usepackage{booktabs}
\usepackage{dcolumn}
\usepackage{siunitx}
\usepackage{tabulary}
\usepackage{enumerate}
\usepackage{multirow}
\usepackage{hhline}
\usepackage{bm}

\usepackage[ruled, lined, algo2e]{algorithm2e}
\SetAlFnt{\small}
\SetKwInput{KwInit}{Initialize}

\SetCommentSty{mycommfont}

\usepackage{subfigure} 
\usepackage{cleveref}
\usepackage{breakcites}
\usepackage{float} 
\usepackage{multirow}
\usepackage{hhline}
\usepackage{braket}
\usepackage{dsfont} 
\usepackage{multirow}
\usepackage{hhline}
\usepackage[makeroom]{cancel}
\usepackage{ifthen}
\usepackage{comment}
\usepackage{subeqnarray}
\usepackage{setspace}
\usepackage{graphics,color}
\usepackage{url}
\usepackage{mathtools}

\usepackage{lipsum}
\usepackage{amsfonts}
\usepackage{graphicx}
\usepackage{epstopdf}

\ifpdf
  \DeclareGraphicsExtensions{.eps,.pdf,.png,.jpg}
\else
  \DeclareGraphicsExtensions{.eps}
\fi

\newtheorem{assumption}{Assumption}
\newtheorem{remark}{Remark}

\newtheorem{example}{Example}

\newcommand\abs[1]{\left|#1\right|}
\newcommand\norm[1]{\left\Vert#1\right\Vert}
\newcommand\infnorm[1]{\left\Vert#1\right\Vert_\infty}

\newcommand{\bigO}{O}

\newcommand{\N}{\mathcal{N}}
\newcommand{\PX}[1]{\Phi^x_{#1}}

\newcommand{\PU}[1]{\Phi^u_{#1}}

\newcommand{\bPX}[1]{\mathbf{\Phi^x_{#1}}}
\newcommand{\bPU}[1]{\mathbf{\Phi^u_{#1}}}
\newcommand{\Z}{\mathcal{Z}}

\newcommand{\pit}{\mathcal{P}^i_t}
\newcommand{\thit}{\theta^i_t}
\newcommand{\Thit}{\hat{\Theta}^i_t}

\newcommand{\A}[1]{A\left(#1\right)}
\newcommand{\B}[1]{B\left(#1\right)}

\newcommand{\colit}{\bm{\phi}^i_t}
\newcommand{\lcolit}{{\phi}^i_t}
\newcommand{\colitx}{\bm{\phi}^{i,x}_t}
\newcommand{\colitu}{\bm{\phi}^{i,u}_t}
\newcommand{\coljtx}{\bm{\phi}^{j,x}_t}
\newcommand{\coljtu}{\bm{\phi}^{j,u}_t}

\newcommand{\lcolitx}{{\phi}^{i,x}_t}
\newcommand{\lcolitu}{{\phi}^{i,u}_t}
\newcommand{\coljt}{\bm{\phi}^j_{t-d(j\rightarrow i)}}
\newcommand{\lcoljt}{{\phi}^j_{t-d(j\rightarrow i)}}
\newcommand{\lcoljtx}{{\phi}^{j,x}_{t-d(j\rightarrow i)}}
\newcommand{\lcoljtu}{{\phi}^{j,u}_{t-d(j\rightarrow i)}}
\newcommand{\lcoljltx}{{\phi}^{j,x}_{t-1-d(j\rightarrow \ell)}}
\newcommand{\lcoljltu}{{\phi}^{j,u}_{t-1-d(j\rightarrow \ell)}}
\newcommand{\lcoljlt}{{\phi}^j_{t-1-d(j\rightarrow \ell)}}

\newcommand{\M}[1]{\mathcal{M}\left(#1\right)}
\newcommand{\din}[1]{\mathcal{D}_{\text{in}}\left(#1\right)}
\newcommand{\dout}[1]{\mathcal{D}_{\text{out}}\left(#1\right)}
\newcommand{\delay}[2]{d \left(#1 \rightarrow #2\right)}
\newcommand{\trans}[1]{\left(#1\right)^T}

\newcommand{\F}{\bm{F}}
\newcommand{\g}{\bm{g}}



%% file: misc/sensitivity.tex
\input{misc/dimmacros.tex}
\subsection{From $\cl{H}_2$-optimal control to Least Squares}
This section presents results about general SLS synthesis. Due to notation overhead, we will drop time indices and suppress the horizon index $k\in[H]$ in closed-loop operators $\PX{}[k]$, $\PX{}[k]$ and write $\Phi^x_k$, $\Phi^u_k$ instead. 
\noindent Let $\Phi^x_k\in\mathbb{R}^{n \times n}$ and $\Phi^u_k\in\mathbb{R}^{m \times n}$ and consider the following canonical SLS synthesis problem with LQR cost for system matrices $[A,B]$ and weighting matrices $C \in \R^{n \times n}, D\in \R^{m \times m}$ :
\begin{align} 
\label{eq:obj} S = \min & \left\| \begin{bmatrix}C & 0\\0& D \end{bmatrix} \begin{bmatrix} \Phi^x_1 & \Phi^x_2 & \dots & \Phi^x_T \\ \Phi^u_1 & \Phi^u_2 & \dots & \Phi^u_T \end{bmatrix} \right\|^2_{F} \\
\notag \text{ s.t.: } & \Phi^x_{1} = I \\
\notag & \Phi^x_{k+1} = A \Phi^x_{k} + B \Phi^u_{k}, \quad \forall \: k: 1 \leq k \leq H \\
\notag & \Phi^x_{H+1} = 0
\end{align}
The objective in \eqref{eq:obj} is equivalent to weighted $\mathcal{H}_2$ norm on the closed-loop operators $\bPX{}$ and  $\bPU{}$, as well as the LQR cost on the state and control input weighed by $C^2$ and $D^2$.
Denote $\phi^{j,x}_{k}\in\mathbb{R}^n$, $\phi^{j,u}_{k}\in\mathbb{R}^m$ as the $j$th column of $\Phi^{x}_{k} \in \R^{n \times n}$, $\Phi^{u}_{k} \in \R^{m \times n}$ and $e_j$ the unit vector in the $j$-th coordinate axis. As described in \Cref{sec:sls}, we can separate the problem by columns and can equivalently restate \eqref {eq:obj} in terms of each column $\phi^{j,x}_k$ and $\phi^{j,u}_k$ : 
\begin{align} 
    \label{eq:objcol}  S_j :=\min & \left\| \begin{bmatrix}C & D \end{bmatrix} \begin{bmatrix} \phi^{j,x}_1 & \phi^{j,x}_{2} & \dots & \phi^{j,x}_{H} \\ \phi^{j,u}_1 & \phi^{j,u}_{2} & \dots & \phi^{j,u}_{H} \end{bmatrix} \right\|^2_{F} \\
    \notag \text{ s.t.: } & \phi^{j,x}_1 = e_j \\
    \notag & \phi^{j,x}_{k+1} = A \phi^{j,x}_{k} + B \phi^{j,u}_{k}, \quad \forall \: 1 \leq k \leq H  \\
    \notag & \phi^{j,x}_{H+1} = 0
\end{align}
We will now fix $j$ and rewrite \eqref{eq:objcol} further and introduce new variables to avoid tedious notation.
Define $u_k = \phi^{j,u}_{k}, \forall 1 \leq k \leq H $,  $\bm{u} = [u^{\top}_1 ,\dots, u^{\top}_H]^\top$ and the block-lower-triangular matrix $\bm{G}_u \in\mathbb{R}^{Hn \times Hm}$, the vector $\xi_j \in \mathbb{R}^{Hn}$ and the lifted weight matrices $\bm{C}$, $\bm{D}$ as
\begin{align}
&\bm{G}_u = \begin{bmatrix}B &0 & 0 & \dots & 0\\ 
    AB &B & 0 & \dots & 0\\ 
    A^2B &AB & B & \dots & 0\\ 
     & \dots & \dots & & \\
     A^{H-1}B & A^{H-2}B & A^{H-3}B & \dots & B\\
\end{bmatrix} & & \xi_j =\begin{bmatrix} -A e_j \\ -A^2 e_j \\ \dots \\ -A^H e_j \end{bmatrix} && \bm{C} = I_H \otimes C && \bm{D} =  I_H \otimes D , 
\end{align}
where $I_k$ is the identity matrix for $\R^k$. Denote by $P_i$, $1 \leq i \leq H$ the $i$-th block-row of $\bm{G}_u$:
\begin{align}
P_i = [A^{i-1}B, A^{i-2}B, \dots, B, 0, \dots,0]
\end{align}
Observe that with these definitions, it holds that for any feasible $\phi^{j,u}_{k}$, $\phi^{j,x}_{k}$ and fro all $\forall 1 \leq k \leq H$: $$\phi^{j,x}_{k+1} = -\xi_{j,k} + P_k \bm{u}$$ due to the constraints in \eqref{eq:objcol}.
Now we can rewrite the subproblem $S_j$ as
\begin{subequations}
\label{eq:SJ-original}
\begin{align}
\label{eq:Sj}S_j = \min\limits_{\bm{u}} \quad & \left\|\begin{bmatrix}\bm{C} \bm{G}_u\\\bm{D} \end{bmatrix} \bm{u} - \begin{bmatrix}\bm{C}\xi_j\\\bm{0}\end{bmatrix} \right\|^2_2 \quad  + (C^{\top}C)_{jj} \\
 \label{eq:cons}\text{ s.t.: } &\quad  0 = A^\top 
  e_j + P_H \bm{u}
\end{align}
\end{subequations}
For large systems which consist of many interconnected (sparsely) small systems, it is often the case that the overall system is $H$-controllable for some suitable choice of $H\ll n$ where $n$ is the global state dimension.

\subsection{Representation as a Least-Squares problem}
We now rewrite \eqref{eq:SJ-original} as a least square problem. Define $\bm{u}^*_{c} :=  P^{\top}_H (P_H P^{\top}_H)^{-1}A^\top e_j $, which is the solution to the optimization problem
\begin{align*}
\min\limits_{\bm{u}} & \quad \|\bm{u}\|^2_2 \\
\text{ s.t. } & 0 = -A^\top e_j + P_H \bm{u}. 
\end{align*}
We can interpret $\bm{u}^*_{c}$ as the smallest control action, measured in $\ell_2$, that drives the system from the origin to $-A^\top e_j$ in $H$ time-steps. This relates to controllability grammians as described in \cite{dullerud2013course}. 
Using $M^+$ to denote the Moore-Penrose Inverse of a matrix $M$, we can also write $\bm{u}^*_{c} :=  P^{+}_H A^\top e_j = P^\top_H W^{-1}_H A^\top e_j  $, where $W_H = P_H P_H^\top$.\\

Let $H$ denote the FIR-Horizon of the problem, then define the matrices
\begin{align}\label{eq:gab}
    &\bm{G}_w(A) = \begin{bmatrix}I &0 & 0 & \dots & 0\\ 
        A &I & 0 & \dots & 0\\ 
        A^2 &A & I & \dots & 0\\ 
         & \dots & \dots & & \\
         A^{H-1} & A^{H-2} & A^{H-3} & \dots & I\\
    \end{bmatrix}, &\bm{G}_u(A,B) = \begin{bmatrix}B &0 & 0 & \dots & 0\\ 
        AB &B & 0 & \dots & 0\\ 
        A^2B &AB & B & \dots & 0\\ 
         & \dots & \dots & & \\
         A^{H-1}B & A^{H-2}B & A^{H-3}B & \dots & B\\
    \end{bmatrix}.
\end{align}
and denote $P_i(A,B)$ as the $i$th block matrix row of $\bm{G}_u(A,B)$:
\begin{align}\label{eq:pdef}
    P_i(A,B) = [A^{i-1}B, A^{i-2}B, \dots, B, 0, \dots,0]
\end{align}
$\bm{G}_u(A,B)$ can be written as $\bm{G}_u(A,B) = \bm{G}_w(A)(I_H \otimes B)$, where $I_H$ is the identity matrix in $\R^H$. Let $Z \in \mathbb{R}^{H \times H}$ be defined as the nilpotent matrix
\begin{align}
Z = \begin{bmatrix} \bm{0}_{H-1 \times 1}  & I_{H-1} \\
0& \bm{0}_{1 \times H-1} \end{bmatrix},
\end{align}
and notice it's psuedo-inverse is $Z^+ = Z^\top$. Using $Z$, it is easy to verify that $\bm{G}_w(A)$ can be expressed as:
\begin{align}
    \bm{G}_w(A) = \left(I_H - Z^+ \otimes A\right)^{-1}.
\end{align}

Ignoring the constant terms in \eqref{eq:Sj}, we can reparametrize $\bm{u}= -\bm{u}^*_c + \bm{u}'$ where $\bm{u}'\in\mathrm{null}(P_H)$ and describe \eqref{eq:SJ-original} as the optimization problem:
\begin{align}\label{eq:sj}
    S_j \quad &:= \min\limits_{\bm{u}' \in \mathrm{null}(P_H(A,B))} \Big\|\begin{bmatrix}\bm{C}& 0\\ 0 & \bm{D}\end{bmatrix}\begin{bmatrix}\bm{G}_u(A,B) \\ I \end{bmatrix}(\bm{u}'-\bm{u}_c^*(A,B))\Big\|^2_2.
   \end{align}
Let $\bm{u}^*(A,B)$ be a minimizer of the above problem for fixed $A,B$, we are interested in the SLS solutions $$\phi^{*j}(A,B) := \begin{bmatrix}\phi^{*j}_x(A,B) \\ \phi^{*j}_u(A,B) \end{bmatrix} = \begin{bmatrix}\bm{G}_u(A,B)\\I \end{bmatrix}(\bm{u}^*(A,B)-\bm{u}_c^*(A,B))$$ 
and how these solutions are perturbed with changes in $A,B$.

For the rest of the discussion, we will drop mentioning the explicit dependence on $(A,B)$ and the column index $j$ to reduce the notational burden.
First, we (over-)parametrize $\bm{u}$ as $$\bm{u} = (I-P^+_H P_H)\bm{\eta}, $$ to cast the above problem into an unconstrained one:
\begin{align}\label{eq:sj}
    S_j \quad &:= \min\limits_{\bm{\eta}} \Big\|\underbrace{\begin{bmatrix}\bm{C}& 0\\ 0 & \bm{D}\end{bmatrix}\begin{bmatrix}\bm{G}_u(A,B) \\ I \end{bmatrix}(I-P^+_HP_H)}_{\F}\bm{\eta}- \underbrace{\begin{bmatrix}\bm{C}& 0\\ 0 & \bm{D}\end{bmatrix}\begin{bmatrix}\bm{G}_u(A,B) \\ I \end{bmatrix}\bm{u}^*_c(A,B)}_{\g}\Big\|^2_2
\end{align}
The unique min-norm solution $\eta^*$ to the above problem is $\eta^* = \F^+\g$ and therefore the optimal solution $\phi^*$ takes the form
\begin{align}
\label{eq:phi}
\phi^* = \begin{bmatrix}\bm{C}^{-1}& 0\\ 0 & \bm{D}^{-1}\end{bmatrix}(\F \F^+\g - \g) = \begin{bmatrix}\bm{C}^{-1}& 0\\ 0 & \bm{D}^{-1}\end{bmatrix}\underbrace{(\F \F^+ - I)\g}_{\nu^*} =: \begin{bmatrix}\bm{C}^{-1}& 0\\ 0 & \bm{D}^{-1}\end{bmatrix} \nu^*
\end{align}

\subsection{Local lipshitzness of $\cl{H}_2$-optimal closed-loop operators}
\noindent Here, we perform perturbation analysis on the term $\nu^* = (\F \F^+ - I)\g$. Throughout the discussion, we will make frequent use of the following identities:
\begin{lemma}\label{lem:diffs}
For arbitrary matrices $X,Y\in\R^{n \times m}$ and $A, B \in \R^{n \times n}$, it holds that 
\begin{enumerate}[i)]
\item $A^k_1-A^k_2 = \sum^{k-1}_{j=0} A^{k-1-j}_1(A_1-A_2)A^j_2$
\item $XX^+-YY^+ = (I-XX^+)(X-Y)Y^+ + \left[(I-YY^+)(X-Y)X^+\right]^\top$
\item If $A$ and $B$ are invertible, then $A^{-1}-B^{-1}= A^{-1}(B-A)B^{-1}$.
\end{enumerate}
\end{lemma}
\noindent The following is a corollary from Theorem 4.1 in \cite{wedin1973perturbation}:
\begin{theorem}\label{thm:wedin}
Let $X$ and $Y$ be matrices with equal rank, let $\|\:\cdot\:\|_2$ denote the induced 2-norm and $\|\:\cdot\:\|_F$ denote the Frobenius norm. The following inequalities hold:
\begin{align*}
\|X^+-Y^+\|_2 &\leq \varphi \|X^+\|_2\|Y^+\|_2 \|X-Y\|_2\\
\|X^+-Y^+\|_F &\leq \sqrt{2} \|X^+\|_2\|Y^+\|_2 \|X-Y\|_F
\end{align*}
where $\varphi = \tfrac{1 + \sqrt{5}}{2}$ denotes the golden ratio constant.
\end{theorem}
 
 \noindent Next we present the core theorem of the perturbation analysis: Given two \textit{arbitrary} controllable systems $(A_1,B_1)$ and $(A_2,B_2)$, \thmref{thm:mainsensitivity} bounds the worst-case difference in solutions $\|\phi^{*}_1-\phi^{*}_2\|_2$ in terms of the differences in parameters space $\|A_1-A_2\|_{2}$ and $\|B_1-B_2\|_2$ between both systems. This result is the first perturbation bound for $\cl{H}_2$-optimal control with SLS (considering arbitrary pairs of $A_1,A_2$ and $B_1,B_2$) . 
\begin{theorem}\label{thm:mainsensitivity}
Let $C,D \succ 0$, let $(A_1,B_1)$ and $(A_2,B_2)$ be two controllable pairs of system matrices with FIR horizon $H$ and let $\phi^{*j}_1$ and $\phi^{*j}_2$ be the corresponding SLS-solutions to the subproblem $S_j$. Then, it holds that:
\begin{align}
\|\phi^{j*}_1-\phi^{j*}_2\|_2 \leq \Gamma_A\|A_1-A_2\|_F + \Gamma_B \|B_1-B_2\|_F
\end{align}
where the Lipshitz-constants $\Gamma_A,\Gamma_B$ stand for
\begin{align*}
    \Gamma_A &= \kappa_{CD}\Gamma'_1 + \kappa_{CD}\Gamma'_2 \|B_1\|_2\|\bm{G}_w(A_1)\|_2, &&\kappa_{CD} = \frac{\max\{\sigma_{max}(C),\sigma_{max}(D)\}}{\min\{\sigma_{min}(C),\sigma_{min}(D)\}}  \\
    \Gamma_B &= \kappa_{CD}\Gamma'_2 \|\bm{G}_w(A_2)\|_2 && 
\end{align*}
and $\Gamma'_1$ and $\Gamma'_2$ are defined as:
\begin{align*}
    \Gamma'_1 &= \alpha_{H,1} \alpha_{H,2}H(1+\left\|\bm{G}_{u,2}\right\|_2)\|P^+_{H,2}\|_2  \\
    \Gamma'_2 &= \alpha_{H,1} \|P^+_{H,1}\|_2 \left(1 + {\varphi} \|P^+_{H,2}\|_2 +{\varphi} \|P^+_{H,2}\|_2\left\|\bm{G}_{u,2}\right\|_2 \right) + \|\g_2\|_2(\|\F^+_1\|_2 + \|\F^+_2\|_2) + \dots \\
    \notag & \quad + {\varphi}\|\g_2\|_2(\|\F^+_1\|_2 + \|\F^+_2\|_2)\|P^+_{H,1}\|_2 \|P^+_{H,2}\|_2(\|P_{H,1}\|_2 + \|P_{H,2}\|_2)(1+\|\bm{G}_{u,1}\|_2).
\end{align*}
and $\varphi=\tfrac{1+\sqrt{5}}{2}$ is the golden ratio. 
\end{theorem}
\begin{proof}
Recall the identities of \lemref{lem:diffs}. Write $\nu^*_1-\nu^*_2$ where $\nu^*_i$ is from \eqref{eq:phi} for $(A_i,B_i)$ as
\begin{align}
\notag    \nu^*_1-\nu^*_2 &= (\F_1\F^+_1-I)(\g_1-\g_2) + (\F_1\F^+_1 - \F_2\F^+_2)\g_2 \\
 \label{eq:nu1bound}   \|\nu^*_1-\nu^*_2\|_2&\leq \|\g_1-\g_2\|_2 + \|\F_1\F^+_1 - \F_2\F^+_2\|_2\|\g_2\|_2,
\end{align}
where we used the fact that $(\F_1\F^+_1-I)$ is a projection and therefore $\|\F_1\F^+_1-I\|_2 =1$. Rewrite $\F_1\F^+_1 - \F_2\F^+_2$ as
$$ (I-\F_1\F^+_1)(\F_1-\F_2)\F^+_2 + \left[(I-\F_2\F^+_2)(\F_1-\F_2)\F^+_1\right]^\top $$ to conclude that 
\begin{align}
\|\F_1\F^+_1 - \F_2\F^+_2\|_2 \leq \|\F_1-\F_2\|_2(\|\F^+_1\|_2 + \|\F^+_2\|_2).
\end{align}
Substitution into \eqref{eq:nu1bound} yields:
\begin{align}
     \label{eq:nu2bound}   \|\nu^*_1-\nu^*_2\|_2 &\leq \|\g_1-\g_2\|_2 + \|\F_1-\F_2\|_2(\|\F^+_1\|_2 + \|\F^+_2\|_2)\|\g_2\|_2,
\end{align}
\begin{enumerate}
\item \underline{Bounding $\|\F_1-\F_2\|_2$}: Rewrite $\F_1-\F_2$ as 
\begin{align}
    \begin{bmatrix}\bm{C}^{-1}& 0\\ 0 & \bm{D}^{-1}\end{bmatrix}(\F_1-\F_2) &=  \begin{bmatrix}\bm{G}_{u,1} \\ I \end{bmatrix}(I-P^+_{H,1}P_{H,1}) -\begin{bmatrix}\bm{G}_{u,2} \\ I \end{bmatrix}(I-P^+_{H,2}P_{H,2}) \\
    &=  \begin{bmatrix}\bm{G}_{u,1} \\ I \end{bmatrix}(P^+_{H,2}P_{H,2}-P^+_{H,1}P_{H,1}) + \begin{bmatrix}\bm{G}_{u,1} - \bm{G}_{u,2} \\ 0 \end{bmatrix}(I-P^+_{H,2}P_{H,2})
\end{align}
From the above we can derive the inequality:
\begin{align}\label{eq:F12bound}
\frac{\|\F_1-\F_2\|_2}{\max\{\|C\|_2,\|D\|_2\}} &\leq (1+\|\bm{G}_{u,1}\|_2)\|P^+_{H,2}-P^+_{H,1}\|_2(\|P_{H,1}\|_2 + \|P_{H,2}\|_2) + \|\bm{G}_{u,1}-\bm{G}_{u,2}\|_2 
\end{align}
Now we will use the result \thmref{thm:wedin} to bound $\|P^+_{H,2}-P^+_{H,1}\|_2$ as
\begin{align}
    \|P^+_{H,2}-P^+_{H,1}\|_2 \leq {\varphi} \|P^+_{H,1}\|_2 \|P^+_{H,2}\|_2 \|P_{H,2}-P_{H,1}\|_2
\end{align}
Furthermore, noticing $P_{H,2}-P_{H,1} = [\bm{0},\dots,\bm{0},\bm{I}_n](\bm{G}_{u,2}-\bm{G}_{u,1})$ we can conclude 
\begin{align}
    \|P^+_{H,2}-P^+_{H,1}\|_2 \leq {\varphi} \|P^+_{H,1}\|_2 \|P^+_{H,2}\|_2 \|\bm{G}_{u,2}-\bm{G}_{u,1}\|_2.
\end{align}
We combine this into \eqref{eq:F12bound} to obtain
\begin{align}\label{eq:F12bound-2}
    \notag &\frac{\|\F_1-\F_2\|_2}{\max\{\|C\|_2,\|D\|_2\}}\\
     \leq &\left( 1+ {\varphi} \|P^+_{H,1}\|_2 \|P^+_{H,2}\|_2 (1+\|\bm{G}_{u,1}\|_2)(\|P_{H,1}\|_2 + \|P_{H,2}\|_2) \right)\|\bm{G}_{u,1}-\bm{G}_{u,2}\|_2 
    \end{align}
\item \underline{Bounding $\|\g_1-\g_2\|_2$}:
Introduce the constant $\alpha_H := \max_{0\leq k\leq H} \|A^k\|_2$ and observe that $\|A^H_1-A^H_2\|_2$ can be bounded as:
\begin{align}
    \|A^H_1-A^H_2\|_2 = \|\sum^{H-1}_{j=0} A^{H-1-j}_1(A_1-A_2)A^j_2\| \leq H \alpha_{H,1} \alpha_{H,2} \|A_1-A_2\|_2
\end{align}
We can rewrite $\g_1-\g_2$ as
\begin{align}
    \begin{bmatrix}\bm{C}^{-1}& 0\\ 0 & \bm{D}^{-1}\end{bmatrix}(\g_1-\g_2) &= \begin{bmatrix}\bm{G}_{u,1} \\ I \end{bmatrix}P^+_{H,1}A^H_1e_j - \begin{bmatrix}\bm{G}_{u,2} \\ I \end{bmatrix}P^+_{H,2}A^H_2e_j \\
&=\begin{bmatrix}(\bm{G}_{u,1}-\bm{G}_{u,2}) \\ 0 \end{bmatrix}P^+_{H,1}A^H_1e_j + \begin{bmatrix}\bm{G}_{u,2} \\ I\end{bmatrix}(P^+_{H,1}-P^+_{H,2})A^H_1e_j \\
\notag&\quad \dots+ \begin{bmatrix}\bm{G}_{u,2} \\ I \end{bmatrix}P^+_{H,2}(A^H_1-A^H_2)e_j
\end{align}
and obtain the bound:
\begin{align}
\frac{\|\g_1-\g_2\|_2}{\max\{\|C\|_2,\|D\|_2\}} &\leq \alpha_{H,1}\left\| \bm{G}_{u,1}-\bm{G}_{u,2}\right\|_2 \|P^+_{H,1}\|_2+ \alpha_{H,1}(1+\left\|\bm{G}_{u,2}\right\|_2)\left\|P^+_{H,1}-P^+_{H,2}\right\|_2 \\
&\quad \dots + \alpha_{H,1} \alpha_{H,2}H(1+\left\|\bm{G}_{u,2}\right\|_2)\|P^+_{H,2}\|_2\|A_1-A_2\|_2 \\
&\leq  \alpha_{H,1} \|P^+_{H,1}\|_2 \left(1 + {\varphi} \|P^+_{H,2}\|_2 +{\varphi} \|P^+_{H,2}\|_2\left\|\bm{G}_{u,2}\right\|_2 \right)\left\| \bm{G}_{u,1}-\bm{G}_{u,2}\right\|_2 \\
&\quad \dots + \alpha_{H,1} \alpha_{H,2}H(1+\left\|\bm{G}_{u,2}\right\|_2)\|P^+_{H,2}\|_2\|A_1-A_2\|_2
\end{align}
\end{enumerate}
We get the bound 
\begin{align}
    \frac{\|\nu^*_1-\nu^*_2\|_2}{\max\{\|C\|_2,\|D\|_2\}} \leq \Gamma'_1 \|A_1-A_2\|_2 + \Gamma'_2 \|\bm{G}_{u,1}-\bm{G}_{u,2}\|_2
\end{align}
where $\Gamma'_1$ and $\Gamma'_2$ are the constants:
\begin{align}
    \Gamma'_1 &= \alpha_{H,1} \alpha_{H,2}H(1+\left\|\bm{G}_{u,2}\right\|_2)\|P^+_{H,2}\|_2  \\
    \Gamma'_2 &= \alpha_{H,1} \|P^+_{H,1}\|_2 \left(1 + {\varphi} \|P^+_{H,2}\|_2 +{\varphi} \|P^+_{H,2}\|_2\left\|\bm{G}_{u,2}\right\|_2 \right) + \|\g_2\|_2(\|\F^+_1\|_2 + \|\F^+_2\|_2) + \dots \\
    \notag & \quad + {\varphi}\|\g_2\|_2(\|\F^+_1\|_2 + \|\F^+_2\|_2)\|P^+_{H,1}\|_2 \|P^+_{H,2}\|_2(\|P_{H,1}\|_2 + \|P_{H,2}\|_2)(1+\|\bm{G}_{u,1}\|_2) 
\end{align}
Using \lemref{lem:Gwu}, we obtain the final bound:
\begin{align}\label{eq:Gamma12}
    \|\phi^*_1 - \phi^*_2\|_2 \leq  \kappa_{CD}\|\nu^*_1-\nu^*_2\|_2 \leq \Gamma_A \|A_1-A_2\|_2 + \Gamma_B \|B_1-B_2\|_2
\end{align}
with the constants $\Gamma_A,\Gamma_B$ defined as:
\begin{align}\label{eq:GammaAB}
\Gamma_A &= \kappa_{CD}\Gamma'_1 + \kappa_{CD}\Gamma'_2 \|B_1\|_2\|\bm{G}_w(A_1)\|_2 \|\bm{G}_w(A_2)\|_2 \\
\Gamma_B &= \kappa_{CD}\Gamma'_2 \|\bm{G}_w(A_2)\|_2
\end{align}
\end{proof}
\subsection{Global lipshitzness of $\cl{H}_2$-optimal closed-loop operators over compact sets $\cl{S}$}

\noindent This section derives a global Lipshitz bound for $\cl{H}_2$-optimal SLS solutions over a compact set of controllable systems $\cl{S}$. As a starting point we consider the previous theorem \thmref{thm:mainsensitivity}. Our main proof strategy is to derive global bounds on the constants $\Gamma_A$ and $\Gamma_B$ instead of for a fixed pair of systems. We proceed with a collection lemmas bounding individual terms in the equations \eqref{eq:Gamma12} and \eqref{eq:GammaAB} for $\cl{S}$.

\subsubsection{Auxiliary Lemmas}
\begin{lemma}\label{lem:Gwu}
For any pair of system matrices $(A_1,B_1)$ and $(A_2,B_2)$ (with compatible dimensions) holds
\begin{align}
    \|\bm{G}_w(A_1)-\bm{G}_w(A_2)\|_2 & \leq \|\bm{G}_w(A_1)\|_2 \|\bm{G}_w(A_2)\|_2 \| A_1 -  A_2\|_2 \\
    \notag \|\bm{G}_u(A_1,B_1)-\bm{G}_u(A_2,B_2)\|_2& \leq \|B_1\|_2\|\bm{G}_w(A_1)\|_2 \|\bm{G}_w(A_2)\|_2 \| A_1 -  A_2\|_2 + \|\bm{G}_w(A_2)\|_2\|B_1- B_2\|_2
\end{align}
\end{lemma}
\begin{proof}
Using \lemref{lem:diffs} we can write 
Using $\bm{G}_u(A,B) = \bm{G}_w(A)(I_H \otimes B)$ and \lemref{lem:diffs} we can write $\bm{G}_{u,1}-\bm{G}_{u,2}$ as

\begin{align}
    \bm{G}_{u,1}-\bm{G}_{u,2} &= \bm{G}_w(A_1)(I_H \otimes B_1) - \bm{G}_w(A_2)(I_H \otimes B_2) \\
    &= \left(\bm{G}_w(A_1)-\bm{G}_w(A_2)\right)(I_H \otimes B_1) + \bm{G}_w(A_2)\left( I_H \otimes (B_1- B_2)\right) 
\end{align}
It holds that
\begin{align}
    \bm{G}_w(A_1)-\bm{G}_w(A_2) &= \bm{G}_w(A_1)(\bm{G}_w(A_2)^{-1}-\bm{G}_w(A_1)^{-1})\bm{G}_w(A_2) \\
    &= \bm{G}_w(A_1)(Z^+ \otimes (A_1 -  A_2))\bm{G}_w(A_2)
\end{align}
which leads to the bound
\begin{align}
\left\| \bm{G}_w(A_1)-\bm{G}_w(A_2) \right\|_2 \leq \|\bm{G}_w(A_1)\|_2 \|A_1-A_2\|_2 \|\bm{G}_w(A_2)\|_2 
\end{align}
\end{proof}
In total, we need to global bounds on the quantities $\|\bm{G}_{u}\|_2$,$\|\bm{G}_{w}\|_2$, $\|P^+_{H}\|_2$, $\|P_{H}\|_2$, $\|\F^+\|_2$, $\|\g\|_2$.

\begin{lemma}
Let $(A,B)$ be pair of fixed system matrices, let $\bm{G}_{u}(A,B)$, $\bm{G}_w(A)$ be the matrices defined in \eqref{eq:gab}, and let $W^u_H = \sum^{H-1}_{i=0} A^iBB^\top A^{i\top}$, $W^w_H = \sum^{H-1}_{i=0} A^i A^{i\top}$  be the $H$th controllability grammian w.r.t to the input $u$ and the distrubance $w$, respectively. Then it holds:
\begin{align}
    &\|\bm{G}_{u}(A,B)\|_2 \leq \sqrt{H \sigma_{max}(W^u_H(A,B))} && \|\bm{G}_{w}(A)\|_2 \leq \sqrt{H \sigma_{max}(W^w_H(A))} 
\end{align}
\end{lemma}
\begin{proof}
$\|\bm{G}_{u}\|_2$ is defined as $\|\bm{G}_{u}\|^2_2 := \max\limits_{\|u\|_2 = 1} \|\bm{G}_u \bm{u}\|^2_2$, by decomposing $\bm{u} = [u^\top_0,\dots,u^\top_{H-1}]^\top$ we can rewrite this as
\begin{align}
\|\bm{G}_{u}\|^2_2 &= \max\limits_{\|u\|_2 = 1} \left\| \begin{bmatrix} B u_0 \\ ABu_0 + Bu_1\\ \dots \\ A^{H-1}Bu_0 + \dots + Bu_{H-1} \end{bmatrix} \right\|^2_2 = \max\limits_{\|u\|_2 = 1} \sum^{H}_{k=1} \|P_k \bm{u}\|^2_2 \\
&\leq  \sum^{H}_{k=1} \max\limits_{\|u\|_2 = 1} \|P_k \bm{u}\|^2_2 =\sum^{H}_{k=1}  \|P_k\|^2_2 \leq H \|P_H\|^2_2 \leq H \|W^u_H\|_2
\end{align}
Where we used the fact that $\|P_k\|^2_2$ increases in $k$ and that $\|P_k\|^2_2$ is equal to the induced $2$-norm of the corresponding controllabillity grammian $W^u_k = \sum^{k-1}_{i=0} A^iBB^\top A^{i\top}$. Thus, we obtain the bound
$$\|\bm{G}_{u}(A,B)\|_2 \leq \sqrt{H \sigma_{max}(W^u_H(A,B))}, $$
and the bound on $\|\bm{G}_{w}(A)\|_2$ follows in the same way.
\end{proof}
\begin{lemma}
    Let $(A,B)$ be pair of $H$-controllable fixed system matrices, let $P_H(A,B)$ be the matrix defined in \eqref{eq:pdef}, and let $W^u_H = \sum^{H-1}_{i=0} A^iBB^\top A^{i\top}$ be the $H$th controllability grammian w.r.t to the input $u$. Then, the induced $2$ norm of $P_H(A,B)$ and its Moore-Penrose Inverse $P^+_H(A,B)$ can be written as:
    \begin{align}
        &\|P_H(A,B)\|_2 = \left(\sigma_{max}(W^u_H(A,B))\right)^{\tfrac{1}{2}}  && \|P^+_H(A,B)\|_2 = \left(\sigma_{min}(W^u_H(A,B))\right)^{-\tfrac{1}{2}}
    \end{align}
\end{lemma}
\begin{proof}
    Because we assume a sufficient degree of controllability, $P_H(A,B)$ is full row-rank. This implies that 
    \begin{align}
    &\|P_H(A,B)\|_2 = \sqrt{\lambda_{max}(P_H(A,B)P^\top_H(A,B))} = \sqrt{ \sigma_{max}(W^u_H(A,B))}\\
    &\left(\|P^+_H(A,B)\|_2\right)^{-1}  = \sqrt{\lambda_{min}(P_H(A,B)P^\top_H(A,B))} = \sqrt{ \sigma_{min}(W^u_H(A,B))}
    \end{align}
\end{proof}

\begin{lemma}
Let $(A,B)$ be a fixed pair of $H$-controllable system matrices, and let $\F(A,B)$ denote the matrix 
\begin{align}
    \F(A,B) = \begin{bmatrix}\bm{C}& 0\\ 0 & \bm{D}\end{bmatrix}\begin{bmatrix}\bm{G}_u(A,B) \\ I \end{bmatrix}(I-P^+_H(A,B)P_H(A,B)).
\end{align}
Then, $\|\F^+(A,B)\|_2 \leq \sigma^{-1}_{min}(D)$.
\end{lemma}
\begin{proof}
For an arbitrary matrix $M$, $(\|M^+\|_2)^{-1}$ is equal to the smallest \underline{non-zero} singular eigenvalue of $M$ (we will denote this quantity as $\sigma_{-1}(M)$). Thus, in order to bound $\|M^+\|_2$ from above, we have to bound $\sigma_{-1}(M)$ from below. Denote $\bm{L}$ as the matrix $$ \bm{L}:=\begin{bmatrix}\bm{C}& 0\\ 0 & \bm{D}\end{bmatrix}\begin{bmatrix}\bm{G}_u(A,B) \\ I \end{bmatrix} $$ and notice that it is full column rank and has rank of $H\times n_u$. The projection $\Pi_{\cl{N}(P_H)}:=(I-P^+_H(A,B)P_H(A,B))$ has rank $H\times n_u-n_x$ due the assumption of $H$-controllability. Hence, $\bm{F}=\bm{L}\Pi_{\cl{N}(P_H)}$ is full column rank with rank $r_{\F} := H\times n_u-n_x$ and has a null space $\cl{N}(\F)$ of dimension $n_x$. From these observations, we can equivalently say that $\sigma_{-1}(\F)$ is the $r_{\F}$th largest (or equivalently $n_x+1$ smallest) singular eigenvalue of $\F$. Using the Minimax principle, we can therefore write:
\begin{align}
\sigma_{-1}(\F) &= \max\limits_{\mathrm{proj. }\Pi, \text{ s.t.: } \mathrm{rank}(\Pi) = r_{\F}} \min\limits_{x\text{ s.t.: }\|\Pi x\|=1} x^\top \Pi \F^\top \F \Pi x \\
&= \max\limits_{\mathrm{proj. }\Pi, \text{ s.t.: } \mathrm{rank}(\Pi) = r_{\F}} \min\limits_{x\text{ s.t.: }\|\Pi x\|=1} x^\top \Pi \Pi_{\cl{N}(P_H)} \bm{L}^\top \bm{L} \Pi_{\cl{N}(P_H)}  \Pi x 
\end{align}
Now recall that $\Pi_{\cl{N}(P_H)}$ is of rank $r_{\F}$, hence it is a feasible choice for the variable $\Pi$ of the outer optimization problem. This leads to the bound 
\begin{align}
    \sigma_{-1}(\F) &\geq \min\limits_{x\text{ s.t.: }\|\Pi_{\cl{N}(P_H)} x\|=1} x^\top \Pi_{\cl{N}(P_H)} \bm{L}^\top \bm{L} \Pi_{\cl{N}(P_H)} x \\
    &\geq \min\limits_{z\text{ s.t.: }\|z\|=1} z^\top \bm{L}^\top \bm{L} z = \sigma_{min}(\bm{L})
\end{align}
We obtain a simple, but possibly conservative, lower bound on $\sigma_{min}(\bm{L})$ as follows:
\begin{align*}
    &\sigma^2_{min}(\bm{L}) = \min\limits_{z\text{ s.t.: }\|z\|=1} \|\bm{L} z\|^2_2 = \min\limits_{z\text{ s.t.: }\|z\|=1} \|\bm{C}\bm{G}_u(A,B) z\|^2_2 + \|\bm{D}z\|^2_2 \geq \sigma^2_{min}(\bm{C}\bm{G}_u(A,B)) + \sigma^2_{min}(\bm{D})\\
    &\implies \quad \sigma_{min}(\bm{L}) \geq \sigma_{min}(\bm{D})
\end{align*}
Finally, this provides us with the final result: $\|\F^+(A,B)\|_2 = \sigma^{-1}_{-1}(\bm{F}) \leq  \sigma^{-1}_{min}(\bm{L}) \leq \sigma^{-1}_{min}(\bm{D})$
\end{proof}

We obtain an upper bound for $\|\g\|_2$, as a corollary of the previous three Lemmas:
\begin{lemma}
Let $(A,B)$ be a fixed pair of $H$-controllable system matrices. Let $\g = \bm{L}\bm{u}^*_c$, where $\bm{L}$ and $\bm{u}^*_{c}$ are defined as:
\begin{align}
    & \bm{L}:=\begin{bmatrix}\bm{C}& 0\\ 0 & \bm{D}\end{bmatrix}\begin{bmatrix}\bm{G}_u(A,B) \\ I \end{bmatrix} & &\bm{u}^*_{c} :=  P^{+}_H A^H e_j = P^\top_HW^{-1}_HA^H e_j.
\end{align}
Then, it holds:
$$\|\bm{g}\|_2 \leq \left(\|C\|_2 \sqrt{H}\sigma^{\tfrac{1}{2}}_{max}(W^u_H) + \|D\|_2\right)\sigma^{-\tfrac{1}{2}}_{min}(W^u_H)\alpha_H $$
where $\alpha_H := \max_{0\leq k\leq H} \|A^k\|_2$
\end{lemma}

\subsubsection{The final bound}
With the results of the last section, we can now bound the constants $\Gamma_A$ and $\Gamma_B$ used in \thmref{thm:mainsensitivity}. Rather than writing the explicit form of the constants we shall only analyze how they scale with system parameters. Recall $\Gamma_A$, $\Gamma_B$ are defined as 
\begin{align*}
    \Gamma_A &= \kappa_{CD}\Gamma'_1 + \kappa_{CD}\Gamma'_2 \|B_1\|_2\|\bm{G}_w(A_1)\|_2 \|\bm{G}_w(A_2)\|_2 \\
    \Gamma_B &= \kappa_{CD}\Gamma'_2 \|\bm{G}_w(A_2)\|_2,
\end{align*}
where $\Gamma'_1$, $\Gamma'_2$ are dominated by the terms: 
\begin{align*}
    \Gamma'_1 &\sim \cl{O}\left(\alpha_{H,1} \alpha_{H,2}H\left\|\bm{G}_{u,2}\right\|_2\|P^+_{H,2}\|_2\right)\\
    \Gamma'_2 &\sim \cl{O}\left(\|\g_2\|_2(\|\F^+_1\|_2 + \|\F^+_2\|_2)\|P^+_{H,1}\|_2 \|P^+_{H,2}\|_2(\|P_{H,1}\|_2 + \|P_{H,2}\|_2)(1+\|\bm{G}_{u,1}\|_2)\right) 
\end{align*}

Let us first revisit the collection of bounds we have derived:
\begin{enumerate}[i)]
\item $\|\bm{G}_{u}(A,B)\|_2 \leq \sqrt{H \sigma_{max}(W^u_H(A,B))}$,  $\|\bm{G}_{w}(A)\|_2 \leq \sqrt{H \sigma_{max}(W^w_H(A))} $
\item $\|P_H(A,B)\|_2 = \left(\sigma_{max}(W^u_H(A,B))\right)^{\tfrac{1}{2}},\quad \|P^+_H(A,B)\|_2 = \left(\sigma_{min}(W^u_H(A,B))\right)^{-\tfrac{1}{2}}$
\item $\|\F^+(A,B)\|_2 \leq \sigma^{-1}_{min}(D)$
\item $\|\bm{g}\|_2 \leq \left(\|C\|_2 \sqrt{H}\sigma^{\tfrac{1}{2}}_{max}(W^u_H) + \|D\|_2\right)\sigma^{-\tfrac{1}{2}}_{min}(W^u_H)\alpha_H $
\item $\alpha_H := \max_{0\leq k\leq H} \|A^k\|_2$
\item $\kappa_{CD} = \frac{\max\{\sigma_{max}(C),\sigma_{max}(D)\}}{\min\{\sigma_{min}(C),\sigma_{min}(D)\}}$
\end{enumerate}

Before we state the final bound, we require the following standard controllability result \cite{dullerud2013course}.

\begin{lemma}\label{lem:baseCtrl}
Let $\cl{S} $ be a compact set of matrices where each element $(A\in\R^{n \times n},B\in\R^{n \times m}) \in\cl{S}$ represents a controllable linear dynamical system with equations $x({t+1}) = Ax(t)+Bu(t) + w(t)$, state $x(t)\in \R^n$, input $u(t)\in\R^m$ and disturbance $w(t)\in\R^n$. Then, there exists an FIR Horizon $H\leq n$, and positive scalar constants $\overline{\sigma}^w$, $\underline{\sigma}^w$, $\overline{\sigma}^u$, $\underline{\sigma}^u$ such that the following statements hold:
    \begin{itemize}
        \item For any $(A,B)\in\cl{S}$ and any initial state, $\zeta_0$, there exists an input $u(0), u(1), \dots, u({H-1})$, such that the system trajectory $x({t+1})  = A x(t) + B u(t),\forall {t \leq H-1}, x(0) = \zeta_0$ satisfies $x(H)=0$ at time $H$.  
        \item For any $(A,B)\in\cl{S}$, the matrix $P_H = [A^{H-1}B, A^{H-2}B, \dots, B] \in \R^{n \times Hm}$ is full column rank.
        \item For any $(A,B) \in \cl{S}$, the following FIR-SLS-constraint is feasible:\\ There exist $\Phi^x[1], \dots, \Phi^x[H]\in\mathbb{R}^{n\times n}$ and $ \Phi^u[0], \dots, \Phi^u[H-1]\in\mathbb{R}^{m \times n} $ such that:
        \begin{align*}
        \Phi^x[0] = I, \quad\forall k = 0,...,H-1:\:\: \Phi^x[k+1] = A\Phi^x[k] + B\Phi^u[k], \text{ and  }\Phi^x[H]=0
        \end{align*}
        \item For any $(A,B)\in\cl{S}$, the corresponding grammians $W^u_H(A,B)$ and $W^w_H(A)$ are positive-definite and their singularvalues satisfy the inequalities:
        \begin{align*}
             &\underline{\sigma}^u \leq \sigma_{min}(W^u_H(A,B)), && \sigma_{max}(W^u_H(A,B)) \leq \overline{\sigma}^u \\
             &\underline{\sigma}^w \leq \sigma_{min}(W^w_H(A)), && \sigma_{max}(W^w_H(A)) \leq \overline{\sigma}^w 
        \end{align*}
    \end{itemize}
\end{lemma}

\noindent
We can not use $ \underline{\sigma}_u,  \overline{\sigma}_u$,  $\underline{\sigma}_w$,  $\overline{\sigma}_w$ in \Cref{lem:baseCtrl} in conjuncture of the bounds derived above to obtain
\begin{align}
    &\Gamma'_2 = \cl{O}\left(
        \alpha_H \:\kappa_{CD}\: H \:\left(\frac{\overline{\sigma}_u}{\underline{\sigma}_u}\right)^{\frac{3}{2}}\right) && \Gamma'_1 = \cl{O}\left( \alpha^2_H H^{\frac{3}{2}}\left(\frac{\overline{\sigma}_u}{\underline{\sigma}_u}\right)^{\frac{1}{2}} \right)
\end{align}
and finally
\begin{align}
    &\Gamma_A = \cl{O}\left(
        \alpha^2_H \:\kappa^2_{CD}\:\|B_1\|_2\: H^2 \:\left(\frac{\overline{\sigma}_u}{\underline{\sigma}_u}\right)^{\frac{3}{2}} \overline{\sigma}_w\right) && \Gamma_B = \cl{O}\left(
            \alpha_H \:\kappa^2_{CD}\: H^{\frac{3}{2}} \:\left(\frac{\overline{\sigma}_u}{\underline{\sigma}_u}\right)^{\frac{3}{2}} \overline{\sigma}_w^{\frac{1}{2}}\right)
\end{align}

\begin{theorem}\label{thm:mainsensitivity_global}
Let $C,D \succ 0$, and let $\cl{S}$ be a compact set of controllable systems with known FIR horizon $H$ and constants $ \underline{\sigma}_u,  \overline{\sigma}_u$,  $\underline{\sigma}_w$, $\overline{\sigma}_w$ as defined in \lemref{lem:baseCtrl}. Then there are fixed constants $\Gamma_A,\Gamma_B$, such that for any two pairs of system matrices $(A_1,B_1), (A_2,B_2) \in \cl{S}$ the corresponding $\cl{H}_2$ optimal SLS-solutions of problem $S_j$ ($j$ arbitrary), denoted $\phi^{*j}_1$ and $\phi^{*j}_2$, satisfy the following inquality: 
\begin{align}
\|\phi^{j*}_1-\phi^{j*}_2\|_2 \leq \Gamma_A\|A_1-A_2\|_F + \Gamma_B \|B_1-B_2\|_F.
\end{align}
Furthermore, $\Gamma_A$ and $\Gamma_B$ satisfy
\begin{align}
    &\Gamma_A = \cl{O}\left(
        \alpha^2_H \:\kappa^2_{CD}\:\beta\: H^2 \:\left(\frac{\overline{\sigma}_u}{\underline{\sigma}_u}\right)^{\frac{3}{2}} \overline{\sigma}_w\right) && \Gamma_B = \cl{O}\left(
            \alpha_H \:\kappa^2_{CD}\: H^{\frac{3}{2}} \:\left(\frac{\overline{\sigma}_u}{\underline{\sigma}_u}\right)^{\frac{3}{2}} \overline{\sigma}_w^{\frac{1}{2}}\right),
\end{align}
where $\beta:= \max\limits_{(A,B)\in\cl{S}}\|B\|_2$ and $\kappa_{CD}$ stands for
$$\kappa_{CD} = \frac{\max\{\sigma_{max}(C),\sigma_{max}(D)\}}{\min\{\sigma_{min}(C),\sigma_{min}(D)\}}$$
\end{theorem}


%% file: misc/dimmacros.tex
\newcommand{\R}{\mathbb{R}}
\newcommand{\cl}[1]{\mathcal{#1}}

\newcommand{\lemref}[1]{(Lem.\ref{#1})}
\newcommand{\thmref}[1]{(Thm.\ref{#1})}
\newcommand{\corref}[1]{(Coro.\ref{#1})}
\newcommand{\propref}[1]{(Prop.\ref{#1})}
\newcommand{\aspref}[1]{(Ass.\ref{#1})}
\newcommand{\secref}[1]{(Sec.\ref{#1})}
\newcommand{\conref}[1]{(Const. \ref{#1})}
\newcommand{\algoref}[1]{(Algo. \ref{#1})}
\newcommand{\figref}[1]{(Fig. \ref{#1})}
\newcommand{\defref}[1]{(Def. \ref{#1})}
\newcommand{\remref}[1]{(Rem. \ref{#1})}

%% file: sample-acmsmall.bbl

\begin{thebibliography}{93}


\ifx \showCODEN    \undefined \def \showCODEN     #1{\unskip}     \fi
\ifx \showDOI      \undefined \def \showDOI       #1{#1}\fi
\ifx \showISBNx    \undefined \def \showISBNx     #1{\unskip}     \fi
\ifx \showISBNxiii \undefined \def \showISBNxiii  #1{\unskip}     \fi
\ifx \showISSN     \undefined \def \showISSN      #1{\unskip}     \fi
\ifx \showLCCN     \undefined \def \showLCCN      #1{\unskip}     \fi
\ifx \shownote     \undefined \def \shownote      #1{#1}          \fi
\ifx \showarticletitle \undefined \def \showarticletitle #1{#1}   \fi
\ifx \showURL      \undefined \def \showURL       {\relax}        \fi
\providecommand\bibfield[2]{#2}
\providecommand\bibinfo[2]{#2}
\providecommand\natexlab[1]{#1}
\providecommand\showeprint[2][]{arXiv:#2}

\bibitem[Abbasi-Yadkori and Szepesv{\'a}ri(2011)]%
        {abbasi2011regret}
\bibfield{author}{\bibinfo{person}{Yasin Abbasi-Yadkori} {and}
  \bibinfo{person}{Csaba Szepesv{\'a}ri}.} \bibinfo{year}{2011}\natexlab{}.
\newblock \showarticletitle{Regret bounds for the adaptive control of linear
  quadratic systems}. In \bibinfo{booktitle}{\emph{Proceedings of the 24th
  Annual Conference on Learning Theory}}. JMLR Workshop and Conference
  Proceedings, \bibinfo{pages}{1--26}.
\newblock


\bibitem[Agarwal et~al\mbox{.}(2019)]%
        {agarwal2019online}
\bibfield{author}{\bibinfo{person}{Naman Agarwal}, \bibinfo{person}{Brian
  Bullins}, \bibinfo{person}{Elad Hazan}, \bibinfo{person}{Sham Kakade}, {and}
  \bibinfo{person}{Karan Singh}.} \bibinfo{year}{2019}\natexlab{}.
\newblock \showarticletitle{Online control with adversarial disturbances}. In
  \bibinfo{booktitle}{\emph{International Conference on Machine Learning}}.
  PMLR, \bibinfo{pages}{111--119}.
\newblock


\bibitem[Ak{\c{c}}ay(2004)]%
        {akccay2004size}
\bibfield{author}{\bibinfo{person}{H{\"u}seyin Ak{\c{c}}ay}.}
  \bibinfo{year}{2004}\natexlab{}.
\newblock \showarticletitle{The size of the membership-set in a probabilistic
  framework}.
\newblock \bibinfo{journal}{\emph{Automatica}} \bibinfo{volume}{40},
  \bibinfo{number}{2} (\bibinfo{year}{2004}), \bibinfo{pages}{253--260}.
\newblock


\bibitem[Alemzadeh and Mesbahi(2019)]%
        {alemzadeh2019distributed}
\bibfield{author}{\bibinfo{person}{Siavash Alemzadeh} {and}
  \bibinfo{person}{Mehran Mesbahi}.} \bibinfo{year}{2019}\natexlab{}.
\newblock \showarticletitle{Distributed q-learning for dynamically decoupled
  systems}. In \bibinfo{booktitle}{\emph{2019 American Control Conference
  (ACC)}}. IEEE, \bibinfo{pages}{772--777}.
\newblock


\bibitem[Alemzadeh et~al\mbox{.}(2021)]%
        {alemzadeh2021d3pi}
\bibfield{author}{\bibinfo{person}{Siavash Alemzadeh},
  \bibinfo{person}{Shahriar Talebi}, {and} \bibinfo{person}{Mehran Mesbahi}.}
  \bibinfo{year}{2021}\natexlab{}.
\newblock \showarticletitle{D3PI: Data-Driven Distributed Policy Iteration for
  Homogeneous Interconnected Systems}.
\newblock \bibinfo{journal}{\emph{arXiv preprint arXiv:2103.11572}}
  (\bibinfo{year}{2021}).
\newblock


\bibitem[Alonso et~al\mbox{.}(2021)]%
        {alonso2021data}
\bibfield{author}{\bibinfo{person}{Carmen~Amo Alonso}, \bibinfo{person}{Fengjun
  Yang}, {and} \bibinfo{person}{Nikolai Matni}.}
  \bibinfo{year}{2021}\natexlab{}.
\newblock \showarticletitle{Data-driven Distributed and Localized Model
  Predictive Control}.
\newblock \bibinfo{journal}{\emph{arXiv preprint arXiv:2112.12229}}
  (\bibinfo{year}{2021}).
\newblock


\bibitem[Anderson et~al\mbox{.}(2019)]%
        {anderson2019system}
\bibfield{author}{\bibinfo{person}{James Anderson}, \bibinfo{person}{John~C
  Doyle}, \bibinfo{person}{Steven~H Low}, {and} \bibinfo{person}{Nikolai
  Matni}.} \bibinfo{year}{2019}\natexlab{}.
\newblock \showarticletitle{System level synthesis}.
\newblock \bibinfo{journal}{\emph{Annual Reviews in Control}}
  \bibinfo{volume}{47} (\bibinfo{year}{2019}), \bibinfo{pages}{364--393}.
\newblock


\bibitem[Anderson and Matni(2017)]%
        {anderson2017structured}
\bibfield{author}{\bibinfo{person}{James Anderson} {and}
  \bibinfo{person}{Nikolai Matni}.} \bibinfo{year}{2017}\natexlab{}.
\newblock \showarticletitle{Structured state space realizations for SLS
  distributed controllers}. In \bibinfo{booktitle}{\emph{2017 55th Annual
  Allerton Conference on Communication, Control, and Computing (Allerton)}}.
  IEEE, \bibinfo{pages}{982--987}.
\newblock


\bibitem[Antoniadis et~al\mbox{.}(2016)]%
        {antoniadis2016chasing}
\bibfield{author}{\bibinfo{person}{Antonios Antoniadis}, \bibinfo{person}{Neal
  Barcelo}, \bibinfo{person}{Michael Nugent}, \bibinfo{person}{Kirk Pruhs},
  \bibinfo{person}{Kevin Schewior}, {and} \bibinfo{person}{Michele
  Scquizzato}.} \bibinfo{year}{2016}\natexlab{}.
\newblock \showarticletitle{Chasing convex bodies and functions}. In
  \bibinfo{booktitle}{\emph{LATIN 2016: Theoretical Informatics}}. Springer,
  \bibinfo{pages}{68--81}.
\newblock


\bibitem[Argue(2022)]%
        {argue2022chasing}
\bibfield{author}{\bibinfo{person}{Charles Argue}.}
  \bibinfo{year}{2022}\natexlab{}.
\newblock \emph{\bibinfo{title}{Chasing Convex Bodies and Functions}}.
\newblock \bibinfo{thesistype}{Ph.\,D. Dissertation}. \bibinfo{school}{Carnegie
  Mellon University}.
\newblock


\bibitem[Argue et~al\mbox{.}(2019)]%
        {argue2019nearly}
\bibfield{author}{\bibinfo{person}{CJ Argue}, \bibinfo{person}{S{\'e}bastien
  Bubeck}, \bibinfo{person}{Michael~B Cohen}, \bibinfo{person}{Anupam Gupta},
  {and} \bibinfo{person}{Yin~Tat Lee}.} \bibinfo{year}{2019}\natexlab{}.
\newblock \showarticletitle{A nearly-linear bound for chasing nested convex
  bodies}. In \bibinfo{booktitle}{\emph{Proceedings of the Thirtieth Annual
  ACM-SIAM Symposium on Discrete Algorithms}}. SIAM, \bibinfo{pages}{117--122}.
\newblock


\bibitem[Argue et~al\mbox{.}(2021)]%
        {argue2021chasing}
\bibfield{author}{\bibinfo{person}{CJ Argue}, \bibinfo{person}{Anupam Gupta},
  \bibinfo{person}{Ziye Tang}, {and} \bibinfo{person}{Guru Guruganesh}.}
  \bibinfo{year}{2021}\natexlab{}.
\newblock \showarticletitle{Chasing convex bodies with linear competitive
  ratio}.
\newblock \bibinfo{journal}{\emph{Journal of the ACM (JACM)}}
  \bibinfo{volume}{68}, \bibinfo{number}{5} (\bibinfo{year}{2021}),
  \bibinfo{pages}{1--10}.
\newblock


\bibitem[Aswani et~al\mbox{.}(2013)]%
        {aswani2013provably}
\bibfield{author}{\bibinfo{person}{Anil Aswani}, \bibinfo{person}{Humberto
  Gonzalez}, \bibinfo{person}{S~Shankar Sastry}, {and} \bibinfo{person}{Claire
  Tomlin}.} \bibinfo{year}{2013}\natexlab{}.
\newblock \showarticletitle{Provably safe and robust learning-based model
  predictive control}.
\newblock \bibinfo{journal}{\emph{Automatica}} \bibinfo{volume}{49},
  \bibinfo{number}{5} (\bibinfo{year}{2013}), \bibinfo{pages}{1216--1226}.
\newblock


\bibitem[Bai et~al\mbox{.}(1998)]%
        {bai1998convergence}
\bibfield{author}{\bibinfo{person}{Er-Wei Bai}, \bibinfo{person}{Hyonyong Cho},
  {and} \bibinfo{person}{Roberto Tempo}.} \bibinfo{year}{1998}\natexlab{}.
\newblock \showarticletitle{Convergence properties of the membership set}.
\newblock \bibinfo{journal}{\emph{Automatica}} \bibinfo{volume}{34},
  \bibinfo{number}{10} (\bibinfo{year}{1998}), \bibinfo{pages}{1245--1249}.
\newblock


\bibitem[Borrelli et~al\mbox{.}(2017)]%
        {borrelli2017predictive}
\bibfield{author}{\bibinfo{person}{Francesco Borrelli},
  \bibinfo{person}{Alberto Bemporad}, {and} \bibinfo{person}{Manfred Morari}.}
  \bibinfo{year}{2017}\natexlab{}.
\newblock \bibinfo{booktitle}{\emph{Predictive control for linear and hybrid
  systems}}.
\newblock \bibinfo{publisher}{Cambridge University Press}.
\newblock


\bibitem[Bu et~al\mbox{.}(2019)]%
        {bu2019lqr}
\bibfield{author}{\bibinfo{person}{Jingjing Bu}, \bibinfo{person}{Afshin
  Mesbahi}, \bibinfo{person}{Maryam Fazel}, {and} \bibinfo{person}{Mehran
  Mesbahi}.} \bibinfo{year}{2019}\natexlab{}.
\newblock \showarticletitle{LQR through the lens of first order methods:
  Discrete-time case}.
\newblock \bibinfo{journal}{\emph{arXiv preprint arXiv:1907.08921}}
  (\bibinfo{year}{2019}).
\newblock


\bibitem[Bubeck et~al\mbox{.}(2020)]%
        {bubeck2020chasing}
\bibfield{author}{\bibinfo{person}{S{\'e}bastien Bubeck},
  \bibinfo{person}{Bo'az Klartag}, \bibinfo{person}{Yin~Tat Lee},
  \bibinfo{person}{Yuanzhi Li}, {and} \bibinfo{person}{Mark Sellke}.}
  \bibinfo{year}{2020}\natexlab{}.
\newblock \showarticletitle{Chasing nested convex bodies nearly optimally}. In
  \bibinfo{booktitle}{\emph{Proceedings of the Fourteenth Annual ACM-SIAM
  Symposium on Discrete Algorithms}}. SIAM, \bibinfo{pages}{1496--1508}.
\newblock


\bibitem[Chen and Hazan(2021)]%
        {chen2021black}
\bibfield{author}{\bibinfo{person}{Xinyi Chen} {and} \bibinfo{person}{Elad
  Hazan}.} \bibinfo{year}{2021}\natexlab{}.
\newblock \showarticletitle{Black-box control for linear dynamical systems}. In
  \bibinfo{booktitle}{\emph{Conference on Learning Theory}}. PMLR,
  \bibinfo{pages}{1114--1143}.
\newblock


\bibitem[Christianson et~al\mbox{.}(2022)]%
        {christianson2022chasing}
\bibfield{author}{\bibinfo{person}{Nicolas Christianson},
  \bibinfo{person}{Tinashe Handina}, {and} \bibinfo{person}{Adam Wierman}.}
  \bibinfo{year}{2022}\natexlab{}.
\newblock \showarticletitle{Chasing convex bodies and functions with black-box
  advice}. In \bibinfo{booktitle}{\emph{Conference on Learning Theory}}. PMLR,
  \bibinfo{pages}{867--908}.
\newblock


\bibitem[Cohen et~al\mbox{.}(2019)]%
        {cohen2019learning}
\bibfield{author}{\bibinfo{person}{Alon Cohen}, \bibinfo{person}{Tomer Koren},
  {and} \bibinfo{person}{Yishay Mansour}.} \bibinfo{year}{2019}\natexlab{}.
\newblock \showarticletitle{Learning Linear-Quadratic Regulators Efficiently
  with only sqrt(T) Regret}. In \bibinfo{booktitle}{\emph{International
  Conference on Machine Learning}}. PMLR, \bibinfo{pages}{1300--1309}.
\newblock


\bibitem[Dean et~al\mbox{.}(2020a)]%
        {dean2020sample}
\bibfield{author}{\bibinfo{person}{Sarah Dean}, \bibinfo{person}{Horia Mania},
  \bibinfo{person}{Nikolai Matni}, \bibinfo{person}{Benjamin Recht}, {and}
  \bibinfo{person}{Stephen Tu}.} \bibinfo{year}{2020}\natexlab{a}.
\newblock \showarticletitle{On the sample complexity of the linear quadratic
  regulator}.
\newblock \bibinfo{journal}{\emph{Foundations of Computational Mathematics}}
  \bibinfo{volume}{20}, \bibinfo{number}{4} (\bibinfo{year}{2020}),
  \bibinfo{pages}{633--679}.
\newblock


\bibitem[Dean et~al\mbox{.}(2020b)]%
        {dean2020robust}
\bibfield{author}{\bibinfo{person}{Sarah Dean}, \bibinfo{person}{Nikolai
  Matni}, \bibinfo{person}{Benjamin Recht}, {and} \bibinfo{person}{Vickie Ye}.}
  \bibinfo{year}{2020}\natexlab{b}.
\newblock \showarticletitle{Robust guarantees for perception-based control}. In
  \bibinfo{booktitle}{\emph{Learning for Dynamics and Control}}. PMLR,
  \bibinfo{pages}{350--360}.
\newblock


\bibitem[Dean et~al\mbox{.}(2019)]%
        {dean2019safely}
\bibfield{author}{\bibinfo{person}{Sarah Dean}, \bibinfo{person}{Stephen Tu},
  \bibinfo{person}{Nikolai Matni}, {and} \bibinfo{person}{Benjamin Recht}.}
  \bibinfo{year}{2019}\natexlab{}.
\newblock \showarticletitle{Safely learning to control the constrained linear
  quadratic regulator}. In \bibinfo{booktitle}{\emph{2019 American Control
  Conference (ACC)}}. IEEE, \bibinfo{pages}{5582--5588}.
\newblock


\bibitem[Didier et~al\mbox{.}(2022)]%
        {didier2022system}
\bibfield{author}{\bibinfo{person}{Alexandre Didier}, \bibinfo{person}{Jerome
  Sieber}, {and} \bibinfo{person}{Melanie~N Zeilinger}.}
  \bibinfo{year}{2022}\natexlab{}.
\newblock \showarticletitle{A system level approach to regret optimal control}.
\newblock \bibinfo{journal}{\emph{IEEE Control Systems Letters}}
  (\bibinfo{year}{2022}).
\newblock


\bibitem[Dullerud and Paganini(2013)]%
        {dullerud2013course}
\bibfield{author}{\bibinfo{person}{Geir~E Dullerud} {and}
  \bibinfo{person}{Fernando Paganini}.} \bibinfo{year}{2013}\natexlab{}.
\newblock \bibinfo{booktitle}{\emph{A course in robust control theory: a convex
  approach}}. Vol.~\bibinfo{volume}{36}.
\newblock \bibinfo{publisher}{Springer Science \& Business Media}.
\newblock


\bibitem[Fang et~al\mbox{.}(2011)]%
        {fang2011smart}
\bibfield{author}{\bibinfo{person}{Xi Fang}, \bibinfo{person}{Satyajayant
  Misra}, \bibinfo{person}{Guoliang Xue}, {and} \bibinfo{person}{Dejun Yang}.}
  \bibinfo{year}{2011}\natexlab{}.
\newblock \showarticletitle{Smart grid—The new and improved power grid: A
  survey}.
\newblock \bibinfo{journal}{\emph{IEEE communications surveys \& tutorials}}
  \bibinfo{volume}{14}, \bibinfo{number}{4} (\bibinfo{year}{2011}),
  \bibinfo{pages}{944--980}.
\newblock


\bibitem[Faradonbeh and Modi(2022)]%
        {faradonbeh2022joint}
\bibfield{author}{\bibinfo{person}{Mohamad Kazem~Shirani Faradonbeh} {and}
  \bibinfo{person}{Aditya Modi}.} \bibinfo{year}{2022}\natexlab{}.
\newblock \showarticletitle{Joint Learning-Based Stabilization of Multiple
  Unknown Linear Systems}.
\newblock \bibinfo{journal}{\emph{arXiv preprint arXiv:2201.01387}}
  (\bibinfo{year}{2022}).
\newblock


\bibitem[Faradonbeh et~al\mbox{.}(2018)]%
        {faradonbeh2018finite}
\bibfield{author}{\bibinfo{person}{Mohamad Kazem~Shirani Faradonbeh},
  \bibinfo{person}{Ambuj Tewari}, {and} \bibinfo{person}{George Michailidis}.}
  \bibinfo{year}{2018}\natexlab{}.
\newblock \showarticletitle{Finite-time adaptive stabilization of linear
  systems}.
\newblock \bibinfo{journal}{\emph{IEEE Trans. Automat. Control}}
  \bibinfo{volume}{64}, \bibinfo{number}{8} (\bibinfo{year}{2018}),
  \bibinfo{pages}{3498--3505}.
\newblock


\bibitem[Faradonbeh et~al\mbox{.}(2020)]%
        {faradonbeh2020optimism}
\bibfield{author}{\bibinfo{person}{Mohamad Kazem~Shirani Faradonbeh},
  \bibinfo{person}{Ambuj Tewari}, {and} \bibinfo{person}{George Michailidis}.}
  \bibinfo{year}{2020}\natexlab{}.
\newblock \showarticletitle{Optimism-based adaptive regulation of
  linear-quadratic systems}.
\newblock \bibinfo{journal}{\emph{IEEE Trans. Automat. Control}}
  \bibinfo{volume}{66}, \bibinfo{number}{4} (\bibinfo{year}{2020}),
  \bibinfo{pages}{1802--1808}.
\newblock


\bibitem[Fardad and Jovanovi{\'c}(2014)]%
        {fardad2014design}
\bibfield{author}{\bibinfo{person}{Makan Fardad} {and}
  \bibinfo{person}{Mihailo~R Jovanovi{\'c}}.} \bibinfo{year}{2014}\natexlab{}.
\newblock \showarticletitle{On the design of optimal structured and sparse
  feedback gains via sequential convex programming}. In
  \bibinfo{booktitle}{\emph{2014 American Control Conference}}. IEEE,
  \bibinfo{pages}{2426--2431}.
\newblock


\bibitem[Fattahi et~al\mbox{.}(2020)]%
        {fattahi2020efficient}
\bibfield{author}{\bibinfo{person}{Salar Fattahi}, \bibinfo{person}{Nikolai
  Matni}, {and} \bibinfo{person}{Somayeh Sojoudi}.}
  \bibinfo{year}{2020}\natexlab{}.
\newblock \showarticletitle{Efficient learning of distributed linear-quadratic
  control policies}.
\newblock \bibinfo{journal}{\emph{SIAM Journal on Control and Optimization}}
  \bibinfo{volume}{58}, \bibinfo{number}{5} (\bibinfo{year}{2020}),
  \bibinfo{pages}{2927--2951}.
\newblock


\bibitem[Furieri et~al\mbox{.}(2020)]%
        {furieri2020learning}
\bibfield{author}{\bibinfo{person}{Luca Furieri}, \bibinfo{person}{Yang Zheng},
  {and} \bibinfo{person}{Maryam Kamgarpour}.} \bibinfo{year}{2020}\natexlab{}.
\newblock \showarticletitle{Learning the globally optimal distributed LQ
  regulator}. In \bibinfo{booktitle}{\emph{Learning for Dynamics and Control}}.
  PMLR, \bibinfo{pages}{287--297}.
\newblock


\bibitem[Gholami and Sun(2020)]%
        {gholami2020fast}
\bibfield{author}{\bibinfo{person}{Amin Gholami} {and} \bibinfo{person}{Xu~Andy
  Sun}.} \bibinfo{year}{2020}\natexlab{}.
\newblock \showarticletitle{A fast certificate for power system small-signal
  stability}. In \bibinfo{booktitle}{\emph{2020 59th IEEE Conference on
  Decision and Control (CDC)}}. IEEE, \bibinfo{pages}{3383--3388}.
\newblock


\bibitem[Goel and Wierman(2019)]%
        {goel2019online}
\bibfield{author}{\bibinfo{person}{Gautam Goel} {and} \bibinfo{person}{Adam
  Wierman}.} \bibinfo{year}{2019}\natexlab{}.
\newblock \showarticletitle{An online algorithm for smoothed regression and lqr
  control}. In \bibinfo{booktitle}{\emph{The 22nd International Conference on
  Artificial Intelligence and Statistics}}. PMLR, \bibinfo{pages}{2504--2513}.
\newblock


\bibitem[Han and Skelton(2003)]%
        {han2003lmi}
\bibfield{author}{\bibinfo{person}{Jeongheon Han} {and}
  \bibinfo{person}{Robert~E Skelton}.} \bibinfo{year}{2003}\natexlab{}.
\newblock \showarticletitle{An LMI optimization approach for structured linear
  controllers}. In \bibinfo{booktitle}{\emph{42nd IEEE International Conference
  on Decision and Control (IEEE Cat. No. 03CH37475)}},
  Vol.~\bibinfo{volume}{5}. IEEE, \bibinfo{pages}{5143--5148}.
\newblock


\bibitem[Hazan et~al\mbox{.}(2020)]%
        {hazan2020nonstochastic}
\bibfield{author}{\bibinfo{person}{Elad Hazan}, \bibinfo{person}{Sham Kakade},
  {and} \bibinfo{person}{Karan Singh}.} \bibinfo{year}{2020}\natexlab{}.
\newblock \showarticletitle{The nonstochastic control problem}. In
  \bibinfo{booktitle}{\emph{Algorithmic Learning Theory}}. PMLR,
  \bibinfo{pages}{408--421}.
\newblock


\bibitem[Ho and Doyle(2019)]%
        {adaptivesls}
\bibfield{author}{\bibinfo{person}{Dimitar Ho} {and} \bibinfo{person}{John~C.
  Doyle}.} \bibinfo{year}{2019}\natexlab{}.
\newblock \showarticletitle{Scalable Robust Adaptive Control from the System
  Level Perspective}. In \bibinfo{booktitle}{\emph{2019 American Control
  Conference (ACC)}}. \bibinfo{pages}{3683--3688}.
\newblock
\urldef\tempurl%
\url{https://doi.org/10.23919/ACC.2019.8814896}
\showDOI{\tempurl}


\bibitem[Ho et~al\mbox{.}(2021)]%
        {ho2021online}
\bibfield{author}{\bibinfo{person}{Dimitar Ho}, \bibinfo{person}{Hoang Le},
  \bibinfo{person}{John Doyle}, {and} \bibinfo{person}{Yisong Yue}.}
  \bibinfo{year}{2021}\natexlab{}.
\newblock \showarticletitle{Online Robust Control of Nonlinear Systems with
  Large Uncertainty}. In \bibinfo{booktitle}{\emph{International Conference on
  Artificial Intelligence and Statistics}}. PMLR, \bibinfo{pages}{3475--3483}.
\newblock


\bibitem[Ho et~al\mbox{.}(1972)]%
        {ho1972team}
\bibfield{author}{\bibinfo{person}{Yu-Chi Ho} {et~al\mbox{.}}}
  \bibinfo{year}{1972}\natexlab{}.
\newblock \showarticletitle{Team decision theory and information structures in
  optimal control problems--Part I}.
\newblock \bibinfo{journal}{\emph{IEEE Transactions on Automatic control}}
  \bibinfo{volume}{17}, \bibinfo{number}{1} (\bibinfo{year}{1972}),
  \bibinfo{pages}{15--22}.
\newblock


\bibitem[Hu et~al\mbox{.}(2022)]%
        {hu2022sample}
\bibfield{author}{\bibinfo{person}{Yang Hu}, \bibinfo{person}{Adam Wierman},
  {and} \bibinfo{person}{Guannan Qu}.} \bibinfo{year}{2022}\natexlab{}.
\newblock \showarticletitle{On the Sample Complexity of Stabilizing LTI Systems
  on a Single Trajectory}.
\newblock \bibinfo{journal}{\emph{arXiv preprint arXiv:2202.07187}}
  (\bibinfo{year}{2022}).
\newblock


\bibitem[Ibrahimi et~al\mbox{.}(2012)]%
        {ibrahimi2012efficient}
\bibfield{author}{\bibinfo{person}{Morteza Ibrahimi}, \bibinfo{person}{Adel
  Javanmard}, {and} \bibinfo{person}{Benjamin Roy}.}
  \bibinfo{year}{2012}\natexlab{}.
\newblock \showarticletitle{Efficient reinforcement learning for high
  dimensional linear quadratic systems}.
\newblock \bibinfo{journal}{\emph{Advances in Neural Information Processing
  Systems}}  \bibinfo{volume}{25} (\bibinfo{year}{2012}).
\newblock


\bibitem[Ioannou and Fidan(2006)]%
        {ioannou2006adaptive}
\bibfield{author}{\bibinfo{person}{Petros Ioannou} {and}
  \bibinfo{person}{Bari{\c{s}} Fidan}.} \bibinfo{year}{2006}\natexlab{}.
\newblock \bibinfo{booktitle}{\emph{Adaptive control tutorial}}.
\newblock \bibinfo{publisher}{SIAM}.
\newblock


\bibitem[Jiang and Wang(2001)]%
        {jiang2001input}
\bibfield{author}{\bibinfo{person}{Zhong-Ping Jiang} {and}
  \bibinfo{person}{Yuan Wang}.} \bibinfo{year}{2001}\natexlab{}.
\newblock \showarticletitle{Input-to-state stability for discrete-time
  nonlinear systems}.
\newblock \bibinfo{journal}{\emph{Automatica}} \bibinfo{volume}{37},
  \bibinfo{number}{6} (\bibinfo{year}{2001}), \bibinfo{pages}{857--869}.
\newblock


\bibitem[Jing et~al\mbox{.}(2021)]%
        {jing2021learning}
\bibfield{author}{\bibinfo{person}{Gangshan Jing}, \bibinfo{person}{He Bai},
  \bibinfo{person}{Jemin George}, \bibinfo{person}{Aranya Chakrabortty}, {and}
  \bibinfo{person}{Piyush~K Sharma}.} \bibinfo{year}{2021}\natexlab{}.
\newblock \showarticletitle{Learning Distributed Stabilizing Controllers for
  Multi-Agent Systems}.
\newblock \bibinfo{journal}{\emph{IEEE Control Systems Letters}}
  (\bibinfo{year}{2021}).
\newblock


\bibitem[Kashyap and Lessard(2019)]%
        {kashyap2019explicit}
\bibfield{author}{\bibinfo{person}{Mruganka Kashyap} {and}
  \bibinfo{person}{Laurent Lessard}.} \bibinfo{year}{2019}\natexlab{}.
\newblock \showarticletitle{Explicit agent-level optimal cooperative
  controllers for dynamically decoupled systems with output feedback}. In
  \bibinfo{booktitle}{\emph{2019 IEEE 58th Conference on Decision and Control
  (CDC)}}. IEEE, \bibinfo{pages}{8254--8259}.
\newblock


\bibitem[Lale et~al\mbox{.}(2022)]%
        {lale2022reinforcement}
\bibfield{author}{\bibinfo{person}{Sahin Lale}, \bibinfo{person}{Kamyar
  Azizzadenesheli}, \bibinfo{person}{Babak Hassibi}, {and}
  \bibinfo{person}{Animashree Anandkumar}.} \bibinfo{year}{2022}\natexlab{}.
\newblock \showarticletitle{Reinforcement learning with fast stabilization in
  linear dynamical systems}. In \bibinfo{booktitle}{\emph{International
  Conference on Artificial Intelligence and Statistics}}. PMLR,
  \bibinfo{pages}{5354--5390}.
\newblock


\bibitem[Lamperski(2020)]%
        {lamperski2020computing}
\bibfield{author}{\bibinfo{person}{Andrew Lamperski}.}
  \bibinfo{year}{2020}\natexlab{}.
\newblock \showarticletitle{Computing Stabilizing Linear Controllers via Policy
  Iteration}. In \bibinfo{booktitle}{\emph{2020 59th IEEE Conference on
  Decision and Control (CDC)}}. IEEE, \bibinfo{pages}{1902--1907}.
\newblock


\bibitem[Lamperski and Lessard(2015)]%
        {lamperski2015optimal}
\bibfield{author}{\bibinfo{person}{Andrew Lamperski} {and}
  \bibinfo{person}{Laurent Lessard}.} \bibinfo{year}{2015}\natexlab{}.
\newblock \showarticletitle{Optimal decentralized state-feedback control with
  sparsity and delays}.
\newblock \bibinfo{journal}{\emph{Automatica}}  \bibinfo{volume}{58}
  (\bibinfo{year}{2015}), \bibinfo{pages}{143--151}.
\newblock


\bibitem[Lemos and Pinto(2012)]%
        {lemos2012distributed}
\bibfield{author}{\bibinfo{person}{Joao~M Lemos} {and} \bibinfo{person}{Luis~F
  Pinto}.} \bibinfo{year}{2012}\natexlab{}.
\newblock \showarticletitle{Distributed linear-quadratic control of serially
  chained systems: application to a water delivery canal [applications of
  control]}.
\newblock \bibinfo{journal}{\emph{IEEE Control Systems Magazine}}
  \bibinfo{volume}{32}, \bibinfo{number}{6} (\bibinfo{year}{2012}),
  \bibinfo{pages}{26--38}.
\newblock


\bibitem[Li et~al\mbox{.}(2015)]%
        {li2015overview}
\bibfield{author}{\bibinfo{person}{Shengbo~Eben Li}, \bibinfo{person}{Yang
  Zheng}, \bibinfo{person}{Keqiang Li}, {and} \bibinfo{person}{Jianqiang
  Wang}.} \bibinfo{year}{2015}\natexlab{}.
\newblock \showarticletitle{An overview of vehicular platoon control under the
  four-component framework}. In \bibinfo{booktitle}{\emph{2015 IEEE Intelligent
  Vehicles Symposium (IV)}}. IEEE, \bibinfo{pages}{286--291}.
\newblock


\bibitem[Li et~al\mbox{.}(2021a)]%
        {li2021safe}
\bibfield{author}{\bibinfo{person}{Yingying Li}, \bibinfo{person}{Subhro Das},
  \bibinfo{person}{Jeff Shamma}, {and} \bibinfo{person}{Na Li}.}
  \bibinfo{year}{2021}\natexlab{a}.
\newblock \showarticletitle{Safe Adaptive Learning-based Control for
  Constrained Linear Quadratic Regulators with Regret Guarantees}.
\newblock \bibinfo{journal}{\emph{arXiv preprint arXiv:2111.00411}}
  (\bibinfo{year}{2021}).
\newblock


\bibitem[Li et~al\mbox{.}(2021b)]%
        {li2021distributed}
\bibfield{author}{\bibinfo{person}{Yingying Li}, \bibinfo{person}{Yujie Tang},
  \bibinfo{person}{Runyu Zhang}, {and} \bibinfo{person}{Na Li}.}
  \bibinfo{year}{2021}\natexlab{b}.
\newblock \showarticletitle{Distributed reinforcement learning for
  decentralized linear quadratic control: A derivative-free policy optimization
  approach}.
\newblock \bibinfo{journal}{\emph{IEEE Trans. Automat. Control}}
  (\bibinfo{year}{2021}).
\newblock


\bibitem[Lin et~al\mbox{.}(2022)]%
        {lin2022online}
\bibfield{author}{\bibinfo{person}{Yiheng Lin}, \bibinfo{person}{James Preiss},
  \bibinfo{person}{Emile Anand}, \bibinfo{person}{Yingying Li},
  \bibinfo{person}{Yisong Yue}, {and} \bibinfo{person}{Adam Wierman}.}
  \bibinfo{year}{2022}\natexlab{}.
\newblock \showarticletitle{Online Adaptive Controller Selection in
  Time-Varying Systems: No-Regret via Contractive Perturbations}.
\newblock \bibinfo{journal}{\emph{arXiv preprint arXiv:2210.12320}}
  (\bibinfo{year}{2022}).
\newblock


\bibitem[Lin et~al\mbox{.}(2020)]%
        {lin2020distributed}
\bibfield{author}{\bibinfo{person}{Yiheng Lin}, \bibinfo{person}{Guannan Qu},
  \bibinfo{person}{Longbo Huang}, {and} \bibinfo{person}{Adam Wierman}.}
  \bibinfo{year}{2020}\natexlab{}.
\newblock \showarticletitle{Distributed reinforcement learning in multi-agent
  networked systems}.
\newblock \bibinfo{journal}{\emph{arXiv}} (\bibinfo{year}{2020}).
\newblock


\bibitem[Ma and Zhao(2015)]%
        {ma2015stabilization}
\bibfield{author}{\bibinfo{person}{Dan Ma} {and} \bibinfo{person}{Jun Zhao}.}
  \bibinfo{year}{2015}\natexlab{}.
\newblock \showarticletitle{Stabilization of networked switched linear systems:
  An asynchronous switching delay system approach}.
\newblock \bibinfo{journal}{\emph{Systems \& Control Letters}}
  \bibinfo{volume}{77} (\bibinfo{year}{2015}), \bibinfo{pages}{46--54}.
\newblock


\bibitem[Matni and Chandrasekaran(2016)]%
        {matni2016regularization}
\bibfield{author}{\bibinfo{person}{Nikolai Matni} {and} \bibinfo{person}{Venkat
  Chandrasekaran}.} \bibinfo{year}{2016}\natexlab{}.
\newblock \showarticletitle{Regularization for design}.
\newblock \bibinfo{journal}{\emph{IEEE Trans. Automat. Control}}
  \bibinfo{volume}{61}, \bibinfo{number}{12} (\bibinfo{year}{2016}),
  \bibinfo{pages}{3991--4006}.
\newblock


\bibitem[Morgan et~al\mbox{.}(2014)]%
        {morgan2014model}
\bibfield{author}{\bibinfo{person}{Daniel Morgan}, \bibinfo{person}{Soon-Jo
  Chung}, {and} \bibinfo{person}{Fred~Y Hadaegh}.}
  \bibinfo{year}{2014}\natexlab{}.
\newblock \showarticletitle{Model predictive control of swarms of spacecraft
  using sequential convex programming}.
\newblock \bibinfo{journal}{\emph{Journal of Guidance, Control, and Dynamics}}
  \bibinfo{volume}{37}, \bibinfo{number}{6} (\bibinfo{year}{2014}),
  \bibinfo{pages}{1725--1740}.
\newblock


\bibitem[Mukherjee and Vu(2022)]%
        {mukherjee2022reinforcement}
\bibfield{author}{\bibinfo{person}{Sayak Mukherjee} {and}
  \bibinfo{person}{Thanh~Long Vu}.} \bibinfo{year}{2022}\natexlab{}.
\newblock \showarticletitle{Reinforcement Learning of Structured Stabilizing
  Control for Linear Systems with Unknown State Matrix}.
\newblock \bibinfo{journal}{\emph{IEEE Trans. Automat. Control}}
  (\bibinfo{year}{2022}).
\newblock


\bibitem[Perdomo et~al\mbox{.}(2021)]%
        {perdomo2021stabilizing}
\bibfield{author}{\bibinfo{person}{Juan Perdomo}, \bibinfo{person}{Jack
  Umenberger}, {and} \bibinfo{person}{Max Simchowitz}.}
  \bibinfo{year}{2021}\natexlab{}.
\newblock \showarticletitle{Stabilizing Dynamical Systems via Policy Gradient
  Methods}.
\newblock \bibinfo{journal}{\emph{Advances in Neural Information Processing
  Systems}}  \bibinfo{volume}{34} (\bibinfo{year}{2021}).
\newblock


\bibitem[Qu et~al\mbox{.}(2020a)]%
        {qu2020bscalable}
\bibfield{author}{\bibinfo{person}{Guannan Qu}, \bibinfo{person}{Yiheng Lin},
  \bibinfo{person}{Adam Wierman}, {and} \bibinfo{person}{Na Li}.}
  \bibinfo{year}{2020}\natexlab{a}.
\newblock \showarticletitle{Scalable multi-agent reinforcement learning for
  networked systems with average reward}.
\newblock \bibinfo{journal}{\emph{arXiv preprint arXiv:2006.06626}}
  (\bibinfo{year}{2020}).
\newblock


\bibitem[Qu et~al\mbox{.}(2020b)]%
        {qu2020scalable}
\bibfield{author}{\bibinfo{person}{Guannan Qu}, \bibinfo{person}{Adam Wierman},
  {and} \bibinfo{person}{Na Li}.} \bibinfo{year}{2020}\natexlab{b}.
\newblock \showarticletitle{Scalable reinforcement learning of localized
  policies for multi-agent networked systems}. In
  \bibinfo{booktitle}{\emph{Learning for Dynamics and Control}}. PMLR,
  \bibinfo{pages}{256--266}.
\newblock


\bibitem[Recht(2019)]%
        {recht2019tour}
\bibfield{author}{\bibinfo{person}{Benjamin Recht}.}
  \bibinfo{year}{2019}\natexlab{}.
\newblock \showarticletitle{A tour of reinforcement learning: The view from
  continuous control}.
\newblock \bibinfo{journal}{\emph{Annual Review of Control, Robotics, and
  Autonomous Systems}}  \bibinfo{volume}{2} (\bibinfo{year}{2019}),
  \bibinfo{pages}{253--279}.
\newblock


\bibitem[Rotkowitz(2008)]%
        {rotkowitz2008information}
\bibfield{author}{\bibinfo{person}{Michael Rotkowitz}.}
  \bibinfo{year}{2008}\natexlab{}.
\newblock \showarticletitle{On information structures, convexity, and linear
  optimality}. In \bibinfo{booktitle}{\emph{2008 47th IEEE Conference on
  Decision and Control}}. IEEE, \bibinfo{pages}{1642--1647}.
\newblock


\bibitem[Rotkowitz and Lall(2005)]%
        {rotkowitz2005characterization}
\bibfield{author}{\bibinfo{person}{Michael Rotkowitz} {and}
  \bibinfo{person}{Sanjay Lall}.} \bibinfo{year}{2005}\natexlab{}.
\newblock \showarticletitle{A characterization of convex problems in
  decentralized control}.
\newblock \bibinfo{journal}{\emph{IEEE transactions on Automatic Control}}
  \bibinfo{volume}{50}, \bibinfo{number}{12} (\bibinfo{year}{2005}),
  \bibinfo{pages}{1984--1996}.
\newblock


\bibitem[Shah and Parrilo(2013)]%
        {shah2013cal}
\bibfield{author}{\bibinfo{person}{Parikshit Shah} {and}
  \bibinfo{person}{Pablo~A Parrilo}.} \bibinfo{year}{2013}\natexlab{}.
\newblock \showarticletitle{H2-Optimal Decentralized Control Over Posets: A
  State-Space Solution for State-Feedback}.
\newblock \bibinfo{journal}{\emph{IEEE Trans. Automat. Control}}
  \bibinfo{volume}{58}, \bibinfo{number}{12} (\bibinfo{year}{2013}),
  \bibinfo{pages}{3084--3096}.
\newblock


\bibitem[Shi et~al\mbox{.}(2020)]%
        {shi2020online}
\bibfield{author}{\bibinfo{person}{Guanya Shi}, \bibinfo{person}{Yiheng Lin},
  \bibinfo{person}{Soon-Jo Chung}, \bibinfo{person}{Yisong Yue}, {and}
  \bibinfo{person}{Adam Wierman}.} \bibinfo{year}{2020}\natexlab{}.
\newblock \showarticletitle{Online optimization with memory and competitive
  control}.
\newblock \bibinfo{journal}{\emph{Advances in Neural Information Processing
  Systems}}  \bibinfo{volume}{33} (\bibinfo{year}{2020}),
  \bibinfo{pages}{20636--20647}.
\newblock


\bibitem[Shi et~al\mbox{.}(2012)]%
        {shi2012robust}
\bibfield{author}{\bibinfo{person}{Yang Shi}, \bibinfo{person}{Ji Huang}, {and}
  \bibinfo{person}{Bo Yu}.} \bibinfo{year}{2012}\natexlab{}.
\newblock \showarticletitle{Robust tracking control of networked control
  systems: application to a networked DC motor}.
\newblock \bibinfo{journal}{\emph{IEEE Transactions on Industrial Electronics}}
  \bibinfo{volume}{60}, \bibinfo{number}{12} (\bibinfo{year}{2012}),
  \bibinfo{pages}{5864--5874}.
\newblock


\bibitem[Sieber et~al\mbox{.}(2021)]%
        {sieber2021system}
\bibfield{author}{\bibinfo{person}{Jerome Sieber}, \bibinfo{person}{Samir
  Bennani}, {and} \bibinfo{person}{Melanie~N Zeilinger}.}
  \bibinfo{year}{2021}\natexlab{}.
\newblock \showarticletitle{A system level approach to tube-based model
  predictive control}.
\newblock \bibinfo{journal}{\emph{IEEE Control Systems Letters}}
  \bibinfo{volume}{6} (\bibinfo{year}{2021}), \bibinfo{pages}{776--781}.
\newblock


\bibitem[Simchowitz and Foster(2020)]%
        {simchowitz2020naive}
\bibfield{author}{\bibinfo{person}{Max Simchowitz} {and} \bibinfo{person}{Dylan
  Foster}.} \bibinfo{year}{2020}\natexlab{}.
\newblock \showarticletitle{Naive exploration is optimal for online lqr}. In
  \bibinfo{booktitle}{\emph{International Conference on Machine Learning}}.
  PMLR, \bibinfo{pages}{8937--8948}.
\newblock


\bibitem[Simchowitz et~al\mbox{.}(2018)]%
        {simchowitz2018learning}
\bibfield{author}{\bibinfo{person}{Max Simchowitz}, \bibinfo{person}{Horia
  Mania}, \bibinfo{person}{Stephen Tu}, \bibinfo{person}{Michael~I Jordan},
  {and} \bibinfo{person}{Benjamin Recht}.} \bibinfo{year}{2018}\natexlab{}.
\newblock \showarticletitle{Learning without mixing: Towards a sharp analysis
  of linear system identification}. In \bibinfo{booktitle}{\emph{Conference On
  Learning Theory}}. PMLR, \bibinfo{pages}{439--473}.
\newblock


\bibitem[Sontag(2008)]%
        {sontag2008input}
\bibfield{author}{\bibinfo{person}{Eduardo~D Sontag}.}
  \bibinfo{year}{2008}\natexlab{}.
\newblock \showarticletitle{Input to state stability: Basic concepts and
  results}.
\newblock In \bibinfo{booktitle}{\emph{Nonlinear and optimal control theory}}.
  \bibinfo{publisher}{Springer}, \bibinfo{pages}{163--220}.
\newblock


\bibitem[Sturz et~al\mbox{.}(2020)]%
        {sturz2020distributed}
\bibfield{author}{\bibinfo{person}{Yvonne~R Sturz}, \bibinfo{person}{Annika
  Eichler}, {and} \bibinfo{person}{Roy~S Smith}.}
  \bibinfo{year}{2020}\natexlab{}.
\newblock \showarticletitle{Distributed control design for heterogeneous
  interconnected systems}.
\newblock \bibinfo{journal}{\emph{IEEE Trans. Automat. Control}}
  (\bibinfo{year}{2020}).
\newblock


\bibitem[Talebi et~al\mbox{.}(2021a)]%
        {talebi2021distributed}
\bibfield{author}{\bibinfo{person}{Shahriar Talebi}, \bibinfo{person}{Siavash
  Alemzadeh}, {and} \bibinfo{person}{Mehran Mesbahi}.}
  \bibinfo{year}{2021}\natexlab{a}.
\newblock \showarticletitle{Distributed Model-Free Policy Iteration for
  Networks of Homogeneous Systems}. In \bibinfo{booktitle}{\emph{2021 60th IEEE
  Conference on Decision and Control (CDC)}}. IEEE,
  \bibinfo{pages}{6970--6975}.
\newblock


\bibitem[Talebi et~al\mbox{.}(2021b)]%
        {talebi2021regularizability}
\bibfield{author}{\bibinfo{person}{Shahriar Talebi}, \bibinfo{person}{Siavash
  Alemzadeh}, \bibinfo{person}{Niyousha Rahimi}, {and} \bibinfo{person}{Mehran
  Mesbahi}.} \bibinfo{year}{2021}\natexlab{b}.
\newblock \showarticletitle{On regularizability and its application to online
  control of unstable LTI systems}.
\newblock \bibinfo{journal}{\emph{IEEE Trans. Automat. Control}}
  (\bibinfo{year}{2021}).
\newblock


\bibitem[Treven et~al\mbox{.}(2021)]%
        {treven2021learning}
\bibfield{author}{\bibinfo{person}{Lenart Treven}, \bibinfo{person}{Sebastian
  Curi}, \bibinfo{person}{Mojm{\'\i}r Mutn{\`y}}, {and}
  \bibinfo{person}{Andreas Krause}.} \bibinfo{year}{2021}\natexlab{}.
\newblock \showarticletitle{Learning stabilizing controllers for unstable
  linear quadratic regulators from a single trajectory}. In
  \bibinfo{booktitle}{\emph{Learning for Dynamics and Control}}. PMLR,
  \bibinfo{pages}{664--676}.
\newblock


\bibitem[Tsitsiklis and Athans(1985)]%
        {tsitsiklis1985complexity}
\bibfield{author}{\bibinfo{person}{John Tsitsiklis} {and}
  \bibinfo{person}{Michael Athans}.} \bibinfo{year}{1985}\natexlab{}.
\newblock \showarticletitle{On the complexity of decentralized decision making
  and detection problems}.
\newblock \bibinfo{journal}{\emph{IEEE Trans. Automat. Control}}
  \bibinfo{volume}{30}, \bibinfo{number}{5} (\bibinfo{year}{1985}),
  \bibinfo{pages}{440--446}.
\newblock


\bibitem[Tu and Recht(2019)]%
        {tu2019gap}
\bibfield{author}{\bibinfo{person}{Stephen Tu} {and} \bibinfo{person}{Benjamin
  Recht}.} \bibinfo{year}{2019}\natexlab{}.
\newblock \showarticletitle{The gap between model-based and model-free methods
  on the linear quadratic regulator: An asymptotic viewpoint}. In
  \bibinfo{booktitle}{\emph{Conference on Learning Theory}}. PMLR,
  \bibinfo{pages}{3036--3083}.
\newblock


\bibitem[Tu(2019)]%
        {tu2019sample}
\bibfield{author}{\bibinfo{person}{Stephen~L Tu}.}
  \bibinfo{year}{2019}\natexlab{}.
\newblock \bibinfo{booktitle}{\emph{Sample complexity bounds for the linear
  quadratic regulator}}.
\newblock \bibinfo{publisher}{University of California, Berkeley}.
\newblock


\bibitem[Umenberger and Sch{\"o}n(2020)]%
        {umenberger2020optimistic}
\bibfield{author}{\bibinfo{person}{Jack Umenberger} {and}
  \bibinfo{person}{Thomas~B Sch{\"o}n}.} \bibinfo{year}{2020}\natexlab{}.
\newblock \showarticletitle{Optimistic robust linear quadratic dual control}.
  In \bibinfo{booktitle}{\emph{Learning for Dynamics and Control}}. PMLR,
  \bibinfo{pages}{550--560}.
\newblock


\bibitem[Wang and Matni(2016)]%
        {wang2016localized}
\bibfield{author}{\bibinfo{person}{Yuh-Shyang Wang} {and}
  \bibinfo{person}{Nikolai Matni}.} \bibinfo{year}{2016}\natexlab{}.
\newblock \showarticletitle{Localized LQG optimal control for large-scale
  systems}. In \bibinfo{booktitle}{\emph{2016 American Control Conference
  (ACC)}}. IEEE, \bibinfo{pages}{1954--1961}.
\newblock


\bibitem[Wang et~al\mbox{.}(2014)]%
        {wang2014localized}
\bibfield{author}{\bibinfo{person}{Yuh-Shyang Wang}, \bibinfo{person}{Nikolai
  Matni}, {and} \bibinfo{person}{John~C Doyle}.}
  \bibinfo{year}{2014}\natexlab{}.
\newblock \showarticletitle{Localized LQR optimal control}. In
  \bibinfo{booktitle}{\emph{53rd IEEE Conference on Decision and Control}}.
  IEEE, \bibinfo{pages}{1661--1668}.
\newblock


\bibitem[Wang et~al\mbox{.}(2018)]%
        {wang2018separable}
\bibfield{author}{\bibinfo{person}{Yuh-Shyang Wang}, \bibinfo{person}{Nikolai
  Matni}, {and} \bibinfo{person}{John~C Doyle}.}
  \bibinfo{year}{2018}\natexlab{}.
\newblock \showarticletitle{Separable and localized system-level synthesis for
  large-scale systems}.
\newblock \bibinfo{journal}{\emph{IEEE Trans. Automat. Control}}
  \bibinfo{volume}{63}, \bibinfo{number}{12} (\bibinfo{year}{2018}),
  \bibinfo{pages}{4234--4249}.
\newblock


\bibitem[Wang et~al\mbox{.}(2019)]%
        {wang2019system}
\bibfield{author}{\bibinfo{person}{Yuh-Shyang Wang}, \bibinfo{person}{Nikolai
  Matni}, {and} \bibinfo{person}{John~C Doyle}.}
  \bibinfo{year}{2019}\natexlab{}.
\newblock \showarticletitle{A system-level approach to controller synthesis}.
\newblock \bibinfo{journal}{\emph{IEEE Trans. Automat. Control}}
  \bibinfo{volume}{64}, \bibinfo{number}{10} (\bibinfo{year}{2019}),
  \bibinfo{pages}{4079--4093}.
\newblock


\bibitem[Wedin(1973)]%
        {wedin1973perturbation}
\bibfield{author}{\bibinfo{person}{Per-{\AA}ke Wedin}.}
  \bibinfo{year}{1973}\natexlab{}.
\newblock \showarticletitle{Perturbation theory for pseudo-inverses}.
\newblock \bibinfo{journal}{\emph{BIT Numerical Mathematics}}
  \bibinfo{volume}{13}, \bibinfo{number}{2} (\bibinfo{year}{1973}),
  \bibinfo{pages}{217--232}.
\newblock


\bibitem[Yazdanian and Mehrizi-Sani(2014)]%
        {yazdanian2014distributed}
\bibfield{author}{\bibinfo{person}{Mehrdad Yazdanian} {and}
  \bibinfo{person}{Ali Mehrizi-Sani}.} \bibinfo{year}{2014}\natexlab{}.
\newblock \showarticletitle{Distributed control techniques in microgrids}.
\newblock \bibinfo{journal}{\emph{IEEE Transactions on Smart Grid}}
  \bibinfo{volume}{5}, \bibinfo{number}{6} (\bibinfo{year}{2014}),
  \bibinfo{pages}{2901--2909}.
\newblock


\bibitem[Ye et~al\mbox{.}(2022)]%
        {ye2022regret}
\bibfield{author}{\bibinfo{person}{Lintao Ye}, \bibinfo{person}{Ming Chi},
  {and} \bibinfo{person}{Vijay Gupta}.} \bibinfo{year}{2022}\natexlab{}.
\newblock \showarticletitle{Regret Bounds for Learning Decentralized Linear
  Quadratic Regulator with Partially Nested Information Structure}.
\newblock \bibinfo{journal}{\emph{arXiv preprint arXiv:2210.08886}}
  (\bibinfo{year}{2022}).
\newblock


\bibitem[Ye et~al\mbox{.}(2021)]%
        {ye2021sample}
\bibfield{author}{\bibinfo{person}{Lintao Ye}, \bibinfo{person}{Hao Zhu}, {and}
  \bibinfo{person}{Vijay Gupta}.} \bibinfo{year}{2021}\natexlab{}.
\newblock \showarticletitle{On the Sample Complexity of Decentralized Linear
  Quadratic Regulator with Partially Nested Information Structure}.
\newblock \bibinfo{journal}{\emph{arXiv preprint arXiv:2110.07112}}
  (\bibinfo{year}{2021}).
\newblock


\bibitem[Yeh et~al\mbox{.}(2022)]%
        {yeh2022robust}
\bibfield{author}{\bibinfo{person}{Christopher Yeh}, \bibinfo{person}{Jing Yu},
  \bibinfo{person}{Yuanyuan Shi}, {and} \bibinfo{person}{Adam Wierman}.}
  \bibinfo{year}{2022}\natexlab{}.
\newblock \showarticletitle{Robust online voltage control with an unknown grid
  topology}. In \bibinfo{booktitle}{\emph{Proceedings of the Thirteenth ACM
  International Conference on Future Energy Systems}}.
  \bibinfo{pages}{240--250}.
\newblock


\bibitem[Yu et~al\mbox{.}(2021)]%
        {yu2021localized}
\bibfield{author}{\bibinfo{person}{Jing Yu}, \bibinfo{person}{Yuh-Shyang Wang},
  {and} \bibinfo{person}{James Anderson}.} \bibinfo{year}{2021}\natexlab{}.
\newblock \showarticletitle{Localized and Distributed $\mathcal{H}_2$ State
  Feedback Control}. In \bibinfo{booktitle}{\emph{2021 American Control
  Conference (ACC)}}. IEEE, \bibinfo{pages}{2732--2738}.
\newblock


\bibitem[Zhang and Zhou(2016)]%
        {zhang2016controllability}
\bibfield{author}{\bibinfo{person}{Yuan Zhang} {and} \bibinfo{person}{Tong
  Zhou}.} \bibinfo{year}{2016}\natexlab{}.
\newblock \showarticletitle{Controllability analysis for a networked dynamic
  system with autonomous subsystems}.
\newblock \bibinfo{journal}{\emph{IEEE Trans. Automat. Control}}
  \bibinfo{volume}{62}, \bibinfo{number}{7} (\bibinfo{year}{2016}),
  \bibinfo{pages}{3408--3415}.
\newblock


\bibitem[Zhao et~al\mbox{.}(2021)]%
        {zhao2021learning}
\bibfield{author}{\bibinfo{person}{Feiran Zhao}, \bibinfo{person}{Xingyun Fu},
  {and} \bibinfo{person}{Keyou You}.} \bibinfo{year}{2021}\natexlab{}.
\newblock \showarticletitle{Learning Stabilizing Controllers of Linear Systems
  via Discount Policy Gradient}.
\newblock \bibinfo{journal}{\emph{arXiv preprint arXiv:2112.09294}}
  (\bibinfo{year}{2021}).
\newblock


\bibitem[Zheng et~al\mbox{.}(2020)]%
        {zheng2020equivalence}
\bibfield{author}{\bibinfo{person}{Yang Zheng}, \bibinfo{person}{Luca Furieri},
  \bibinfo{person}{Antonis Papachristodoulou}, \bibinfo{person}{Na Li}, {and}
  \bibinfo{person}{Maryam Kamgarpour}.} \bibinfo{year}{2020}\natexlab{}.
\newblock \showarticletitle{On the Equivalence of Youla, System-Level, and
  Input--Output Parameterizations}.
\newblock \bibinfo{journal}{\emph{IEEE Trans. Automat. Control}}
  \bibinfo{volume}{66}, \bibinfo{number}{1} (\bibinfo{year}{2020}),
  \bibinfo{pages}{413--420}.
\newblock


\bibitem[Zheng et~al\mbox{.}(2017)]%
        {zheng2017scalable}
\bibfield{author}{\bibinfo{person}{Yang Zheng}, \bibinfo{person}{Richard~P
  Mason}, {and} \bibinfo{person}{Antonis Papachristodoulou}.}
  \bibinfo{year}{2017}\natexlab{}.
\newblock \showarticletitle{Scalable design of structured controllers using
  chordal decomposition}.
\newblock \bibinfo{journal}{\emph{IEEE Trans. Automat. Control}}
  \bibinfo{volume}{63}, \bibinfo{number}{3} (\bibinfo{year}{2017}),
  \bibinfo{pages}{752--767}.
\newblock


\end{thebibliography}
